%% file: main.tex
\documentclass[reqno,centertags,12pt]{amsart}
\usepackage[utf8]{inputenc}
\usepackage[T1]{fontenc}
\usepackage[american]{babel}
\usepackage[babel=true,expansion=alltext,protrusion=alltext-nott,final]{microtype}
\usepackage{times}
\usepackage{dsfont}
\usepackage{mathtools}

\usepackage{tikz}
\usetikzlibrary{arrows.meta,calc}
\usepackage{pgfplots}
\pgfplotsset{compat=1.17}

\usepackage[svgnames]{xcolor}
\usepackage[breaklinks=true,pagebackref]{hyperref}

\hypersetup{pdftitle={PaperTemplate}, pdfauthor={M. R. Schulz}, pdfsubject={35P15; 35J10, 81Q10}, pdfkeywords={}}

\usepackage[utf8]{inputenc}
\usepackage[sb]{libertine}
\usepackage[varqu,varl]{inconsolata}
\usepackage[libertine]{newtxmath}
\useosf
\usepackage[T1]{fontenc}
\usepackage{textcomp}

\usepackage{amsthm}
\usepackage{enumitem,multirow}

\usepackage[margin={3cm},marginparwidth={2.5cm}]{geometry}

\newcommand{\abs}[1]{\left| #1 \right|}
\newcommand{\scp}[2]{\left\langle #1, #2 \right\rangle}

\newcommand{\norm}[1]{\left\Vert #1 \right\Vert}
\newcommand{\N}{{\mathbb{N}}}
\newcommand{\R}{{\mathbb{R}}}

\newcommand{\loc}{\text{\rm{loc}}}

\allowdisplaybreaks
\DeclareMathOperator{\supp}{supp}

\DeclareMathOperator {\ess}{ess}
\DeclareMathOperator {\disc}{disc}

\newtheorem{theorem}{Theorem}[section]
\newtheorem{proposition}[theorem]{Proposition}
\newtheorem{lemma}[theorem]{Lemma}
\newtheorem{corollary}[theorem]{Corollary}

\theoremstyle{definition}
\newtheorem{definition}[theorem]{Definition}

\newtheorem{remark}[theorem]{Remark}

\numberwithin{theorem}{section}
\numberwithin{equation}{section}
\newcounter{smalllist}

\usepackage{color}

\thanks{\copyright 2026 by the authors. Faithful reproduction of this article,
       in its entirety, by any means is permitted for non-commercial purposes}
\date{\today}

\begin{document}

\title[Four Particles in Dimension Two]{On the Efimov Effect for Four Particles in Dimension Two}

\author{Jonathan Rau}
\address{Department of Mathematics, Institute for Analysis, Karlsruhe Institute of Technology, 76131 Karlsruhe, Germany}
 \email{jonathan.rau@kit.edu}
\author{Marvin R. Schulz}
\address{Department of Mathematical Sciences, University of Copenhagen, Universitetsparken 5, 2100 Copenhagen, Denmark}
%}
\email{masc@math.ku.dk}

\begin{abstract}
We prove that the Schrödinger operator describing four particles in two dimensions, interacting solely through short--range three--body forces, can possess infinitely many bound states. This holds under the assumption that each three--body subsystem has a virtual level at zero energy. Our result establishes an analog of the Efimov effect for such four--particle systems in two dimensions.
\end{abstract}
%\MSC[2010]{81Q05 (primary);  35Q40, 81V45 (secondary)}
%\end{frontmatter}    

\maketitle
{\hypersetup{linkcolor=black}
\tableofcontents}

%%%%%%%%%%%%%%%%%%%%%%%%%%%%%%%%%%%%%%%%%%%
%Introduction
%%%%%%%%%%%%%%%%%%%%%%%%%%%%%%%%%%%%%%%%%%%
\section{Introduction}
The results in this work are twofold. We establish an Efimov type effect for four particles in dimension two. To this end, we prove a quantum tunneling effect for the double--well operator in the critical four--dimensional setting at the level of resonances.

Our analysis of quantum tunneling is fully variational and of independent interest beyond its application to the Efimov effect.
Notably, we rigorously prove an adiabatic reduction reminiscent of the Born–-Oppenheimer approximation. This approximation is frequently assumed in the study of the Efimov effect (see, e.g., \cite[Section 4]{HS:2025}, \cite{BFT:2026}), while its quantitative control is known to add substantial technical difficulties.

The Efimov effect, named after the physicist Vitaly N. Efimov, is a striking phenomenon in three--particle quantum mechanics. It describes a situation in which a three--particle system in $\R^3$ with short-range interactions can exhibit infinitely many negative eigenvalues. This occurs despite the absence of bound states in all associated two--particle subsystems, provided that at least two of them possess a zero-energy resonance at the bottom of the essential spectrum. This behavior is particularly remarkable because one--particle Schrödinger operators with short-range potentials usually only have a finite discrete spectrum. Another fundamental feature of the Efimov effect is its universality, meaning the effect in the same way independently of the microscopic details of the involved potentials. In particular, let $N(z)$ be the number of eigenvalues below the threshold $z<0$. Then there exists a constant $\mathcal{A}_0>0$, depending only on the masses of the particles, such that
\begin{equation}
    \lim_{z \to 0^-} \frac{N(z)}{|\log \abs{z}|} = \mathcal{A}_0.\label{sec1:eq:universality efimov}
\end{equation}

The effect was first discovered in 1970 by Efimov ~\cite{E:1970}. A rigorous mathematical proof was first established in 1974 by Yafaev  ~\cite{J:1974}  using so-called symmetrized Faddeev equations. Moreover, it was shown in~\cite{Y:1975} and \cite{Z:1974} that if at most one of the two--particle subsystems has a virtual level, then the corresponding Schrödinger has at most finitely many bound states. Variational proofs of the Efimov effect were established by Ovchinnikov and Sigal \cite{OS:1979} and subsequently improved by Tamura \cite{T:1991}. The universality relation \eqref{sec1:eq:universality efimov} was proven in 1993 by A.~Sobolev \cite{S:1993} (In particular see \cite[Theorem 3.3]{S:1993} for the explicit constant $\mathcal{A}_0$ in the case of equal masses). For reference to significant physical and mathematical findings until the end of the 1990s see for example, \cite{TVS:1993}, \cite{KS:1979},  \cite{P:1995}, \cite{P:1996}, \cite{VZ:1982}, \cite{VZ:1984} and \cite{VZ:1992}.

The mechanism of the effect can be understood as the binding of particles through a conspiracy of potential wells. This mechanism can only occur in spatial dimensions $d\leq 4$, since for $d\geq 5$ the Schrödinger operator does not admit zero-energy resonances but only bound states in $L^2(\R^d)$. Klaus and Simon proved in \cite{KS:1979} the asymptotic behavior for the ground state energy of the double--well operator at large separation $L\gg1$ of potential wells. In particular, they showed for that ground state energy $E(L)\sim -\mu_0^2 L^{-2}$, where $\mu_0\in(0,1)$ is the unique solution of the equation $e^{-\mu_0}=\mu_0$. It is remarkable that this constant is independent of the microscopic properties of the potential. This bound on the ground state energy of the double--well operator plays a decisive role in Tamura's variational proof of the Efimov effect \cite{T:1991}. For the critical dimension $d=4$, Pinchover \cite{P:1996} proved that $-L^{-2} \lesssim E(L) \lesssim -L^{-2}\log(L)^{-1}$. Recently, one of the authors showed in \cite{thesis:S:2025} that the lower bound can be improved to match the logarithmic correction in the upper bound. This allowed us to identify the setting of four particles in spatial dimension two as a configuration that is able to show an Efimov type effect under the additional assumptions on the interacting potentials provided below.

Due to technological improvements, resonant quantum systems became accessible to experimental investigation through so-called Feshbach resonances \cite{TVS:1993},\cite{C++:1998} and \cite{I++:1998}. In 2002, the Efimov effect was experimentally observed for the first time in an ultracold caesium gas; the results were later published in 2006~\cite{kraemer2006evidence}. Further progress in experiments with ultra--cold gases did lead to various experiments for configurations beyond three particles in three--dimensionals space, see e.g. \cite{Z:2009}, \cite{G:2009}, \cite{G:Li7:2009}, \cite{P:2009}, \cite{B:2009} and \cite{X:2020}). For a mathematical treatment of such cases see \cite{SZ:2025} and \cite{HSV:2025}.

For a broader reference from physics literature we refer to the review by Naidon and Endo \cite{NE:2017}. Motivated by the recent experimental developments, the Efimov effect once again became the focus of greater interest among physicists and mathematicians. 

Following Amado and Greenwood \cite{AG:1973} it was predicted that there will be no Efimov effect for more than three particles in dimension three. In 2013, Gridnev proved rigorously in \cite{G:2013} the absence of the Efimov effect for four bosons in three dimensions. This was subsequently strengthened in \cite{unpub:BBV:2020}, where the non-occurrence of the Efimov effect for $N\geq 4$ bosons in spatial dimensions $d\geq 3$ was established. Considering so--called spinless fermions (one can think of highly polarized fermions) Gridnev showed in \cite{G:2014} an Efimove--type effect (Super--Efimov effect) for three of these particles in dimension three. Later, in \cite{unpub:BBV:2020} it was proven that for these spinless fermions in dimension one and two a similar Efimov type effect does not exist for particle number $N\geq 4$.

In \cite{BBV:2022}, the absence of the Efimov effect for $N$ bosons interacting via $(N-1)$-body interactions was established for $N\geq 4$ in spatial dimension $d=1$ and for $N\geq 5$ if $d=2$. For the case $N = 4$ and $d = 2$, the physicist Nishida predicted that an Efimov type effect may occur if the particles interact exclusively through three--body potentials \cite{N:2017}. Apparently, such configurations of particles can be produced in experiments; see, for instance,~\cite{P:2014}. 

In the present work, we provide a rigorous proof that an Efimov type effect indeed can emerge in this setting of four bosons in two dimensions interacting solely via three-body interactions that are short--range. Thereby we prove the effect predicted in physics literature. This also shows that the prediction by Amado and Greenwood does not extend to spatial dimension $d=2$. The mechanism discussed here relies crucially on the absence of genuine two-body interactions. We do not address the effect of additional two-body forces, which is expected to fundamentally alter the low-energy behavior.

While our approach is conceptually related to the variational framework of Tamura \cite{T:1991}, substantial new arguments are required in the present setting. In particular, we extend the ideas of Pinchover \cite{P:1996} by constructing a suitable test function for the associated double-well operator, which satisfies the upper bound $E(R)\lesssim -L^{-2}\log(L)^{-1}$ for the ground state energy $E(R)$. This allows us to establish the Efimov effect for four particles in two dimensions.

In the construction of the test function, we study the tunneling effect of a particle with resonances rather than bound states. Near the potential wells, the test function is chosen to coincide with the corresponding resonance function, while away from the wells it is taken to be a solution of $-\Delta + \mu^2$ for a suitable $\mu > 0$. 

\subsection{Structure of the Paper}
The paper is organized as follows. In Section~\ref{sec1:prelim} we introduce the Schrödinger operators and suitable center of mass coordinates. In Section~\ref{sec2:main theorem} we state the main theorem and its proof, which relies on the two Lemmas \ref{sec2:lem:proto_efimov} and \ref{sec2:lem:adiabatic}. These two lemmas are proved in Section~\ref{sec4:adiabatic_approx} and Section~\ref{sec3:new:Tunneling}. We defer some technical details of the proof of Lemma~\ref{sec2:lem:proto_efimov} to Section~\ref{sec:aux_lemmatas}. Appendix~\ref{app:Bessel} provides some details on modified Bessel functions of the second kind, including their series expansions near zero. In Appendix~\ref{app:integrals} we evaluate several integrals involving modified Bessel functions which appear in the previous sections. In the Appendix~\ref{app:D} we discuss the decay of the resonance function in dimension four and as Corollary~\ref{cor:lambda_crit} we present an explicit example of a potential satisfying all imposed conditions.\\

\textbf{Acknowledgments:} 
The authors wish to thank Dirk Hundertmark, Semjon Vugalter and Andreas Bitter for insightful discussions, which greatly contributed to the development and refinement of this manuscript. M.~R.~Schulz  also acknowledges support by the ERC Advanced Grant MathBEC - 101095820 and by the VILLUM Foundation grant no. 10059.

%%%%%%%%%%%%%%%%%%%%%%%%%%%%%%%%%%%%%%%%%%%
%Prelimanries 
%%%%%%%%%%%%%%%%%%%%%%%%%%%%%%%%%%%%%%%%%%%
\section{Preliminaries}\label{sec1:prelim}
We study the Schrödinger operator of four identical bosonic particles in dimension two of identical mass $m>0$ at position $x_1,x_2,x_3,x_4 \in \R^2$. We denote by $I=\{ (123),(124),(134),(234)\}$ the three--particle subsystems. We study the case of short-range potentials defined as follows:
\begin{definition}[Short-Range Potentials]\label{sec2:def:short-range}
    We call a (real--valued) potential $V \in L^{d/2}(\R^d)$ short--range if there exist $R,\delta>0$ such that
    \begin{equation*}
        \abs{V(x)} \leq \abs{x}^{-2-2\delta} \quad \text{ for } \abs{x}>R.
    \end{equation*}
\end{definition}
We assume that there is no background field meaning that the system is invariant under mutual translation and rotations of the particles. Consequently, we only consider potentials depending on relative distances $x_i-x_j$ for $i\neq j$ and $i,j \in \{1,2,3,4\}$.

The four--particle Schrödinger operator then is
\begin{equation} \label{sec2:eq:4part}
   \mathcal{H} = \left( \sum_{k=1}^4 \frac{P_{x_k}^2 }{2m}\right)+ \sum_{\alpha \in I} V_\alpha, \quad P_{x} = -i\nabla_{x_k}, \quad \text{ on } L^2((\R^2)^4).
\end{equation}
For any $\alpha \in I$ the Schrödinger operator of the corresponding three--particle subsystems is given as
\begin{equation}\label{eq:three_particle_operator}
    h_\alpha = \left(\sum_{j\in \alpha } \frac{P_{x_j}^2 }{2m}\right) + V_\alpha.
\end{equation}
\begin{remark}
    As we consider four particles in dimension two,  the three--particle operator above operates on $L^2((\R^2)^3)$ and after reduction to its center of mass motion the corresponding operator will be defined in $L^2((\R^2)^2) \simeq L^2(\R^4)$. Similarly, the four--particle operator $\mathcal{H}$ in \eqref{sec2:eq:4part} is acting on  $L^2((\R^2)^4) \simeq L^2(\R^8)$ and after passing to its center mass, this reduces to $L^2(\R^6)$ which can be decomposed into four inner degrees of freedom of a three--particle subsystem and two inner degrees of freedom corresponding to a separation of three--particle subsystems. Therefore, dimensions four and two will be relevant in the subsequent analysis.
\end{remark}
Due to the HVZ-Theorem (named after Zhislin \cite{Z:1960}, van Winter \cite{vW:1964} and Hunziker \cite{H:1966})
\begin{equation*}
  \sigma_{\ess}(H) = [\Sigma, \infty), \quad  \text{ where }  \Sigma = \min_{\alpha \in I} \inf \sigma(h_\alpha).
\end{equation*}
In the following, we consider the case $\Sigma=0$ and $h_\alpha \ge 0$, where the operator $h_\alpha$ has a virtual level at zero in the sense of Yafaev \cite{Y:1975}; see Definition~\ref{sec1:def:virtual_lvl} below. In this situation, there exists a nontrivial zero--energy solution $\varphi_0$ of $h_\alpha$, satisfying
\[
h_\alpha \varphi_0 = 0
\]
in the distributional sense, with $\varphi_0 \notin L^2(\R^4)$ but $\varphi_0 \in \dot H^1(\R^4)$. For this statement and further background and examples, we refer to \cite[Chapter~1.4]{thesis:S:2025}. 

The homogeneous Sobolev space $\dot H^1(\R^4)$ has already been introduced by Birman \cite{B:1961} and we follow the definition given in \cite[Chapter~2]{book:FLW:2022}:
\begin{definition}[Homogeneous Sobolev Spaces]\label{sec2:def:hom_SOB}
    Let
\begin{equation*}
    \begin{split}
        \dot H^1(\R^4) = \{u \in L^2(\R^4, \abs{x}^{-2} dx): \nabla u \in L^2(\R^4)  \}
    \end{split}
\end{equation*}
be the homogeneous Sobolev space equipped with the norm 
\begin{equation*}
        \norm{u}_{\dot H^1(\R^4)} \coloneqq \left( \int_{\R^4}\abs{\nabla u}^2 dx \right)^{1/2} \, .
\end{equation*}
\end{definition}
\subsection{Center of Mass Coordinates} \label{subsec:com}

We assume that the interaction $V_\alpha$ is invariant under mutual translations of the corresponding three--particle subsystem. Since this holds for all $\alpha \in I$, the four--particle Schrödinger $H$ is invariant under simultaneous translations of all particle coordinates and hence commutes with the total momentum operator. As a consequence,
\begin{equation*}
\sigma(H)=\sigma_{\mathrm{ess}}(H).
\end{equation*}

To study the internal degrees of freedom, we separate the center--of--mass motion from the relative motion. Fix $\alpha \in I$. By standard arguments, the three--particle Schrödinger decomposes as
\begin{equation*}
h_\alpha=(T_\alpha+V_\alpha)\otimes 1 + 1\otimes \frac{-1}{2M}\Delta_{X_\alpha},
\qquad M=\sum_{i\in\alpha}m_i=3m,
\end{equation*}
where $X_\alpha$ denotes the center of mass of the subsystem. The operator $T_\alpha$ acts on the internal configuration space
\begin{equation*}
R_0[\alpha]=\bigl\{(x_i)_{i\in\alpha}\in(\mathbb{R}^2)^3:\, X_\alpha=0\bigr\},
\end{equation*}
which is a four--dimensional linear subspace.

Since the center--of--mass motion is free, we restrict attention to the internal Schrödinger operator
\begin{equation*}
h_\alpha^{\mathrm{COM}}=T_\alpha+V_\alpha
\quad \text{on } L^2(R_0[\alpha]).
\end{equation*}
Introducing Jacobi coordinates for $\alpha=(ijk)$,
\begin{equation*}
q_{ij}=x_j-x_i, \qquad
\eta_{ij}=x_k-\frac{x_i+x_j}{2},
\end{equation*}
the constraint $X_\alpha=0$ together with the variables $(q_{ij},\eta_{ij})$ yields a one–to–one parametrization of $R_0[\alpha]\cong\mathbb{R}^4$.

Identifying $R_0[\alpha]$ with $\mathbb{R}^4$ by using the Jacobi coordinates $(q_{ij},\eta_{ij})$, the internal Schrödinger takes the form
\begin{equation}\label{sec1:eq:three_part_op}
h_\alpha^{\mathrm{COM}}
= P_{q_{ij}}^2 + \frac34 P_{\eta_{ij}}^2
+ V_\alpha\!\left(q_{ij},\eta_{ij}\right)
\quad \text{on } L^2(\mathbb{R}^4).
\end{equation}
Following Yafaev \cite{J:1974}, we introduce the following notion.
\begin{definition}\label{sec1:def:virtual_lvl}
Let $V\in L^2(\mathbb{R}^4)$ be real--valued and let
\[
h(\varepsilon)= -\Delta_y + (1+\varepsilon)V(y)
\quad \text{on } L^2(\mathbb{R}^4).
\]
We say that $h(0)$ has a \emph{virtual level at zero} if
\begin{equation}\label{sec1:eq:resonance}
    h(0)\geq 0
    \quad \text{and} \quad
    \inf \sigma\!\left(h(\varepsilon)\right)<0
    \quad \text{for all } \varepsilon>0 .
\end{equation}
\end{definition}

\begin{remark}\label{sec1:rem:virtual_lvl}
The definition above does not apply directly to the three--particle operator
$h_{\alpha}^{\mathrm{COM}}$ in \eqref{sec1:eq:three_part_op}, since it is defined
on a different configuration space.  
We shall nevertheless say that $h_{\alpha}^{\mathrm{COM}}$ has a virtual level
at zero if the operator obtained from $h_{\alpha}^{\mathrm{COM}}$ by the natural
rescaling of variables is unitarily equivalent to an operator of the form in the Definition \ref{sec1:def:virtual_lvl} possessing a virtual level at zero.

Condition \eqref{sec1:eq:resonance} describes a non-generic situation at the threshold, as it requires a specific tuning of the potential $V$. The Efimov effect arises precisely in this regime of zero-energy instability of the subsystem Schrödinger operator.  To demonstrate that this class of potentials is nonempty, we give in \ref{app:D} (Corollary~\ref{cor:lambda_crit}) an explicit compactly supported potential in dimension $d\geq 3$ that exhibits a virtual level at zero energy.
\end{remark}

We proceed analogously for the four--particle Schrödinger $H$, which is invariant
under joint translations of all particles. As a consequence, $H$ decomposes as
$$
H
=
\Bigl[T + \sum_{\alpha\in I} V_\alpha \Bigr]\otimes 1
\;+\;
1\otimes \frac{1}{2\mathcal M}P^2_\vartheta,
\qquad
\mathcal M=\sum_{i=1}^4 m_i=4m .
$$
Here $T$ denotes the Laplace--Beltrami operator on $L^2(\mathcal{R}_0)$, where
$$
\mathcal{R}_0
\coloneqq
\Bigl\{(x_1,\dots,x_4)\in(\mathbb R^2)^4 \,\Big|\, \sum_{j=1}^4 x_j=0\Bigr\}
\subset \mathbb R^8 .
$$
We define the four--particle center--of--mass Schrödinger
$$
H_{\mathrm{COM}}
=
T+\sum_{\alpha\in I} V_\alpha
\quad \text{on } L^2(\mathcal{R}_0).
$$

Let $J\coloneqq\{(12),(13),(14),(23),(24),(34)\}$ denote the set of all two--particle subsystems. Fix $\gamma,\delta\in J$ such that $\gamma=(ij)$ and $\delta=(kl)$ are
disjoint, i.e.\ $\{i,j,k,l\}=\{1,2,3,4\}$. We introduce the coordinates
\begin{equation}\label{sec1:eq:jacobi_coords_2plus2}
q_\gamma=x_j-x_i,\qquad
q_\delta=x_l-x_k,\qquad
\xi_{\gamma,\delta}=\frac{x_k+x_l}{2}-\frac{x_i+x_j}{2},
\qquad
\vartheta=\frac{1}{4}\sum_{r=1}^4 x_r .
\end{equation}
Let $P_r=-i\nabla_r$ for $r\in\mathbb R^2$. Then
$$
H \simeq H_{\mathrm{COM}}\otimes (8m)^{-1}P_\vartheta^2 ,
$$
where, by a slight abuse of notation,
$$
H_{\mathrm{COM}}
=
m^{-1}P_{q_\gamma}^2
+m^{-1}P_{q_\delta}^2
+(2m)^{-1}P_{\xi_{\gamma,\delta}}^2
+\sum_{\alpha\in I} V_\alpha
\quad \text{on } L^2(\mathbb R^6).
$$
and define
$$
\eta_\gamma^\pm=\xi_{\gamma,\delta}\pm q_\delta/2,
\qquad
\eta_\delta^\pm=-\xi_{\gamma,\delta}\pm q_\gamma/2 .
$$
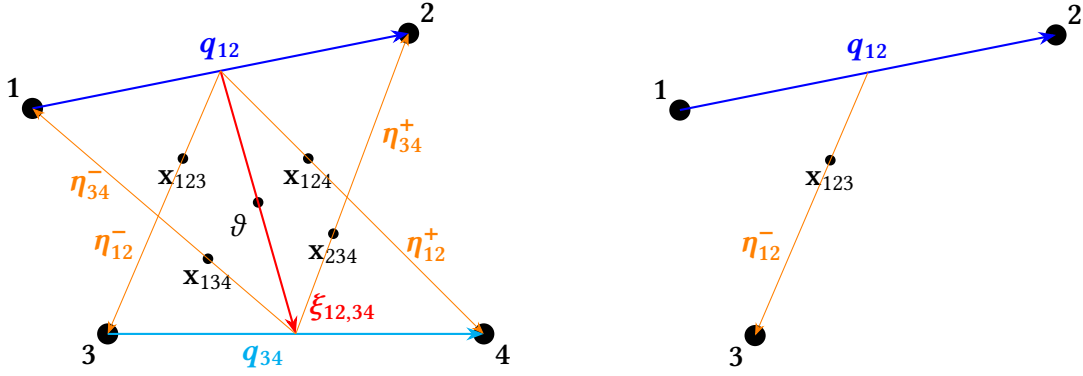
\begin{figure}[ht!]
\centering
\begin{minipage}{0.48\textwidth}
    \centering
    \input{figure_001}
\end{minipage}
\hfill
\begin{minipage}{0.48\textwidth}
    \centering
    \input{figure_002}
\end{minipage}
\caption{The choice of coordinates for the inner degrees of freedom and their connection to usual Jacobian coordinates of the subsystem $(123)$.} \label{fig:001}
\end{figure}
The pairs $(q_\gamma,\eta_\gamma^\pm)$ and $(q_\delta,\eta_\delta^\pm)$ are the same as the usual Jacobi coordinates defined earlier for the four distinct three--particle subsystems $\alpha\in I$. See Figure \ref{fig:001}. They are related to the original particle positions
$x_1,\dots,x_4\in\mathbb R^2$ and satisfy, in the corresponding three--particle
center--of--mass frame $X_\alpha=\sum_{s\in\alpha}x_s=0$, the relation
$$ 
\frac12|q|^2+\frac23|\eta|^2
=
\sum_{j\in\alpha}|x_j|^2 .
$$

We assume that the three--particle interaction potentials $V_\alpha$ are short--range and identical, i.e.\ $V_\alpha=V$. Fixing the coordinates $(q_\gamma,q_\delta,\xi_{\gamma,\delta})$, the operator
$H_{\mathrm{COM}}$ is unitarily equivalent to the operator on $L^2((\R^2)^3)$
given by
\begin{equation}\label{sec1:eq:4_part_op}
\begin{aligned}
H_{\mathrm{COM}}
={}&
m^{-1}P_{q_\gamma}^2
+m^{-1}P_{q_\delta}^2
+(2m)^{-1}P_{\xi_{\gamma,\delta}}^2
\\
&+V\!\left(q_\gamma,\xi_{\gamma,\delta}+q_\delta/2\right)
+V\!\left(q_\gamma,\xi_{\gamma,\delta}-q_\delta/2\right)
\\
&+V\!\left(q_\delta,\xi_{\gamma,\delta}+q_\gamma/2\right)
+V\!\left(q_\delta,\xi_{\gamma,\delta}-q_\gamma/2\right).
\end{aligned}
\end{equation}
In the following we always assume that $V \in L^2((\R^2)^2)$ is short--range (see Definition \ref{sec2:def:short-range}) with fixed parameters $R,\delta>0$.
\begin{remark}
It is sufficient to prove the infinitude of the discrete spectrum of the operator above, since the spectrum of $H_{\mathrm{COM}}$ is independent of the particular choice of coordinates.
\end{remark}
%%%%%%%%%%%%%%%%%%%%%%%%%%%%%%%%%%%%%%%%%%%%%%%%%%%%%%%%%%%%
%%%%%%%%%%% SECTION MAIN THEOREM
%%%%%%%%%%%%%%%%%%%%%%%%%%%%%%%%%%%%%%%%%%%%%%%%%%%%%%%%%%%
\section{Main Theorem}\label{sec2:main theorem}
\begin{theorem}\label{sec2:thrm:main}
Assume the setting of Section~\ref{sec1:prelim}, with identical particles and short--range three--body interactions in $L^2((\R^2)^2)$.
If the internal three--particle Schrödinger $h^{\mathrm{COM}}_\alpha$ has a
virtual level at zero, then the four--particle internal Schrödinger
$H_{\mathrm{COM}}$ has infinitely many negative eigenvalues
accumulating at zero.
\end{theorem}
\begin{remark}
For identical particles the operators $h_{\alpha}^{\mathrm{COM}}$, $\alpha\in I$, are unitarily equivalent. Hence, it suffices to assume the existence of a virtual level at zero for a single three--particle subsystem. Theorem~\ref{sec2:thrm:main} establishes an Efimov type effect for four identical bosons in two spatial dimensions interacting exclusively via three--body forces.
\end{remark}
The proof of Theorem~\ref{sec2:thrm:main} is carried out in three steps.

In \textbf{Step~1} we show that, due to a \emph{conspiracy of potential wells} in the sense of~\cite{KS:1979}, an associated double--well operator in $\mathbb{R}^4$ exhibits a tunneling effect at the level of resonance functions.

In \textbf{Step~2} we prove that this tunneling mechanism induces an effective long--range attractive potential in the four--particle internal Schrödinger $H_{\mathrm{COM}}$. The corresponding statements are formulated as Lemmas~\ref{sec2:lem:proto_efimov} and~\ref{sec2:lem:adiabatic} below.
Assuming these two lemmas, \textbf{Step~3} completes the argument and yields Theorem~\ref{sec2:thrm:main}.

We begin with the first step, which consists of establishing the following.
\begin{lemma}[Quantum Tunneling at Threshold in Dimension Four]\label{sec2:lem:proto_efimov}
Let $V:\mathbb R^4\to\mathbb R$ be as in Theorem~\ref{sec2:thrm:main}. For
$\ell \in\mathbb R^4$ with $L=|\ell|$ define
\begin{equation}\label{sec3:eq:dw_op}
H[\ell]=P_x^2+V(x-\ell/2)+V(x+\ell/2)\quad \text{on }L^2(\mathbb R^4,dx).
\end{equation}
For any given $\varepsilon>0$ there exists $L_0=L_0(\varepsilon)>0$ such that there exists a family of functions $\{ \widetilde\Phi(\ell,\cdot ) \}_{\ell\in \R^4} \subset H^1(\mathbb R^4)$ with the following properties:
\begin{equation}\label{sec2:eq:tunneling_gs}
\quad \|\widetilde\Phi(\ell,\cdot )\|_{L^2(\mathbb R^4)}=1 \, \text{ and } \, \langle \widetilde\Phi(\ell,\cdot ), H[\ell]\widetilde\Phi(\ell,\cdot )\rangle
\le -(2-\varepsilon)\,L^{-2}\log(L)^{-1},
\qquad L>L_0.
\end{equation}
Moreover, for the explicit choice $\ell=(0,0,r)$
with $r\in\mathbb R^2$ and $L=|r|>L_0$, the map \mbox{
$r \mapsto \widetilde \Phi((0,0,r), \cdot)$ } may be chosen in $H^1(\mathbb R^2;L^2(\mathbb R^4))$ and satisfies
\begin{equation}\label{sec3:eq:for_born_op}
\left\langle \widetilde\Phi(\ell,\cdot ),
\left(P_r^2-\frac14 P_{(x_3,x_4)}^2\right)\widetilde\Phi(\ell,\cdot )\right\rangle_{L^2(dx)}
\le \left( \frac{2}{3} +\varepsilon \right)\,L^{-2}\log(L)^{-1},
\qquad L>L_0.
\end{equation}
\end{lemma}
\begin{remark}
The estimate \eqref{sec3:eq:for_born_op} quantifies the variation of the energy
expectation under shifts of the two potential wells and is required to control
the corresponding correction terms in the subsequent analysis, where an
adiabatic reduction with respect to the parameter $r$ is performed.
\end{remark}
\begin{proof}
    The proof of this Lemma is the main difficulty, and present it in various steps in  Section~\ref{sec3:new:Tunneling}.
\end{proof}
\begin{remark}
In the following we will always denote $L = |\ell|$. From Lemma~\ref{sec2:lem:proto_efimov} we directly obtain
\begin{equation*}
\inf_{\psi\in H^1(\mathbb R^4)}
\frac{\langle \psi, H[\ell]\psi\rangle}{\|\psi\|^2}
\;\lesssim\;
L^{-2}\log(L)^{-1}.
\end{equation*}
This bound on the ground--state energy of the double--well operator
\eqref{sec3:eq:dw_op} agrees with estimates previously obtained in
\cite{P:1996} and \cite[Chapter~6]{thesis:S:2025}, but here it is derived from an
explicit test function and yields an explicit bound on the involved constant.
Moreover, it provides the input needed to establish the estimate
\eqref{sec3:eq:for_born_op}.
\end{remark}
The second step of the proof will be showing that
\begin{lemma}[Effective Long--Range Behavior]\label{sec2:lem:adiabatic}
Let $H_{\mathrm{COM}}$ be the four--particle internal Schrödinger on
$L^2(\mathbb R^6)$ defined in \eqref{sec1:eq:4_part_op}.
Then there exist $\Phi\in H^1(\mathbb R^6)$, $T>0$, and $c>0$ with
$\|\Phi\|_{L^2(\mathbb R^6)}=1$ such that for every
$u\in H^1(\mathbb R^2)$ with $\|u\|_{L^2(\mathbb R^2)}=1$ and
\[
\supp(u)\subset\{r\in\mathbb R^2:\ |r|>T\},
\]
one has
\begin{equation}\label{sec2:eq:effective_potential}
\langle u\Phi, H_{\mathrm{COM}}(u\Phi)\rangle_{L^2(\mathbb R^6)}
\le
\langle u,\bigl(P^2_r - c\,|r|^{-2}\log(|r|)^{-1}\bigr)u\rangle_{L^2(\mathbb R^2)} .
\end{equation}
\end{lemma}
\begin{proof}
In the proof we use the tunneling estimates of
Lemma~\ref{sec2:lem:proto_efimov} and perform the adiabatic reduction with remaining parameter $\ell(r)$ to find an effective long--range attractive potential. We give the proof in Section~\ref{sec4:adiabatic_approx}.
\end{proof}
Let us assume Lemma~\ref{sec2:lem:adiabatic} for the moment. We now complete the
argument by showing how it implies Theorem~\ref{sec2:thrm:main}.
\begin{proof}[Proof of Theorem~\ref{sec2:thrm:main}]
We prove that $H_{\mathrm{COM}}$ has infinitely many discrete eigenvalues by
constructing linearly independent functions $(\psi_k)_{k\in\N}\subset H^1(\R^6)$
with $\langle \psi_k, H_{\mathrm{COM}}\psi_k\rangle < 0$ for all sufficiently large $k\in\N$.
By the min--max principle this implies that $H_{\mathrm{COM}}$ has infinitely many
negative eigenvalues.

By Lemma~\ref{sec2:lem:adiabatic} there exist a normalized $\Phi\in H^1(\R^6)$ and
constants $T,\varepsilon>0$ such that for any sequence
$(u_k)_{k\in\N}\subset H^1(\R^2)$ with
$\|u_k\|_{L^2(\R^2)}=1$ and $\supp(u_k)\subset\{r\in\R^2:\,|r|>T\}$,
the functions $\psi_k = u_k\Phi$ satisfy
\begin{equation}\label{sec2:eq:main_thrm_prove_adiabiatic}
\langle \psi_k, H_{\mathrm{COM}}\psi_k\rangle
\le 
\langle u_k,(P^2_r - \varepsilon |r|^{-2}\log(|r|)^{-1})u_k\rangle_{L^2(\R^2)} \, .
\end{equation}
It therefore suffices to construct $u_k$ such that the right--hand side of
\eqref{sec2:eq:main_thrm_prove_adiabiatic} is negative.

Assume $u$ is radial, meaning $u$ only depends on $|r|$ and supported in $\{r>T\}$, then introducing spherical coordinates with $\rho=|r|$ and by abuse of notation we denote the function $u:\R_+ \to \R$ by the same letter such that
\begin{equation*}
\langle u,(P^2_r - \varepsilon |r|^{-2}\log(|r|)^{-1})u\rangle_{L^2(\R^2)}
= 2\pi \int_T^\infty \bigl(|\partial_\rho u|^2 - \varepsilon \rho^{-2}\log(\rho)^{-1}|u|^2\bigr)\, \rho\,d\rho.
\end{equation*}
Setting $t=\log (\rho)$ and $f(t)=u(e^t)$ yields
\begin{equation} \label{sec3:eq:subst}
   \langle u,(P^2_r - \varepsilon |r|^{-2}\log(|r|)^{-1})u\rangle_{L^2(\R^2)}
=2\pi \int_{\log(T)}^\infty \bigl(| f'|^2 - \varepsilon t^{-1}|f|^2\bigr)\,dt . 
\end{equation}

For $n\in\N$ with $n\geq (T)^{1/2}$ define the interval $I_n = [n^2,\, n^2 + K n]$ where $K >0$ will be chosen below. Let $f_n\in H^1_0(I_n)$ be the first Dirichlet eigenfunction of $-\partial_t^2$
on $I_n$, normalized by $\|f_n\|_{L^2(I_n)}=1$.
Then
\begin{equation} \label{sec3:eq:dirichlet}
    \int_{I_n} | \partial_t f_n|^2\,dt = \frac{\pi^2}{K^2 n^2}, \quad \int_{I_n} t^{-1}|f_n|^2\,dt \ge \frac{1}{n^2 + K n}.
\end{equation}
Inserting \eqref{sec3:eq:dirichlet} into \eqref{sec3:eq:subst} yields
\begin{equation*}
    \int_{I_n}\bigl(|\partial_t f_n|^2 - \varepsilon t^{-1}|f_n|^2\bigr)\,dt
\le 
\frac{\pi^2}{K^2 n^2} - \frac{\varepsilon}{n^2+K n}.
\end{equation*}
Choose $K >\pi\varepsilon^{-1/2}$. Then for all sufficiently large $n$ the right--hand
side is strictly negative.

Consecutive intervals $I_n$ may overlap, but since $n\mapsto n^2$ grows
quadratically, we can pick a subsequence $(n_k)_{k\in\N}$ with $n_k\to\infty$
sufficiently fast so that the intervals $I_{n_k}$ are
pairwise disjoint.  Define
\begin{equation*}
    u_k(\rho)=f_{n_k}(\log \rho ),
\qquad
\text{so that}\quad
\supp(u_k)=\{\, \rho :\, \log( \rho) \in I_{n_k}\,\}
\end{equation*}
are pairwise disjoint intervals contained in $(T,\infty)$.
For all large $k \in \N$ we have
\begin{equation*}
    \langle u_k,(P^2_r - \varepsilon |r|^{-2}\log(|r|)^{-1})u_k\rangle < 0.
\end{equation*}
Setting $\psi_k=u_k\Phi$, we obtain $\psi_k\in H^1(\mathbb R^6)$ with pairwise
disjoint supports and, by \eqref{sec2:eq:main_thrm_prove_adiabiatic},
\begin{equation*}
    \langle \psi_k, H_{\mathrm{COM}}\psi_k\rangle < 0
\qquad\text{for all sufficiently large }k \in \N \, .
\end{equation*}
Moreover, since
$\langle \psi_k, H_{\mathrm{COM}}\psi_k\rangle$ tends to $0$ as $k\to\infty$,
the corresponding negative eigenvalues accumulate at $0$.
This completes the proof of Theorem~\ref{sec2:thrm:main}.
\end{proof}
%%%%%%%%%%%%%%%%%%%%%%%%%%%%%%%%%%%%%%%%%%%%%%%%%%%%%%%%%%%%%%
%%%%%%%%%%%%%% ADIABATIC SECTION 
%%%%%%%%%%%%%%%%%%%%%%%%%%%%%%%%%%%%%%%%%%%%%%%%%%%%%%%%%%%%%%%
\section{Adiabatic Approximation of Ground States}\label{sec4:adiabatic_approx}
In this section, prove Lemma~\ref{sec2:lem:adiabatic} under the assumption of Lemma~\ref{sec2:lem:proto_efimov}. This corresponds to \textbf{Step~2} in the proof of the main Theorem~\ref{sec2:thrm:main}.

The proof is based on an adiabatic reduction in configuration space, reminiscent of the Born--Oppenheimer method, though without any smallness condition on a mass--imbalance. Working in the coordinates introduced in \eqref{sec1:eq:jacobi_coords_2plus2}, we regard the variable $q_\gamma\in\R^2$ as an external parameter. For fixed $q_\gamma$, we identify an operator acting on the remaining variables $(q_\delta,\xi_{\gamma,\delta})\in(\R^2)^2$ which admits a virtual level at zero.
This produces an effective, attractive, long--range potential in the $q_\gamma$--direction. We then compare the resulting effective expression with the full Schrödinger operator $H_{\mathrm{COM}}$ acting on $H^1((\R^3)^2)$.

\begin{proof}[Proof of Lemma~\ref{sec2:lem:adiabatic}]
In Section~\ref{sec1:prelim} we fixed Jacobi coordinates
$q_\gamma,q_\delta,\xi_{\gamma,\delta}\in\R^2$, where
$\gamma,\delta\in J$ denote two disjoint two--particle subsystems.
In these coordinates for fixed $m=1$, the center--of--mass Operator
$H_{\mathrm{COM}}$ acting on $H^1((\R^3)^2)$ takes the form
\begin{equation}\label{sec4:eq:4_part_op}
\begin{aligned}
H_{\mathrm{COM}}
={}&
P_{q_\gamma}^2
+P_{q_\delta}^2
+\frac{1}{2}P_{\xi_{\gamma,\delta}}^2
\\
&+V\!\left(q_\gamma,\xi_{\gamma,\delta}+q_\delta/2\right)
+V\!\left(q_\gamma,\xi_{\gamma,\delta}-q_\delta/2\right)
\\
&+V\!\left(q_\delta,\xi_{\gamma,\delta}+q_\gamma/2\right)
+V\!\left(q_\delta,\xi_{\gamma,\delta}-q_\gamma/2\right).
\end{aligned}
\end{equation}

Since multiplication of \eqref{sec4:eq:4_part_op} by $m$ amounts to a rescaling of
the potential, it suffices to consider the case $m=1$.

Our first step is to isolate, within $H_{\mathrm{COM}}$, an operator that
coincides with the double--well operator $H[\cdot]$ from
Lemma~\ref{sec2:lem:proto_efimov}.
By assumption, for each fixed $q_\gamma\in\R^2$, the single--well operators
\begin{equation}\label{sec3:eq:single_well}
P_{q_\delta}^2 + \tfrac34 P_{\xi_{\gamma,\delta}}^2
+ V\!\left(\tfrac12 q_\delta, \xi_{\gamma,\delta}\pm q_\gamma/2\right)
\end{equation}
admit a virtual level at zero.

We now decompose $H_{\mathrm{COM}}$ accordingly.
For fixed $(q_\delta,\xi_{\gamma,\delta})\in(\R^2)^2$, we define an operator
on $H^1(\R^2,dq_\gamma)$ by
\begin{equation*}
\begin{split}
A[q_\delta,\xi_{\gamma,\delta}]
&= P_{q_\gamma}^2
 + V\!\left(q_\gamma,\xi_{\gamma,\delta}+q_\delta/2\right)
 + V\!\left(q_\gamma, \xi_{\gamma,\delta}-q_\delta/2\right) \\
&= P_{q_\gamma}^2 + W\!\left(q_\gamma,\xi_{\gamma,\delta}+q_\delta/2,\xi_{\gamma,\delta}-q_\delta/2\right).
\end{split}
\end{equation*}
Since $V$ is short--range by assumption (see Definition \ref{sec2:def:short-range} it follows that
\begin{equation*}
    \abs{W(q_\gamma,\cdot)}\leq \abs{q_\gamma}^{-2-2\delta}, \quad \abs{q_\gamma}> R.
\end{equation*}

Conversely, for fixed $q_\gamma\in\R^2$, we define the operator
$B[q_\gamma]$ acting on $H^1((\R^2)^2,dq_\delta\,d\xi_{\gamma,\delta})$ by
\begin{equation*}
\begin{split}
B[q_\gamma]
=\, &P_{q_\delta}^2 + \tfrac34 P_{\xi_{\gamma,\delta}}^2 + V\!\left(q_\delta,\xi_{\gamma,\delta}+q_\gamma/2\right)
 + V\!\left(q_\delta,\xi_{\gamma,\delta}-q_\gamma/2\right).
\end{split}
\end{equation*}

Let $A$ and $B$ denote the corresponding operators on
\begin{equation}\label{sec4:eq:hilbert_space_6}
\mathbb{H}^1
= H^1((\R^2)^3,
dq_\gamma\,d\xi_{\gamma,\delta}\,dq_\delta),
\end{equation}
acting pointwise as $A[q_\delta,\xi_{\gamma,\delta}]$ and $B[q_\gamma]$,
respectively. This then yields
\begin{equation}\label{sec4:eq:hcom_a_b}
H_{\mathrm{COM}} = A + B - \tfrac14 P_{\xi_{\gamma,\delta}}^2.
\end{equation}
Up to scaling, the operator $B[q_\gamma]$ coincides with the double--well operator $H[\cdot]$ from Lemma~\ref{sec2:lem:proto_efimov}.

Next, we construct a product--type trial state based on the
near--threshold state of the operator $B[q_\gamma]$. This will yield an effective contribution in the $q_\gamma$--direction. The subsequent estimates then compare this effective expression with the full Schrödinger \eqref{sec4:eq:4_part_op}. To keep track of variable dependence, we introduce abbreviations for the
relevant Hilbert spaces and their inner products.

We denote by $\langle\cdot,\cdot\rangle_\circ$ the inner product of
$H^1((\R^2)^2,dq_\delta\, d\xi_{\gamma,\delta})$, and by
$\langle\cdot,\cdot\rangle_\ast$ the inner product of
$H^1(\R^2,dq_\gamma)$.
For any $f\in H^1((\R^2)^2,dq_\delta\,d\xi_{\gamma,\delta})$ we define
the dilation $D[\lambda]$ in the $\xi_{\gamma,\delta}$--variable by
\begin{equation*}
(D[\lambda]f)(q_\delta,\xi_{\gamma,\delta})
= \lambda f(q_\delta,\lambda \xi_{\gamma,\delta}).
\end{equation*}
Then for any fixed $q_\delta \in \R^2$ and the usual norm in $L^2(\R^2, d\xi_{\gamma,\delta})$
\begin{equation*}
    \norm{ D[\lambda] f(q_\delta,\cdot) }^2 =  \norm{ f(q_\delta,\cdot) }^2 \,  .
\end{equation*}
Fixing $\lambda = 2(3)^{-1/2}$, we obtain for any $q_\gamma\in\R^2$
\begin{equation}\label{sec3:eq:double_well_scaling}
        \left\langle D[\lambda]f, B[q_\gamma] D[\lambda]f \right\rangle_\circ  =  \left\langle f, H[(0,0,\lambda q_\gamma)] f \right\rangle_{dx}
\end{equation}
where $H[\cdot]$ is the double-well operator on $H^1(\R^4,dx)$ defined in Lemma \ref{sec2:lem:proto_efimov}. For any fixed $q_\gamma \in \R^2$ let $\widetilde \Phi_R \in H^1(\R^4,dx)$ be the function according to Lemma \ref{sec2:lem:proto_efimov} with the explicit choice 
\begin{equation}\label{sec4:eq:choice_of_R}
    \ell(r)=(0,0,r), \quad \text{ with } \quad r=\lambda q_\gamma
\end{equation}
and define
\begin{equation} \label{sec4:eq:dilated_Phi}
    \Phi(q_\gamma, \cdot )= D[1/\lambda]\widetilde \Phi(\ell(\lambda q_\gamma), \cdot) 
\end{equation}
 such that
 \begin{equation*}
     \Phi(q_\gamma, q_\delta,\xi_{\gamma,\delta}) = \lambda \widetilde \Phi\left((0,0,\lambda q_\gamma), (q_\delta,  \lambda \xi_{\gamma,\delta})\right)
 \end{equation*}
 We choose $u \in H^1(\R^2, d{q_\gamma})$ with $\norm{u}_\ast^2 = \langle u, u \rangle_\ast = 1$ supported outside a compact set which will be fixed later. We construct the Ansatz $\psi \in \mathbb{H}^1$ (with $\mathbb{H}^1$ defined in \eqref{sec4:eq:hilbert_space_6}) by setting
 \begin{equation*}
     \psi(q_\gamma,\xi_{\gamma,\delta},q_\delta) = u(q_\gamma) \Phi(q_\gamma,q_\delta,\xi_{\gamma,\delta}), \quad \forall q_\gamma,q_\delta,\xi_{\gamma,\delta} \in \R^2 .
 \end{equation*}

In the next step, we analyze the energy expectation $H_{\mathrm{COM}}$ in \eqref{sec4:eq:hcom_a_b}. Let $\langle \cdot, \cdot \rangle$ be the inner product of the Hilbert space $\mathbb{H}^1$ defined in \eqref{sec4:eq:hilbert_space_6}, then
\begin{equation} \label{sec4:eq:energy_exp}
    \langle \psi, H_{\mathrm{COM}} \psi \rangle = \langle \psi, A \psi \rangle + \langle \psi, B \psi \rangle  - (1/4) \left\langle \psi, P_{\xi_{\gamma,\delta}}^2 \psi \right\rangle
\end{equation}
We study each of the terms in the right-hand side of \eqref{sec4:eq:energy_exp} separately starting with the last one. By definition and Fubini's theorem,
\begin{equation} \label{sec4:eq:born_op_correction}
    \left\langle \psi, P_{\xi_{\gamma,\delta}}^2 \psi \right\rangle = \left\langle u,\left \langle\Phi ,  P_{\xi_{\gamma,\delta}}^2 \Phi \right\rangle_\circ u\right\rangle_\ast . 
\end{equation}
Using the definition of $\Phi$ in \eqref{sec4:eq:dilated_Phi} one finds
\begin{equation} \label{sec4:eq:xi_derivative_scaled}
    \begin{split}
        \left \langle \Phi, P_{\xi_{\gamma,\delta}}^2 \Phi \right\rangle_\circ &= \lambda^2 \int_{\R^2\times \R^2} \abs{\nabla_{y} \widetilde\Phi(\ell(r),(x,y))}^2 d(x,y) \\
        &= \lambda^2 \left \langle \widetilde\Phi(\ell(r),\cdot), P_{y}^2 \widetilde\Phi(\ell(r),\cdot) \right\rangle_{d(x,y)}
    \end{split}
\end{equation}
where the expectation value over $d(x,y)$ is the usual inner product of
$L^2(\R^2\times \R^2)$. Note that this expectation value is a function of parameter $q_\gamma$ as the shift $\ell(r)$ depends on $q_\gamma$ (see \eqref{sec4:eq:choice_of_R}).
 
Next we study the expectation $\langle \psi, B \psi \rangle$ and find by using the relation \eqref{sec3:eq:double_well_scaling}
\begin{equation*}
    \langle \psi, B\psi \rangle = \left \langle u, \left \langle \Phi, B[q_\gamma] \Phi\right \rangle_\circ u \right \rangle_\ast = \left \langle u, \left \langle \widetilde\Phi, H[\ell] \widetilde \Phi\right \rangle_{dx} u \right \rangle_\ast\, .
\end{equation*}
By definition of $\ell(r)$ in \eqref{sec4:eq:choice_of_R} with $r=(0,0,\lambda q_\gamma)$ we have $L = \abs{\ell(r)} = \lambda\abs{q_\gamma}$ and consequently by application of Lemma \ref{sec2:lem:proto_efimov} for given $\varepsilon>0$ there exists $T>0$ depending on $\varepsilon>0$ such that for $\abs{q_\gamma}>T$
\begin{equation*}
  \left \langle \widetilde\Phi_R, H[\ell] \widetilde \Phi_R\right \rangle_{dx} \leq    -(2-\varepsilon)\lambda^{-2}\abs{q_\gamma}^{-2} \log(\lambda \abs{q_\gamma})^{-1}
\end{equation*}
Choosing 
\begin{equation}\label{sec4:eq:supp(u)}
    \supp(u) \subset \{ q_\gamma \in \R^2 : \abs{q_\gamma}>2T\}
\end{equation}
we conclude
\begin{equation}\label{sec4:eq:B_scaled}
     \langle \psi, B\psi \rangle \leq -(2-\varepsilon) \left \langle u, \abs{q_\gamma}^{-2} \lambda^{-2} \log(\lambda \abs{q_\gamma})^{-1} u\right \rangle_\ast 
\end{equation}

We analyze now the expectation of operator $A$ to find
\begin{equation}\label{sec4:eq:A_exp_step1}
    \langle \psi, A \psi \rangle = \left \langle u\Phi,  P_{q_\gamma}^2  \left( u\Phi \right) \right \rangle  + \left \langle u \left \langle \Phi , W \Phi \right \rangle_\circ   u \right \rangle_\ast
\end{equation}
where $W$ is short--range. Since $\|\Phi\|_{L^2(dq_\delta\,d\xi_{\gamma,\delta})}=1$ for all
$q_\gamma$ and $\Phi$ depends smoothly on $q_\gamma$, we have
\begin{equation}\label{sec4:eq:cross_term_vanishes}
    P_{q_\gamma}\langle \Phi,\Phi\rangle_\circ
= 2\,\Re\langle \Phi,P_{q_\gamma} \Phi\rangle_\circ = 0,
\end{equation}
and hence $\langle \Phi,P_{q_\gamma}\Phi\rangle_\circ=0$.
Using \eqref{sec4:eq:cross_term_vanishes} we derive from \eqref{sec4:eq:A_exp_step1}
\begin{equation}\label{sec4:eq:A_exp_step2}
    \langle \psi, A \psi \rangle = \left \langle u,  \left( P^2_{q_\gamma} +\langle \Phi, W \Phi \rangle_\circ \right) u \right \rangle_\ast + \left \langle u, \left \langle \Phi, P_{q_\gamma}^2  \Phi\right \rangle_\circ u \right \rangle_\ast
\end{equation}

By Lemma~\ref{sec2:lem:proto_efimov}, the family $r\mapsto\widetilde\Phi(\ell(r),z)$ with $\ell(r) = (0,0,r)$ for any $z\in \R^4,r\in \R^2$ can be chosen $H^1(\R^2,dr)$--smooth for $r$ sufficiently large. We keep track of the dilation $D[\lambda]$ by using the definition of $\Phi$ in \eqref{sec4:eq:dilated_Phi} and find
\begin{equation}\label{sec4:eq:deriviative_q_gamma}
    \left \langle \Phi, P_{q_\gamma}^2  \Phi \right \rangle_\circ = \lambda^2 \int_{\R^4 }\abs{\nabla_{r} \widetilde\Phi(\ell(r),z) }^2 dz \Big\vert_{r=\lambda q_\gamma}  = \lambda^2 \langle\widetilde\Phi(\ell(r),\cdot ), P^2_r \widetilde\Phi(\ell(r),\cdot) \rangle_{\R^4} \Big\vert_{r=\lambda q_\gamma} \, .
\end{equation}
Combining \eqref{sec4:eq:A_exp_step1}, \eqref{sec4:eq:A_exp_step2} together with \eqref{sec4:eq:deriviative_q_gamma} yields
\begin{equation} \label{sec4:eq:A_op_scaled}
    \langle \psi, A \psi \rangle =  \left \langle u, \left( P^2_{q_\gamma} +\langle \Phi, W \Phi \rangle_\circ \right) u \right \rangle_\ast +  \lambda^2 \left \langle u,  \langle\widetilde\Phi(\ell(r),\cdot ), P^2_r \widetilde\Phi(\ell(r),\cdot) \rangle_{\R^4} \Big\vert_{r=\lambda q_\gamma}  u \right \rangle_\ast
\end{equation}

Inserting \eqref{sec4:eq:xi_derivative_scaled}, \eqref{sec4:eq:B_scaled} and \eqref{sec4:eq:A_op_scaled} we find
\begin{equation*}
    \begin{split}
        &\langle \psi, H_{\mathrm{COM}} \psi \rangle \\
        &=  \left \langle u, \left( P^2_{q_\gamma} +\langle \Phi, W \Phi \rangle_\circ \right) u \right \rangle_\ast +  \left \langle u, \left \langle \widetilde\Phi(\ell(r),\cdot ), H[(0,0,r)] \widetilde \Phi(\ell(r),\cdot )\right \rangle_{d(x,y)}\Big\vert_{r=\lambda q_\gamma} u \right \rangle_\ast \\
        &+ \lambda^2 \left \langle u,  \left \langle\widetilde \Phi(\ell(r),\cdot ), \left( P^2_r  - \frac{1}{4}  P_{y}^2  \right) \widetilde\Phi(\ell(r),\cdot ) \right\rangle_{d(x,y)} \Bigg \vert_{r=\lambda q_\gamma}  u \right \rangle_\ast
    \end{split}
\end{equation*}
For fixed $\varepsilon>0$ we use that $W$ is short range. Therefore, we can choose $T>0$ large enough, meaning that $u$ is supported outside a large compact set to find
\begin{equation*}
    \langle \Phi, W(q_\gamma,\cdot) \Phi \rangle_\circ \leq \varepsilon \abs{q_\gamma}^{-2} \log(\lambda \abs{q_\gamma})^{-1}, \quad \abs{q_\gamma} \geq T
\end{equation*}
By application of Lemma \ref{sec2:lem:proto_efimov} for the same $\varepsilon>0$ we find $c(\varepsilon)>0$ such that
\begin{equation*}
    \langle \psi, H_{\mathrm{COM}} \psi \rangle \leq  \left \langle u, \left[ P^2_{q_\gamma}  - c(\varepsilon) \abs{q_\gamma}^{-2} \log(\lambda \abs{q_\gamma})^{-1} \right]  u \right \rangle_\ast 
\end{equation*}
where (recall $\lambda = 2(3)^{-1/2}$)
\begin{equation*}
    c(\varepsilon) = \lambda^{-2}(2-\varepsilon)-\left(\tfrac{2}{3}+\varepsilon\right) -\varepsilon = \tfrac{3}{4}(2-\varepsilon)-\left(\tfrac{2}{3}+\varepsilon\right) -\varepsilon = \tfrac{5}{6} - \tfrac{11}{4}\varepsilon.
\end{equation*}
For $\varepsilon>0$ small enough we have $c(\varepsilon)=c>1/2$. Consequently for $T>0$ large enough 
\begin{equation*}
    \langle \psi, H_{\mathrm{COM}} \psi \rangle \leq  \left \langle u, \left[ P^2_{q_\gamma}  - \tfrac{1}{2} \abs{q_\gamma}^{-2} \log( \abs{q_\gamma})^{-1} \right]  u \right \rangle_\ast 
\end{equation*}
which proves the statement of Lemma \ref{sec2:lem:adiabatic}.
\end{proof}
%%%%%%%%%%%%%%%%%%%%%%%%%%%%%%%%%%%%%%%%
%%%%%%%%%%% SECTION TUNNELING (NEW VERSION)
%%%%%%%%%%%%%%%%%%%%%%%%%%%%%%%%%%%%%%%%%
\section{Quantum Tunneling for Resonances in $\R^4$} \label{sec3:new:Tunneling}
In this section we give the proof of Lemma~\ref{sec2:lem:proto_efimov}. In particular, we study  the double--well operator $H[\ell]$ from \eqref{sec3:eq:dw_op}, where $\ell\in\R^4$ denotes the separation vector. For convenience, we recall the definition of $H[\ell]$ here
\begin{equation*}
H[\ell]=P_x^2+V(x-\ell/2)+V(x+\ell/2)\quad \text{on }L^2(\mathbb R^4,dx).
\end{equation*}
Our analysis treats the situation for which the single-well operator $P^2_x+V$ admits a virtual level at zero according to Definition \ref{sec1:def:virtual_lvl}. We give two lemmas corresponding to two statements in Lemma~\ref{sec2:lem:proto_efimov}.
\begin{lemma}[Ground State Estimate]
\label{sec5:lem:energy}
There exists a family of functions $\{\phi(\ell,\cdot)\}_{\ell \in \R^4} \subset H^1(\R^4)$ such that
for every $\varepsilon>0$ there exists $L_0(\varepsilon)>0$ with
\begin{equation} \label{sec5:eq:gs_est}
   \frac{\langle\phi(\ell,\cdot),H[\ell]\phi(\ell,\cdot)\rangle}{\|\phi(\ell,\cdot)\|^2}
\le -(2-\varepsilon)L^{-2}\log(L)^{-1},
\quad\text{ if }\,  |\ell|=L>L(\varepsilon). 
\end{equation}
\end{lemma}
\begin{lemma}[Kinetic Control]
\label{sec5:lem:aux}
The family $\{\phi(\ell,\cdot)\}_{\ell \in \R^4} \subset H^1(\R^4)$ from Lemma~\ref{sec5:lem:energy} can be chosen such that the map $\ell \mapsto \phi(\ell,\cdot) (\norm{\phi(\ell,\cdot)})^{-1}$ is in $H^{1}_{\mathrm{loc}}(\R^4)$. For the choice $\ell(r)=(0,0,r)$ with $r\in \R^2$ for any
$\varepsilon>0$ there exists $L_0(\varepsilon)>0$ with
\begin{equation} \label{sec5:eq:aux_kin}
    \left\langle \frac{\phi(\ell(r),\cdot)}{\norm{\phi(\ell(r),\cdot)}} ,
\left(P_{r}^2-\tfrac{1}{4}P_{(x_3,x_4)}^2\right)\frac{\phi(\ell(r),\cdot)}{\norm{\phi(\ell(r),\cdot)}} \right\rangle
\le \left(\tfrac{2}{3}+\varepsilon\right){L}^{-2}\log(L)^{-1}, \quad  |\ell|=L>L(\varepsilon).
\end{equation}
\end{lemma}
\begin{remark}
Whereas the bound stated in Lemma~\ref{sec5:lem:energy} does apply for any $\ell \in \R^4$ the inequality in \eqref{sec5:eq:aux_kin} from Lemma~\ref{sec5:lem:aux} uses the fact that we choose $\ell(r)= (0,0,r)$ with $ r\in \R^2$.

Assuming Lemmas~\ref{sec5:lem:energy} and~\ref{sec5:lem:aux},
Lemma~\ref{sec2:lem:proto_efimov} follows directly, which we
present below this remark.
The construction of the functions  $\{\phi(\ell,\cdot)\}_{\ell \in \R^4}$  is given in
Subsection~\ref{sec5:sub:test}, while the proofs of
Lemmas~\ref{sec5:lem:energy} and~\ref{sec5:lem:aux} are provided in Subsections~\ref{sec5:sub:energy} and~\ref{sec5:sub:aux}, respectively.
Some technical estimates in the proofs of these Lemmas are deferred to Section~\ref{sec:aux_lemmatas}.
\end{remark}
\begin{proof}[Proof of Lemma~\ref{sec2:lem:proto_efimov}]
Fix $r\in\R^2$ and set $\ell(r)=(0,0,r)\in\R^4$, so $L=|\ell(r)|=|r|$.
Let $\widetilde\Phi(\ell(r),\cdot)=\phi(\ell(r),\cdot)/\|\phi(\ell(r),\cdot)\|$. The bound in \eqref{sec5:eq:gs_est} is exactly the inequality required in Lemma~\ref{sec2:lem:proto_efimov}, hence \eqref{sec3:eq:dw_op} follows for $|r|$ large enough.

Moreover, by Lemma~\ref{sec5:lem:aux} the map $r\mapsto\widetilde\Phi(\ell(r),\cdot)$ is in $H^{1}_{\mathrm{loc}}(\R^2)$, and \eqref{sec5:eq:aux_kin} implies \eqref{sec3:eq:for_born_op}. This proves Lemma~\ref{sec2:lem:proto_efimov}.
\end{proof}
\subsection{The Test Function}\label{sec5:sub:test}
In this section we give the construction of the family of test functions $\{\phi(\ell,\cdot)\}_{\ell \in \R^4} \subset H^1(\R^4)$ from Lemma~\ref{sec5:lem:energy} (respectively Lemma~\ref{sec2:lem:proto_efimov}). In contrast to the standard tunneling scenario, the single--well operator
\[
h = P_x^2 + V(x) \qquad \text{on } L^2(\R^4),
\]
admits no square--integrable ground state but only a virtual level at zero energy
in the sense of Definition~\ref{sec1:def:virtual_lvl}.
As a consequence, the usual tunneling ansatz based on symmetric and antisymmetric
linear combinations of translates of a ground state cannot be applied:
any such combination remains non--square integrable and therefore cannot yield an
energy expectation below the bottom of the essential spectrum of $H[\ell]$.
To overcome this difficulty, we adapt a threshold trial--state construction due to Tamura \cite{T:1991} based on a suitable modification of the resonance profile using the Green function of the free Laplacian in $\R^4$. In dimension four this Green function can be expressed in terms of modified Bessel functions of the second kind, whose asymptotic behavior near zero plays a crucial role in the analysis.
We begin by collecting the properties of the zero--energy resonance and of the Green function that will be needed in the following. These are summarized in the following two lemmas.
\begin{lemma}[Resonance Function]\label{sec5:lem:AbfallverhaltenResonanz}
    Let $V \in L^2{(\R^4)}$ be real--valued and short--range in the sense of Definition \ref{sec2:def:short-range} with parameters $R,\delta>0$.
Assume that the operator
\[
h=P_x^2+V
\]
has a virtual level at zero in the sense of Definition~\ref{sec1:def:virtual_lvl}.
Then there exists a (up to normalization) unique function
$\varphi_0\in\dot H^1(\R^4)$ solving $h\varphi_0=0$ in the sense of distributions, which can be chosen real-valued and positive.
Moreover, there exists a constant $c_0\in\R$ such that
\begin{equation}
\label{sec5:eq:exact_tail}
\varphi_0(x)=c_0 |x|^{-2} + \epsilon_1(x).
\end{equation}
Moreover, for any $0<\kappa<1$ there exists a constant $C_\kappa<\infty$ and $L_2>0$ (depending on $R$) such that the error is bounded by  
  	  \begin{equation*}
  	  	|\epsilon_1(x)| 
  	  	  \le 
  	  	    C_\kappa |x|^{-2-2\kappa \delta} 
  	  	    \quad \text{for } |x|\ge L_2\, .
  	  \end{equation*}
      Furthermore, for the gradient of
  of $\varphi_0$, one has a continuous version and 
  \begin{align}\label{eq:asymptotic virtual level_deriv}
  	\nabla\varphi_0(x)
  	=  -2c_0x|x|^{-4} 
  	  +\epsilon_2(x)
  \end{align} 
  where for any $0<\kappa<1$ there exists $c_\kappa>0$ and $L_3>0$ (depending on $R$) with
    	  \begin{equation*}
  	  	|\epsilon_2(x)| 
  	  	  \le 
  	  	    C_\kappa |x|^{-3-2\kappa \delta} 
  	  	    \quad \text{for } |x|\ge L_3.
  	  \end{equation*}
\end{lemma}
\begin{proof}
    The proof of this statement is the content of \cite[Section 3.2, 3.3 and Appendix D]{thesis:S:2025}. In particular, we refer to \cite[Theorem 3.2.2]{thesis:S:2025} and the subsequent Remarks. For convenience of the reader we proof in the Appendix Lemma~\ref{sec5:lem:AbfallverhaltenResonanz} which corresponds to the case of compactly supported radial symmetric potentials for which the functions $\epsilon_1,\epsilon_2$ vanish. The proof of Lemma~\ref{sec5:lem:AbfallverhaltenResonanz} in \cite{thesis:S:2025} is a careful extension of these arguments. 
\end{proof}
\begin{remark}\label{sec5:rem:parameters}
     Note that in \cite[Proposition 3.1]{T:1991} a similar statement with parameter $\kappa=1$ is provided without a proof. The estimates above on the error terms with $\kappa<1$ are sufficient for the subsequent analysis. In the following, we fix $\kappa=1/2$ and define $\rho_0 = \max\{L_2,L_3\}$.

     The constant $c_0$ can be expressed more explicitly meaning that (in dimesion $d=4$)
     \begin{equation*}
         c_0 = -\frac{\langle V,\varphi_0\rangle_{L^2(\R^4)}}{4\pi^2}.
     \end{equation*}
     The inner product above is in the bosonic case always non--vanishing whereas for different particle statistics the leading contribution of the resonance may vanish. This is the only place where the bosonic symmetry of particles is used.
\end{remark}
\begin{lemma}[Green Function of $-\Delta+\mu^2$ in $\R^4$]
\label{sec5:lem:EigenschaftenGmu}
Let $\mu>0$ and define
\begin{equation}\label{sec:eq:Gmu_def}
   G_\mu(x) := \mu\,K_1(\mu|x|)|x|^{-1}, \qquad x\in\R^4\setminus\{0\}, 
\end{equation}
where $K_1$ denotes the modified Bessel function of the second kind.
Then $G_\mu \in C^\infty(\R^4\setminus\{0\})$ and satisfies
\begin{equation}\label{sec5:eq:pdeG}
   (-\Delta+\mu^2)G_\mu = 0
\qquad \text{in }\R^4\setminus\{0\}. 
\end{equation}
Moreover, $G_\mu \in L^1_{\mathrm{loc}}(\R^4)$ and solves
\begin{equation}
(-\Delta+\mu^2)G_\mu = 4\pi^2\,\delta
\label{sec5:eq:GleichungGmu}
\end{equation}
in the sense of distributions.
\end{lemma}
\begin{remark}
The singular behavior $G_\mu(x)\sim |x|^{-2}$ as $|x|\to0$ implies that, although
$G_\mu\in L^1_{\mathrm{loc}}(\R^4)$, one has
\[
G_\mu \notin L^2_{\mathrm{loc}}(\R^4)
\qquad\text{and}\qquad
\nabla G_\mu \notin L^2_{\mathrm{loc}}(\R^4).
\]
In particular, $G_\mu\notin H^1_{\mathrm{loc}}(\R^4)$. 
Moreover, $G_\mu$ decays exponentially at infinity. In particular, for every $s>0$,
\begin{equation*}
    G_\mu \in L^2(\R^4\setminus B_s(0))
\qquad\text{and}\qquad
\nabla G_\mu \in L^2(\R^4\setminus B_s(0)).
\end{equation*}
Up to the normalization constant $4\pi^2$, the function $G_\mu$ coincides with the
fundamental solution obtained via Fourier transform.
\end{remark}
\begin{proof}[Proof of Lemma~\ref{sec5:lem:EigenschaftenGmu}]
    The arguments are standard and we only give a brief summary. By rotational invariance of $-\Delta+\mu^2$, the function $G_\mu$ is radially symmetric. Consequently, for $\overline G_\mu: \abs{x} \mapsto G_\mu(\abs{x}) $ using the radial part of the Laplacian in $\R^4$ the PDE in \eqref{sec5:eq:pdeG} reads 
    \begin{equation*}
    -(\overline G_\mu)''(r) - \frac{3}{r}(\overline G_\mu)'(r) + \mu^2 \overline G_\mu(r)=0, \quad r\neq 0\, .
    \end{equation*}
    Introducing $u(r):=r\,\overline G_\mu(r)$ reduces this equation to the modified
    Bessel equation of order one. The solution decaying at infinity is given by
    $u(r)=\mu K_1(\mu r)$, which yields the stated form of $G_\mu$.
    Using the asymptotic expansion in \eqref{app1:eq:def_bessels} yields
    \begin{equation*}
        K_1(r)=r^{-1}+O(r\log r)\qquad\text{as } r\to0.
    \end{equation*}
    Consequently $G_\mu \in L^1_{\mathrm{loc}}(\R^4)$ and by a standard integration by parts argument for test function $y \in C^\infty_c(\R^4)$ then shows
    \begin{equation*}
       \langle(-\Delta+\mu^2)G_\mu,y\rangle = 4\pi^2\,y(0) 
    \end{equation*}
    which proves \eqref{sec5:eq:GleichungGmu}.
\end{proof}
We give now the explicit definition of function $\phi(\ell,\cdot)$ for any fixed $\ell\in \R^4\setminus\{0\}$. For this we define the following subsets of $\R^4$:
    \begin{align}\label{sec5:eq:Gebiete_def}
        A_\rho(\ell/2) \coloneqq B_{2\rho}(\ell/2)\setminus B_\rho(\ell/2) \, \text{ and }\, \Omega_{\rho}(\ell/2)\coloneqq \R^4 \setminus \left( B_{\rho}(\ell/2)\cup B_{\rho}(-\ell/2)\right)
    \end{align}
where $B_s(\eta)\subset \R^4$ denotes the ball of radius $s>0$ with center at $\eta\in \R^4$.
\begin{figure}[h]
    \centering
    \input{figure_003}
    \caption{The sets $A_\rho(\pm\ell/2), B_\rho(\pm\ell/2)$ and $\Omega_{2\rho}(\ell/2)$. In the relevant regime the parameter $L=|\ell|$ is large which also means that the sets $A_\rho(+\ell/2), B_\rho(+\ell/2)$ are far from $A_\rho(-\ell/2), B_\rho(-\ell/2)$.   }
    \label{sec5:fig::areas}
\end{figure}
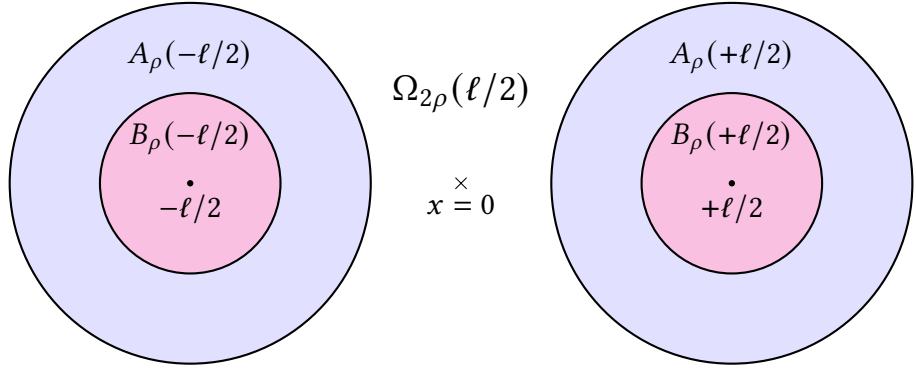
We define for $\rho>0$ the cut--off functions $u_{\rho}:\R^4 \to [0,1]$ and $v_{\rho} = 1-u_{\rho}$ as
\begin{equation}\label{sec5:cut_off_u}
    u_{\rho}(x) = \chi\left(\abs{x}/\rho\right)
\end{equation}
with $\chi\colon[0,\infty)\to[0,1]$ defined as
\begin{equation}\label{sec5:eq:cut_off_chi}
             \chi(r) \coloneqq \min\{ 1, (2-r)_+\} =\begin{cases}
                1, & r\leq 1,\\
                2-r, &1<r\leq 2,\\
                0, & r>2.
            \end{cases}
\end{equation}
The test function $\phi(\ell, \cdot)$ is then defined as:
\begin{definition}\label{sec5:def:Testfunktion}
        Let $G_\mu$ be the function defined in \eqref{sec:eq:Gmu_def} and $\varphi_0$ the zero-energy solution in \eqref{sec5:eq:exact_tail} with normalization $c_0=1$. For any $\ell\in \R^4\setminus\{0\}$ with $L=|\ell|$ let $\rho= \rho(L) = L^\theta$ for some $\theta\in (0,1)$ and $\mu= \mu(L) = \mu_0/L$ for some $\mu_0>0$. We then define $\phi(\ell, \cdot) \in L^2(\R^4)$ for any $x\in \R^4$ as
        \begin{equation*}
            \begin{split}
                \phi(\ell,x) = & u_{\rho}(x-\ell/2)\left(\varphi_0(x-\ell/2)-G_\mu(x+\ell/2)\right) + v_\rho(x-\ell/2)G_\mu(x-\ell/2)\\
                +& u_{\rho}(x+\ell/2)\left(\varphi_0(x+\ell/2)-G_\mu(x-\ell/2)\right) +  v_{\rho}(x+\ell/2)G_\mu(x+\ell/2)
            \end{split}
        \end{equation*}
        with $u_\rho,v_\rho$ defined in \eqref{sec5:cut_off_u}, \eqref{sec5:eq:cut_off_chi}. To restrict the construction to the large–separation regime and avoid
overlapping wells, fix $L_1(\theta)>0$ and let $\Upsilon\in C^\infty(\R^4;[0,1])$
be radial with $\Upsilon(\ell)=0$ for $|\ell|\le L_1$ and
$\Upsilon(\ell)=1$ for $|\ell|\ge 2L_1$.
We then define
\[
\widetilde\Phi(\ell,\cdot)
=
\Upsilon(\ell)\frac{\phi(\ell,\cdot)}{\|\phi(\ell,\cdot)\|}.
\]
\end{definition}
Before proceeding we define the following functions on $\R^4$ for any fixed $\ell\in \R^4$, $\mu_0>0$ and $\theta \in (0,1)$ if $L>0$ is sufficiently large as
\begin{equation}\label{sec5:eq:def_Gamma&f}
    \begin{split}
            \Gamma(\ell,x)&= G_\mu(x +\ell/2)+G_\mu(x -\ell/2) \\
            f(\ell,x)&=\phi(\ell,x) - \Gamma(\ell,x)
    \end{split}
\end{equation}
for any $x\in \R^4$. Then
\begin{equation}\label{sec5:eq:Testfunktion auf Gebieten}
        \phi(\ell,x) =\begin{cases}
            \varphi_0(x \mp \ell/2), & x\in B_\rho(\pm\ell/2),\\
            \Gamma(\ell,x)+f(\ell,x), &x \in A_\rho(\ell/2)\cup A_\rho(-\ell/2),\\
            \Gamma(\ell,x), & x \in  \Omega_{2\rho}(\ell/2).
        \end{cases}
\end{equation}
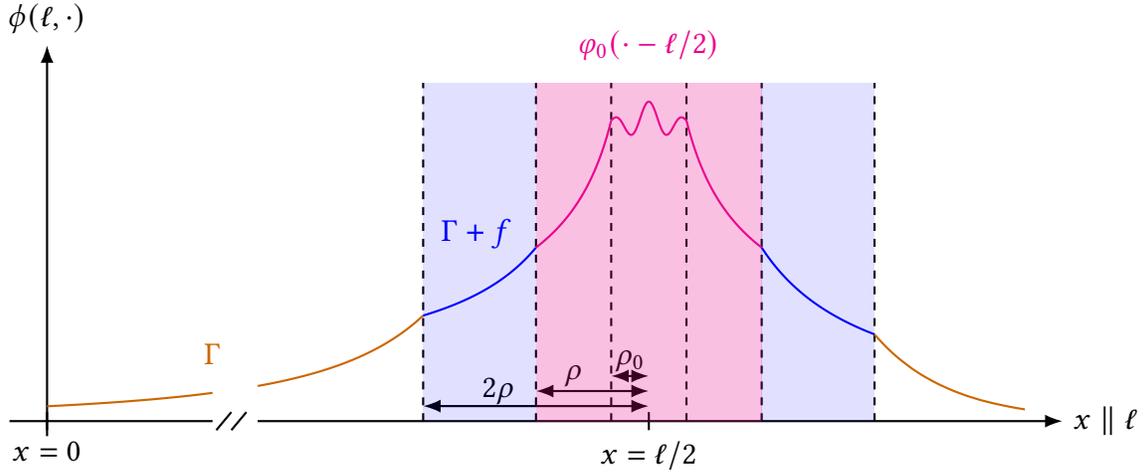
\begin{figure}[h]
    \centering
    \input{figure_004}
    \caption{Sketch of the function $\phi(\ell,\cdot)$ in the vicinity of the potential well centered at $\ell/2$. The function $\phi(\ell,\cdot)$ is symmetric with respect to reflections $x\mapsto -x$. We reduced the sketch to the displayed part.}
    \label{sec5:fig:Skizze Testfunktion}
\end{figure}

\begin{remark}\label{sec5:rem:weak_l_der}
The function $\phi$ coincides with the resonance near each potential well and with a sum of Green functions away from the wells, while the correction term $f$ is supported only on the annuli $A_\rho(\pm \ell/2)$.
For any fixed $x\in \R^4$, the mappings
\[
\ell \mapsto u_{\rho(\ell)}(x\pm \ell/2),
\qquad
\ell \mapsto \varphi_0(x\pm \ell/2)
\]
are locally Lipschitz (hence weakly differentiable) in $\ell$ on $\{|\ell|\ge L_1\}$. The functions $G_\mu(x \pm \ell/2)$ are singular when $x=\mp \ell/2$. However, in the definition of $\phi(\ell,x)$ the terms $G_\mu(x\pm \ell/2)$ appear only multiplied by $v_\rho(x\pm \ell/2)$, and hence are supported in the region $|x\pm \ell/2|\ge \rho(L)$. Therefore, for each fixed $x\in\R^4$ the map $\ell\mapsto \phi(\ell,x)$ is locally Lipschitz (hence weakly differentiable) on $\{|\ell|\ge L_1\}$. We prove below that for each fixed $\ell \in \R^4$ the map $x \mapsto \phi(\ell,x)$ is square integrable. One can check directly from its definition that the weak $\ell$–derivatives of $\phi(\ell,\cdot)$ belong to $L^2(\R^4)$ locally in $\ell$ and therefore by construction the map $\ell \mapsto \|\widetilde \Phi (\ell,\cdot)\|_{L^2(\R^4)}$ belongs to $H^{1}_{\mathrm{loc}}(\R^4)$.

Note that the function $\phi$ depends on the auxiliary parameters $\rho$, $\mu$. The $L$--dependence of $\rho$ and $\mu$ in Definition~\ref{sec5:def:Testfunktion} is chosen in the present form for two reasons.
First, it ensures that the local energy heuristics discussed in the introduction apply in the present setting. Second, since $G_\mu$ satisfies $(-\Delta+\mu^2)G_\mu=0$ and the norm $\|\phi(\ell,\cdot)\|^2$ diverges logarithmically as $L\to\infty$, the ratio $\mu^2/\|\phi(\ell,\cdot)\|^2$ must exhibit the characteristic $L^{-2}\log(L)^{-1}$ scaling of the tunneling energy in dimension four (see \cite{P:1996} and \cite[Chapter 6.2]{thesis:S:2025}).

We suppress the explicit dependence on $\theta\in(0,1),\mu_0>0$ in this function and chose these constants later. The normalized function $\widetilde\Phi$ then coincides with the test function appearing in Lemma~\ref{sec2:lem:proto_efimov} after fixing $\theta,\mu_0>0$. 
\end{remark}

\begin{remark} Before continuing we briefly discuss the choice $\mu_0>0$ and $\theta \in (0,1)$ fixed later and compare it to the case in dimension three in \cite{T:1991}. In dimension three the Green's function of the operator $-\Delta+\mu^2$ is given by $e^{-\mu\abs{x}}/\abs{x}$. On the annulus $A_\rho(\ell/2)$, the interpolation function then behaves as
\begin{equation*}
    \varphi_0(x-\ell/2)
    -\frac{e^{-\mu|x-\ell/2|}}{|x-\ell/2|}
    -\frac{e^{-\mu|x+\ell/2|}}{|x+\ell/2|}.
\end{equation*}
In spatial dimension $d=3$, the resonance satisfies $\varphi_0(x)\sim 1/|x|$. Let $\mu=\mu_0/L$. Choosing $\theta = \theta(\delta) \in (0,1)$ large enough (with $\delta>0$ from Definition \ref{sec2:def:short-range})
\begin{equation} \label{tamura001}
    \frac{1}{|x-\ell/2|}
    -\frac{e^{-\mu|x-\ell/2|}}{|x-\ell/2|}
    \sim \frac{\mu_0}{L}, \quad \text{ on } A_\rho(\ell/2).
\end{equation}
Moreover, for $x\in A_\rho(\ell/2)$ we have $|x+\ell/2|\sim L$, and hence
\begin{equation} \label{tamura002}
    \frac{e^{-\mu|x+\ell/2|}}{|x+\ell/2|}
    \sim \frac{e^{-\mu_0}}{L}.
\end{equation}
In order to make the interpolation function small meaning minimizing the difference between \eqref{tamura001} and \eqref{tamura002} one chooses $\mu_0$ as the unique solution of $t=e^{-t}$ in dimension $d=3$.

In dimension four we cannot proceed in that way. Indeed, due to the series expansion of the modified Bessel functions one finds (for $\theta(\delta)\in(0,1)$ large enough)
\begin{equation*}
    \varphi_0(x-\ell/2)-G_\mu(x-\ell/2)
    \sim -\frac{\mu^2}{2}\left(
    \log|x-\ell/2|+\log(\mu)-\log(2)+\gamma-\frac12
    \right),
\end{equation*}
where $\gamma>0$ is the Euler Mascheroni constant. Consequently, the leading order of this relation is not constant in $x$ (compare to \eqref{tamura001}. In particular, it exhibits a logarithmic dependence. In our analysis we keep track of the  errors due to interpolation and the parameter $\mu_0>0$ enters the discussion as a free parameter. Parameter $\theta \in (0,1)$ will be fixed depending on $\delta>0$. In particular the following choice is not optimal but sufficient
\begin{equation*}
    \theta=\theta(\delta) = \frac{4}{4+\delta}.
\end{equation*}
\end{remark}

\subsubsection{Normalizing the Test Function}
Before continuing with the proofs of Lemmas~\ref{sec5:lem:energy} and~\ref{sec5:lem:aux} we compute the norm $\|\phi(\ell,\cdot)\|$ to leading order in $L=|\ell|$. In the following we always have fixed the choice
\[
\rho=\rho(L)=L^{\theta}, \qquad \mu=\mu(L)=\mu_0/L,
\]
with $\mu_0>0$ and $\theta\in (0,1)$ to be chosen later. We will always assume $L\geq 2L_1$ such that $\Upsilon \equiv 1$ and $B_{2\rho}(\ell/2)\cap B_{2\rho}(-\ell/2)=\varnothing$, see Definition \ref{sec5:def:Testfunktion}.
\begin{remark}\label{sec5:rem:landau}
In the following we use Landau symbols with respect to the limit $L\to\infty$.
For functions, $y_1,y_2\colon(0,\infty)\to\R$ we write
$y_1(L)\in O(y_2(L))$ if there exist constants $C>0$ and $L_0>0$ such that
\[
|y_1(L)|\le C|y_2(L)| \qquad \text{for all } L\ge L_0 .
\]
We also write $a\lesssim b$ if $a\le Cb$ for a constant $C>0$ independent of $L$. We will also use the abuse of notation $O(L^s)=O(L\mapsto L^s)$ for any $s\in \R$.
\end{remark}

We prove the following bound on the norm of the test function.
\begin{lemma}\label{sec5:lem:Norm Testfunktion}
        For fixed $\ell \in \R^4$ let $\phi(\ell,\cdot)$ be the function defined in Definition \ref{sec5:def:Testfunktion} then for $|\ell|=L>0$ large we find
        \begin{align*}
            \|\phi(\ell,\cdot)\|^2_{L^2(\R^4)} =  4\pi^2\log(L) + O(1).
        \end{align*}
\end{lemma}
\begin{proof}
The proof proceeds by reducing the norm of $\phi(\ell,\cdot)$ to the contribution of the
Green's functions $G_\mu$ outside the balls $B_\rho(\pm \ell/2)$, while all interpolation errors remain uniformly bounded.
We subdivide the proof in two steps corresponding to the proof of the following two relations: 
\begin{equation} \label{sec5:eq:relations}
    \begin{split}
        \|\phi(\ell,\cdot)\|^2_{L^2(\R^4)}
         &= \|\Gamma(\ell,\cdot)\|^2_{L^2(\Omega_\rho(\ell/2))} +  4\pi^2\theta\log(L) +O(1), \\
        \|\Gamma(\ell,\cdot)\|^2_{L^2(\Omega_\rho(\ell/2))}&=   4\pi^2(1-\theta)\log(L) +O(1).
    \end{split}
\end{equation}

\medskip
\noindent
\textbf{Step 1:}
We prove the first relation in \eqref{sec5:eq:relations}. Using the symmetry under reflections meaning $\phi(\ell,\cdot)=\phi(-\ell,\cdot)$, we decompose
\begin{equation}\label{sec5:eq:Zerlegung Norm}
\|\phi(\ell,\cdot)\|^2_{L^2(\R^4)}
=
\|\phi(\ell,\cdot)\|^2_{L^2(\Omega_{2\rho}(\ell/2))}
+2\|\phi(\ell,\cdot)\|^2_{L^2(A_\rho(\ell/2))}
+2\|\phi(\ell,\cdot)\|^2_{L^2(B_\rho(\ell/2))}.
\end{equation}
We estimate the third and second term in the right-hand side of \eqref{sec5:eq:Zerlegung Norm} separately.
By \eqref{sec5:eq:Testfunktion auf Gebieten} and the change of variables
$x\mapsto x- \ell/2$ 
\begin{equation}\label{sec5:eq:norm_on_ball}
\begin{split}
    \|\phi(\ell,\cdot)\|^2_{L^2(B_\rho(\ell/2))}
&=\|\varphi_0\|^2_{L^2(B_\rho(0))} \\
&=\|\varphi_0\|^2_{L^2(B_{\rho_0}(0))}
+\|\varphi_0\|^2_{L^2(B_\rho(0)\setminus B_{\rho_0}(0))},
\end{split}
\end{equation}
where $\rho_0>0$ is fixed according to Remark \ref{sec5:rem:parameters}. Since $\rho_0>0$ is independent of $L$ and $\varphi_0 \in \dot H^1$ and consequently also in $L^2_{\loc}$ the norm over $B_{\rho_0}(0)$ is a constant independent of $L$. Using the explicit expression
\begin{equation*}
    \varphi_0(x)= |x|^{-2} + \epsilon_1(x) \quad \text{ where } \quad  \abs{\epsilon_1(x)} \lesssim \abs{x}^{-2-\delta} \quad \text{ for } \, \abs{x} > \rho_0
\end{equation*}
(see, Lemma~\ref{sec5:lem:AbfallverhaltenResonanz} with $c_0=1$ and $\kappa=1/2$) we find 
\begin{equation}\label{sec5:eq:varphi0NormBall}
    \|\varphi_0\|^2_{L^2(B_\rho(0)\setminus B_{\rho_0}(0))}=\abs{S^3} \int_{\rho_0}^\rho r^{-1} \, dr + \int_{B_\rho(0)\setminus B_{\rho_0}(0)} 2\abs{x}^{-2}\epsilon_1(x) + \abs{\epsilon_1(x)}^2\, dx 
\end{equation}
We will prove
\begin{equation}\label{sec5:eq:varphiNormBall_O(1)}
     2\|\varphi_0\|^2_{L^2(B_\rho(0)\setminus B_{\rho_0}(0))} = 4\pi^2\theta\log(L) +O(1).
\end{equation}
The integral on the right--hand side of \eqref{sec5:eq:varphi0NormBall} gives the leading term on the right--hand side of \eqref{sec5:eq:varphiNormBall_O(1)}. It remains to estimate the terms involving $\epsilon_1$, which follows directly by noting that
\begin{equation*}
\begin{split}
    \abs{\int_{B_\rho(0)\setminus B_{\rho_0}(0)} 2\abs{x}^{-2}\epsilon_1(x) + \abs{\epsilon_1(x)}^2\, dx} &\lesssim \int_{B_\rho(0)\setminus B_{\rho_0}(0)} \abs{x}^{-2}\abs{\epsilon_1(x)}\, dx\\
    &\lesssim \int_{B_\rho(0)\setminus B_{\rho_0}(0)} \abs{x}^{-4-\delta}\, dx \in O(1).
    \end{split}
\end{equation*}
Combining \eqref{sec5:eq:norm_on_ball} and \eqref{sec5:eq:varphiNormBall_O(1)} we arrive at
\begin{equation} \label{sec5:eq:contribution_res}
    2\|\phi(\ell,\cdot)\|^2_{L^2(B_{\rho}(\ell/2))} =   4\pi^2\theta\log(L) +O(1).
\end{equation}

We continue with the norm over $A_{\rho}(\ell/2)$ in the right--hand side of \eqref{sec5:eq:Zerlegung Norm}. Inserting \eqref{sec5:eq:Testfunktion auf Gebieten} yields
\begin{align} \label{sec:eq:norm_anulli_expanded}
    \|\phi(\ell,\cdot)\|^2_{L^2(A_\rho(\ell/2))} 
        = &\|\Gamma(\ell,\cdot)\|^2_{L^2(A_\rho(\ell/2))}\\
        &+ 2\Re \scp{f(\ell,\cdot)}{\Gamma(\ell,\cdot)}_{L^2(A_\rho(\ell/2))} + \|f(\ell,\cdot)\|^2_{L^2(A_\rho(\ell/2))}.\notag
\end{align}
We estimate each of the terms in the right-hand side of \eqref{sec:eq:norm_anulli_expanded} individually. To estimate the terms that involve $f$ we use the following bound 
\begin{equation}\label{sec5:eq:prop_fR}
|f(\ell,x)|\lesssim L^{-2}\log L,\, \quad \forall x\in  A_\rho(\ell/2),\, \theta = 4/(4+\delta),
\end{equation}
which follows immediately from the definition of $f$ and the series expansion of $K_1$ in the Appendix in Lemma~\ref{app1:lem:Kj_expansion}. We defer the short proof of \eqref{sec5:eq:prop_fR} to Lemma~\ref{app:lem:Anulus terme}.

Using \eqref{sec5:eq:prop_fR} and $|A_\rho|\sim\rho^4 = L^{4\theta}$ we obtain
\begin{equation}
    \|f(\ell,\cdot)\|^2_{L^2(A_\rho(\ell/2))}\lesssim   L^{-4}\log(L)^2 \abs{A_\rho(\ell/2)} \lesssim  L^{4(\theta-1)}\log(L)^2 \in O(1).\label{sec5:eq: Norm f_R}
\end{equation}
We continue with the inner product term in the right-hand side of \eqref{sec:eq:norm_anulli_expanded} and show that this term is in $O(1)$ using the decay of $G_\mu$ and the scaling $\mu\rho=\mu_0L^{\theta-1}$. Note that for any  $x\in A_\rho(\ell/2)$ it holds $\abs{x+\ell/2}\geq L/2$ and consequently 
\begin{equation*}
        \abs{G_\mu(x+\ell/2)} \lesssim L^{-2}
\end{equation*}
such that
\begin{equation}
        \abs{\scp{f(\ell,\cdot)}{G_\mu(\cdot+\ell/2)}_{L^2(A_\rho(\ell/2))}}
        \lesssim L^{-4}\log(L)\abs{A_\rho(\ell/2)} \lesssim L^{4(\theta-1)}\log(L)\in O(1).\label{sec5:eq:f_R Gmu +R}
\end{equation}
For the term involving $G_\mu(\cdot-\vec \ell/2)$ we use the definition of $G_\mu$ in \eqref{sec:eq:Gmu_def} and a change to spherical coordinates centered at $\ell/2$. This yields  
\begin{equation*}
    \begin{split}
        \abs{\scp{f(\ell,\cdot)}{G_\mu(\cdot-\ell/2)}_{L^2(A_\rho(\ell/2))}}
        &\lesssim L^{-2}\log(L)\int_{A_\rho(\ell/2)}
        \mu \frac{K_1\!\left(\mu\abs{x-\ell/2}\right)}{\abs{x-\ell/2}} \, dx \\
        &\lesssim L^{-2}\log(L)\int_\rho^{2\rho} \mu r^2 K_1(\mu r)\, dr .
    \end{split}
\end{equation*}
Substituting $s=\mu r$ where $\mu=\mu_0L^{-1}$ and using that $s^2K_1(s)$ is bounded on $[0,1]$ (cf.\ Lemma~\ref{app1:lem:Kj_expansion}), we obtain for $L>0$ large enough
\begin{equation}\label{sec5:eq:f_R Gmu -R}
    \begin{split}
        \abs{\scp{f(\ell,\cdot)}{G_\mu(\cdot-\ell/2)}_{L^2(A_\rho(\ell/2))}}
        &\;\lesssim\;
        \mu_0^{-2}\log(L)\int_{\mu\rho}^{2\mu\rho} s^2 K_1(s)\, ds \\
        &\;\lesssim\;
        \log(L)L^{\theta-1} \in O(1),
    \end{split}
\end{equation}
where in the last step we used the scaling relation $\mu\rho=\mu_0L^{\theta-1}$. Combining \eqref{sec5:eq:f_R Gmu +R}, \eqref{sec5:eq:f_R Gmu -R}, \eqref{sec5:eq: Norm f_R} with \eqref{sec:eq:norm_anulli_expanded} and the triangle inequality yields
\begin{align} \label{sec:eq:norm_anulli_estimated}
    \|\phi(\ell,\cdot)\|^2_{L^2(A_\rho(\ell/2))} 
        = &\|\Gamma(\ell,\cdot)\|^2_{L^2(A_\rho(\ell/2))} + O(1).
\end{align}

By definition of the sets $A_\rho, \Omega_\rho$ in \eqref{sec5:eq:Gebiete_def} it follows $A_\rho(+\ell/2) \cup A_\rho(-\ell/2) \cup \Omega_{2\rho}(\ell/2)  = \Omega_{\rho}(\ell/2)$. Therefore, inserting \eqref{sec:eq:norm_anulli_estimated} and \eqref{sec5:eq:contribution_res} in \eqref{sec5:eq:Zerlegung Norm} together with the reflection symmetry of $\phi$ yields the stated relation in the first line of \eqref{sec5:eq:relations}.

\medskip
\noindent
\textbf{Step 2:}
We show for $L>0$ large enough:
\begin{equation} \label{sec5:eq:this_then_proves_norm_lemma}
    \|\Gamma(\ell,\cdot)\|^2_{L^2(\Omega_\rho(\ell/2))}=   4\pi^2(1-\theta)\log(L) +O(1).
\end{equation}
Expanding the square together with the reflection symmetry gives
\begin{equation}\label{sec5:eq:remaining_Gmu_norm_expanded}
    \begin{split}
        \|\Gamma(\ell,\cdot)&\|^2_{L^2(\Omega_\rho(\ell/2))} \\
        &= 2\|G_\mu(\cdot-\ell/2)\|^2_{L^2(\Omega_\rho(\ell/2))} + 2\scp{G_\mu(\cdot+\ell/2)}{G_\mu(\cdot-\ell/2)}_{L^2(\Omega_\rho(\ell/2))}
    \end{split}
\end{equation}
We show that second term in the right--hand side of \eqref{sec5:eq:remaining_Gmu_norm_expanded} is bounded by a constant independent of $L$. Since $G_\mu \geq 0$ we have
\begin{equation*}
    \scp{G_\mu(\cdot+\ell/2)}{G_\mu(\cdot-\ell/2)}_{L^2(\Omega_\rho(\ell/2))} \leq  \scp{G_\mu(\cdot+\ell/2)}{G_\mu(\cdot-\ell/2)}_{L^2(\R^4)}.
\end{equation*}
This inner product over $\R^4$ be solved explicitly (cf. Corollary~\ref{app2:cor:Faltung Gmu mit sich selbst}) to find
\begin{equation} \label{sec5:eq:scp_gmu_o(1)}
    \scp{G_\mu(\cdot+\ell/2)}{G_\mu(\cdot-\ell/2)}_{L^2(\R^4)}= 2\pi^2 K_0(\mu_0) \in O(1).
\end{equation}
Note that For $x\in B_\rho(-\ell/2)$ we have $\abs{G_\mu(x-\ell/2)}\leq c L^{-2} $ and therefore
\begin{equation*}
    \|G_\mu(\cdot-\ell/2)\|^2_{L^2(B_\rho(-\ell/2))} \lesssim  L^{-4}\abs{B_\rho(-\ell/2)} \in O(1)
\end{equation*}
such that
\begin{equation}\label{sec5:eq:norm_gmu_B_rho_o(1)}
        \|G_\mu(\cdot-\ell/2)\|^2_{L^2(\Omega_\rho(-\ell/2))} =% \|G_\mu(\cdot-\Vec{R}/2)\|^2_{L^2(\R^4\setminus B_\rho(\Vec{R}/2))} + O(1) =
        \|G_\mu\|^2_{L^2(\R^4\setminus B_\rho(0))} + O(1).
\end{equation}
Inserting \eqref{sec5:eq:norm_gmu_B_rho_o(1)} and \eqref{sec5:eq:scp_gmu_o(1)} into  the right-hand side of \eqref{sec5:eq:remaining_Gmu_norm_expanded} yields
\begin{equation*}
    \begin{split}
       \|\Gamma(\ell,\cdot)\|^2_{L^2(\Omega_\rho(\ell/2))}= 2\|G_\mu\|^2_{L^2(\R^4\setminus B_\rho(0))}  + O(1).
    \end{split}
\end{equation*}
The remaining norm of $G_\mu$ can be computed explicitly. Using the definition of $G_\mu$ in \eqref{sec:eq:Gmu_def} we find
\begin{equation} \label{sec5:eq:remainin_norm_for_norm}
\begin{split}
    2\|G_\mu\|^2_{L^2(\R^4\setminus B_\rho(0))} &= 2\int_{\R^4\setminus B_\rho(0)} \mu^2\frac{K_1(\mu\abs{x})^2}{\abs{x}^2} dx \\
    &= 4\pi^2 \int_\rho^\infty \mu^2 tK_1(\mu t)^2 dt \\
    & =4\pi^2 \int_{\mu\rho}^\infty sK_1(s)^2 ds
\end{split}
\end{equation}
Using the antiderivative of $sK_1(s)^2$ in Remark \ref{app:cor:bessel_aux} implies
\begin{equation*}
    2\|G_\mu\|^2_{L^2(\R^4\setminus B_\rho(0))} = \left[\frac{1}{2}s^2\left(K_1(s)^2-K_0(s)K_2(s)\right)\right]_{s=\mu\rho}^\infty
\end{equation*}
In Remark~\ref{app:cor:bessel_aux} we show
\begin{equation*}
    \left[\frac{1}{2}s^2\left(K_1(s)^2-K_0(s)K_2(s)\right)\right]_{s=\mu\rho}^\infty =-\log(\mu\rho) +O(1).
\end{equation*}
Since $\mu\rho=\mu_0L^{\theta-1}$ this proves
\eqref{sec5:eq:this_then_proves_norm_lemma} and completes the proof of Lemma~\ref{sec5:lem:Norm Testfunktion}.
\end{proof}
%%%%%%%%%%%%%%%%%%%%%%%%%%%%%%%%%%%%%%%%
%%%%%%%%%%% SUBSECTION ENERGY ESTIMATE
%%%%%%%%%%%%%%%%%%%%%%%%%%%%%%%%%%%%%%%%%
\subsection{Energy Estimate}\label{sec5:sub:energy}
In this subsection we prove Lemma~\ref{sec5:lem:energy}.
We begin by stating three propositions. We then show that these propositions together imply Lemma~\ref{sec5:lem:energy}.
The propositions themselves are proved separately in Subsections~\ref{sec5:subsec:prop1}--\ref{sec5:subsec:prop3}.
Throughout, we assume that $\ell\in\R^4$ with $|\ell|=L>0$ sufficiently large and $\theta=4/(4+\delta)$.
 \begin{proposition}\label{sec5:prop: varphi0 Terme}
        Define
        \begin{equation}\label{sec5:eq:dw_pot}
            V_\pm[\ell](\cdot) =V(\cdot-\ell/2)+V(\cdot+\ell/2),
        \end{equation}
        then for the test function $\phi$ we have
        \begin{equation} \label{sec5:eq:prop01}
            \scp{\phi(\ell,\cdot)}{V_\pm[\ell] \phi(\ell,\cdot)} + 2\|\nabla\phi(\ell,\cdot)\|^2_{L^2(B_\rho(\ell/2))}  = -8\pi^2  L^{-2\theta} + O(L^{-3+\theta}).
        \end{equation}
\end{proposition}
\begin{proof}
    See Subsection~\ref{sec5:subsec:prop1}.
\end{proof}
\begin{proposition}\label{sec5:prop:nabla Annulus Terme}
        Let $\Gamma,f$ be the functions defined in \eqref{sec5:eq:def_Gamma&f}, then
        \begin{equation}
            \|\nabla f(\ell, \cdot)\|^2_{L^2(A_\rho(\ell/2))} \in O(L^{-3+\theta})\label{sec5:eq:Norm nabla f_R}
        \end{equation}
        and 
        \begin{equation}\label{sec5:eq:Cross Term nabla f_R}
            \begin{split}
                &4\Re\langle \nabla f(\ell,\cdot),\nabla\Gamma(\ell,\cdot) \rangle_{L^2(A_\rho(\ell/2))} \\
                &= C_1(\mu_0)L^{-2}-C_2(\mu_0)\log(\mu_0L^{\theta-1})L^{-2}+ O(L^{-3+\theta}),
            \end{split}
        \end{equation}
        with
        \begin{equation}\label{sec5:eq:Konstanten nabla Gmu f_R}
        \begin{split}
            C_1(\mu_0)&= 8\pi^2\mu_0^2(\log(2)-\gamma) + 4\pi^2\mu_0^2 - 16\pi^2\mu_0K_1(\mu_0), \\
            C_2(\mu_0)&=8\pi^2\mu_0^2\, .
        \end{split}
    \end{equation}
    and $\gamma$ the Euler--Mascheroni constant.
\end{proposition}
\begin{proof}
    See Subsection~\ref{sec5:subsec:prop2}.
\end{proof}
\begin{proposition}\label{sec5:prop:nabla G mu Terme}
       Let $G_\mu$ be the function defined in \eqref{sec:eq:Gmu_def}, then
        \begin{equation}\label{sec5:eq:Ergebnis Norm nabla Gmu}
            2\|\nabla G_\mu(\cdot -\ell/2)\|^2_{L^2(\Omega_\rho(\ell/2))} = 8\pi^2 L^{-2\theta} + C_2(\mu_0) \log(\mu_0 L^{\theta-1})L^{-2} + C_3(\mu_0) L^{-2} + O(L^{-3+\theta}).
        \end{equation}
        and
        \begin{equation}\label{sec5:eq:cross Term Gmu}
            2\scp{\nabla G_\mu(\cdot-\ell/2)}{\nabla G_\mu(\cdot+\ell/2)}_{L^2(\Omega_\rho(\ell/2))} =C_4(\mu_0) L^{-2} + O(L^{-3+\theta}),
        \end{equation}
        where $\gamma$ is the Euler--Mascheroni constant and
    \begin{equation}\label{sec5:eq:Konstanten Norm nabla Gmu}
        \begin{split}
            C_2(\mu_0)&=8\pi^2\mu_0^2,\\
            C_3(\mu_0)&= 8\pi^2\mu_0^2(\gamma-\log(2)) - 2\pi^2\mu_0^2, \\
            C_4(\mu_0)&= 16\pi^2 \mu_0K_1(\mu_0)-4\pi^2\mu_0^2K_2(\mu_0).\\
        \end{split}
    \end{equation}
\end{proposition}
\begin{proof}
    See Subsection~\ref{sec5:subsec:prop3}.
\end{proof}

\begin{remark}
    A comparison of the decay rates in Lemma~\ref{sec5:lem:energy} together with the estimate for $\norm{\phi(\ell,\cdot)}^2$ in Lemma~\ref{sec5:lem:Norm Testfunktion} shows that we want to show
    \[
        \langle\phi(\ell,\cdot),H[\ell]\phi(\ell,\cdot)\rangle \sim -c L^{-2}
    \]
    for some constant $c>0$. However, Propositions~\ref{sec5:prop: varphi0 Terme}--\ref{sec5:prop:nabla G mu Terme} contain terms that decay more slowly than $L^{-2}$. We will show that all of these terms cancel.
\end{remark}

\paragraph{Conclusion of the Proof of Lemma~\ref{sec5:lem:energy}:}
Assume for now that Propositions~\ref{sec5:prop: varphi0 Terme}--\ref{sec5:prop:nabla G mu Terme} hold for the test function $\phi$ from Definition~\ref{sec5:def:Testfunktion}. We continue by showing that these estimates imply \eqref{sec5:eq:gs_est}, and hence complete the proof of Lemma~\ref{sec5:lem:energy}.
\begin{proof}[Proof of Lemma~\ref{sec5:lem:energy}]
    By definition of $H[\ell]$ in \eqref{sec3:eq:dw_op} we find for any $\ell \in \R^4$
    \begin{equation} \label{sec5:eq:Quadratische Form H}
        \begin{split}
        \langle\phi(\ell,\cdot),H[\ell]\phi(\ell,\cdot)\rangle = &\|\nabla\phi(\ell,\cdot)\|^2_{L^2(\R^4)} +\int_{\R^4} V_\pm[\ell](x)\abs{\phi(\ell,x)}^2 dx
        \end{split}
    \end{equation}
    where $V_\pm$ is the double--well potential defined in \eqref{sec5:eq:dw_pot}. In the following we will always assume that $|\ell|=L>0$ is chosen sufficiently large. Due to the reflection symmetry of $\phi(\ell,\cdot)$ we conclude from \eqref{sec5:eq:Quadratische Form H} 
    \begin{equation} \label{sec5:eq:quad_form_H_symm}
        \begin{split}
        \langle\phi(\ell,\cdot),H[\ell]\phi(\ell,\cdot)\rangle = &2\|\nabla\phi(\ell,\cdot)\|^2_{L^2(B_\rho(\ell/2))} +2 \|\nabla\phi(\ell,\cdot)\|^2_{L^2(A_\rho(\ell/2))} \\
        &+ \|\nabla\phi(\ell,\cdot)\|^2_{L^2(\Omega_{2\rho}(\ell/2))}+ \langle \phi(\ell,\cdot), V_\pm[\ell]\phi(\ell,\cdot)\rangle_{L^2(\R^4)}.
        \end{split}
    \end{equation}
    Inserting \eqref{sec5:eq:prop01} of Proposition \ref{sec5:prop: varphi0 Terme} into \eqref{sec5:eq:quad_form_H_symm} we find
    \begin{equation}\label{sec5:eq:quad_form_H_firstPROP}
        \begin{split}
&\langle\phi(\ell,\cdot),H[\ell]\phi(\ell,\cdot)\rangle \\
            =&2 \|\nabla\phi(\ell,\cdot)\|^2_{L^2(A_\rho(\ell/2))} + \|\nabla\phi(\ell,\cdot)\|^2_{L^2(\Omega_{2\rho}(\ell/2))}-8\pi^2  L^{-2\theta}  + O(L^{-3+\theta})\\
            =&\Sigma[\phi(\ell,\cdot)] -8\pi^2  L^{-2\theta}+ O(L^{-3+\theta}),
        \end{split}
    \end{equation}
    where we have defined
    \begin{equation*}
        \Sigma[\phi(\ell,\cdot)]=2 \|\nabla\phi(\ell,\cdot)\|^2_{L^2(A_\rho(\ell/2))} + \|\nabla\phi(\ell,\cdot)\|^2_{L^2(\Omega_{2\rho}(\ell/2))}\, .
    \end{equation*}
     We continue by estimating $\Sigma[\phi(\ell,\cdot)]$. Due to definition of functions $f$ and $\Gamma$ in \eqref{sec5:eq:def_Gamma&f} we have
    \begin{equation}\label{sec5:eq:def_f_R:copy}
        \phi(\ell,x) =  \Gamma(\ell, x) + f(\ell,x) , \qquad \text{ on } \,A_\rho(\ell/2) \cup A_\rho(-\ell/2).
   \end{equation}
   Inserting \eqref{sec5:eq:def_f_R:copy} into $\Sigma[\phi(\ell,\cdot)]$ yields
   \begin{equation} \label{sec5:eq:L_term}
    \begin{split}
           \Sigma[\phi(\ell,\cdot)]=& 2\|\nabla f(\ell,\cdot) \|^2_{L^2(A_\rho(\ell/2))} + 4\Re \langle\nabla f(\ell,\cdot) ,\nabla\Gamma(\ell,\cdot)\rangle_{L^2(A_\rho(\ell/2))} \\
            &+ 2 \|\nabla \Gamma(\ell,\cdot)\|^2_{L^2(A_{\rho}(\ell/2))} + \|\nabla \Gamma(\ell,\cdot)\|^2_{L^2(\Omega_{2\rho}(\ell/2))}.
    \end{split}
    \end{equation}
    Since $A_\rho(+\ell/2)\cup A_\rho(-\ell/2)\cup\Omega_{2\rho}(\ell/2) = \Omega_\rho(\ell/2)$ and the reflection symmetry of $\Gamma(\ell,\cdot)$ we see
    \begin{equation}\label{sec5:eq:Geometrie Omega ausgenutzt}
        2 \|\nabla \Gamma(\ell,\cdot)\|^2_{L^2(A_{\rho}(\ell/2))} + \|\nabla \Gamma(\ell,\cdot)\|^2_{L^2(\Omega_{2\rho}(\ell/2))} = \|\nabla \Gamma(\ell,\cdot)\|^2_{L^2(\Omega_{\rho}(\ell/2))}.    
    \end{equation}
    We use \eqref{sec5:eq:Norm nabla f_R} and \eqref{sec5:eq:Cross Term nabla f_R} of Proposition \ref{sec5:prop:nabla Annulus Terme} to estimate the first two terms in the right-hand side of \eqref{sec5:eq:L_term} and conclude with \eqref{sec5:eq:Geometrie Omega ausgenutzt}
    \begin{equation} \label{sec5:eq:L_term_PROP_2}
    \begin{split}
    \Sigma[\phi(\ell,\cdot)]= &\phantom{+} C_1(\mu_0)L^{-2}-C_2(\mu_0)\log(\mu_0L^{\theta-1})L^{-2}\\
    &+ \|\nabla \Gamma(\ell,\cdot)\|^2_{L^2(\Omega_{\rho}(\ell/2))}+O(L^{-3+\theta}).
    \end{split}
    \end{equation}
    Next, we use the definition of $\Gamma$ in \eqref{sec5:eq:def_Gamma&f}, expand the squared norm on the right-hand side of \eqref{sec5:eq:L_term_PROP_2} and use the reflection symmetry to find
    \begin{equation}\label{sec5:eq:derivative_Gmu_norm}
        \begin{split}
            &\|\nabla\Gamma(\ell,\cdot)\|^2_{L^2(\Omega_{\rho}(\ell/2))}=\|\nabla(G_\mu(\cdot+\ell/2)+G_\mu(\cdot-\ell/2))\|^2_{L^2(\Omega_{\rho}(\ell/2))} \\
            &= 2\|  \nabla G_\mu(\cdot-\ell/2) \|^2_{L^2(\Omega_{\rho}(\ell/2))} + 2 \Re \langle \nabla G_\mu(\cdot-\ell/2), \nabla G_\mu(\cdot+\ell/2) \rangle_{L^2(\Omega_{\rho}(\ell/2))}\, .
        \end{split}
    \end{equation}
    Combining \eqref{sec5:eq:Ergebnis Norm nabla Gmu} and \eqref{sec5:eq:cross Term Gmu} with \eqref{sec5:eq:derivative_Gmu_norm} yields
    \begin{equation} \label{sec5:eq:nabla_gmu_final}
        \begin{split}
            \|&\nabla\Gamma(\ell,\cdot)\|^2_{L^2(\Omega_{\rho}(\ell/2))} \\
            &=  8\pi^2 L^{-2\theta} + C_2(\mu_0) \log(\mu_0 L^{\theta-1})L^{-2} + \left(C_3(\mu_0) +C_4(\mu_0)\right) L^{-2} + O(L^{-3+\theta})\, .
        \end{split}
    \end{equation}
    Inserting \eqref{sec5:eq:nabla_gmu_final} in \eqref{sec5:eq:L_term_PROP_2} implies
    \begin{equation*} 
        \Sigma[\phi(\ell,\cdot)] =  \left( C_1(\mu_0) + C_3(\mu_0)+ C_4(\mu_0)\right) L^{-2} +  8\pi^2 L^{-2\theta}+ O(L^{-3+\theta}).
    \end{equation*}
    Inserting the expression for $\Sigma[\phi(\ell,\cdot)]$ into \eqref{sec5:eq:quad_form_H_firstPROP} gives
\begin{equation*}
    \begin{split}
        \langle\phi(\ell,\cdot),H[\ell]\phi(\ell,\cdot)\rangle &= \left( C_1(\mu_0) + C_3(\mu_0)+ C_4(\mu_0)\right) L^{-2}  +  O(L^{-3+\theta}) \\
        &= 2\pi^2 \mu_0^2(1-2K_2(\mu_0)) L^{-2}+ O(L^{-3+\theta}),
    \end{split}
\end{equation*}
where we have used the explicit expression for $C_1(\mu_0), C_3(\mu_0)$ and $ C_4(\mu_0)$ from Proposition \ref{sec5:prop:nabla Annulus Terme} and  \ref{sec5:prop:nabla G mu Terme}. Taking into account the logarithmic growth of the norm of the test Function established in Lemma~\ref{sec5:lem:Norm Testfunktion}, we arrive at
\begin{equation} \label{sec5:eq:almost_done_energy}
\frac{\langle\phi(\ell,\cdot),H[\ell]\phi(\ell,\cdot)\rangle}{\|\phi(\ell,\cdot)\|^2} = \tfrac{1}{2}\mu_0^2(1-2K_2(\mu_0)) L^{-2}\log(L)^{-1}+ O(L^{-3+\theta}).
\end{equation}
Using the expansion $K_2(t)=2t^{-2}+O(1)$ as $t\to0^+$ (see,  \eqref{app1:eq:K-expansions}) we obtain
\begin{equation*}
    \tfrac12\mu_0^2(1-2K_2(\mu_0))\to -2\quad \text{ for }\, \mu_0\to 0^+.
\end{equation*}
Recall that we aim to show an upper bound for the right--hand side of \eqref{sec5:eq:almost_done_energy} of the form $-(2-\varepsilon)L^{-2}\log(L)^{-1}$. Consequently, for this fixed $\varepsilon>0$ we fix $L_0(\varepsilon)>0$ such that for any $L>L_0(\varepsilon)$ the contribution in $O(L^{-3+\theta})$ may only add a term of the size $(\varepsilon/2)L^{-2}\log(L)^{-1}$ and then choose $\mu_0>0$ small enough such that
\begin{equation*}
    \tfrac12\mu_0^2(1-2K_2(\mu_0)) \leq -2+\varepsilon/2.
\end{equation*}
This proves \eqref{sec5:eq:gs_est} and therefore finishes the proof of Lemma~\ref{sec5:lem:energy}.
\end{proof}
\subsubsection{Proof of Proposition~\ref{sec5:prop: varphi0 Terme}}\label{sec5:subsec:prop1}
 \begin{proof}
        We begin by showing
        \begin{equation}\label{sec5:eq:firststep_Vphi0_Terme}
                \int_{\R^4} V(x\pm\ell/2)\abs{\phi(\ell,x)}^2 dx = 2\int_{\R^4} V(x)\abs{\varphi_0(x)}^2 dx + O(L^{-3+\theta}).
            \end{equation}
            First note that
            \begin{equation}\label{sec5:eq:splitting_Vphi0_Terme}
            \begin{split}
                \int_{\R^4} V(x\pm\ell/2)\abs{\phi(\ell,x)}^2 dx =& \scp{\varphi_0(\cdot-\ell/2)}{V_\pm\varphi_0(\cdot-\ell/2)}_{L^2(B_\rho(+\ell/2))}\\
                &+ \scp{\varphi_0(\cdot+\ell/2)}{V_\pm\varphi_0(\cdot+\ell/2)}_{L^2(B_\rho(-\ell/2))}\\
                &+\scp{\phi(\ell,\cdot)}{V_{\pm}\phi(\ell,\cdot)}_{L^2(\Omega_\rho(\ell/2))}.
            \end{split}
            \end{equation}
            Due to the short--range condition of $V_\pm$ and the growth of $\norm{\phi(\ell,\cdot)}^2$ in Lemma \ref{sec5:lem:Norm Testfunktion} we find
            \begin{equation}\label{sec5:eq:error_VPhi}
                \begin{split}
                    \abs{\scp{\phi(\ell,\cdot)}{V_{\pm}\phi(\ell,\cdot)}_{L^2(\Omega_\rho(\ell/2))}} &\leq \sup_{x\in\Omega_\rho(\ell/2)}\abs{V_\pm(x)} \norm{\phi(\ell,\cdot)}^2_{L^2(\R^4)}\\
                    &\lesssim L^{-\theta(2+2\delta)}\log(L) \in O(L^{-3+\theta}),
                \end{split}
            \end{equation}
            where we used that $\theta=4/(4+\delta)$. Next we show
            \begin{equation}\label{sec5:eq:Vphi0_auf_Ganzraum_ergänzen}
                \scp{\varphi_0(\cdot-\ell/2)}{V_\pm\varphi_0(\cdot-\ell/2)}_{L^2(B_\rho(+\ell/2))} = \scp{\varphi_0}{V\varphi_0}_{L^2(\R^4)} + O(L^{-3+\theta}).
            \end{equation}
            Since $V$ is short--range and $\norm{\varphi_0}^2_{L^2(B_\rho(0))} = 4\pi^2\theta \log(L)+O(1)$ (see \eqref{sec5:eq:varphiNormBall_O(1)}) we have
            \begin{equation}\label{sec5:eq:Vpm_zu_V_aufBall}
                \abs{\scp{\varphi_0(\cdot-\ell/2)}{V(\cdot+\ell/2)\varphi_0(\cdot-\ell/2)}_{L^2(B_\rho(+\ell/2))}}\lesssim L^{-\theta(2+2\delta)}\norm{\varphi_0}^2_{L^2(B_\rho(0))} \in O(L^{-3+\theta}).
            \end{equation}
             Using the explicit decay of $\varphi_0$ and $V$ short--range we find, since $\theta=4/(4+\delta)>3/(3+\delta)$,
            \begin{equation}\label{sec5:eq:Vphi0_auserhalb_von_Ball}
                \abs{\int_{\R^4\setminus B_\rho(0)} V\abs{\varphi_0}^2\, dx} \lesssim\int_{\R^4\setminus B_\rho(0)} \abs{x}^{-2-2\delta}\abs{x}^{-4}\, dx \lesssim L^{-\theta(2+2\delta)}\in O(L^{-3+\theta}).
            \end{equation}
            Consequently \eqref{sec5:eq:Vphi0_auf_Ganzraum_ergänzen} follows from \eqref{sec5:eq:Vpm_zu_V_aufBall} and \eqref{sec5:eq:Vphi0_auserhalb_von_Ball}.
            Proceed similarly to see
            \begin{equation}\label{sec5:eq:Vphi0plus_auf_Ganzraum_ergänzen}
                \scp{\varphi_0(\cdot+\ell/2)}{V_\pm\varphi_0(\cdot+\ell/2)}_{L^2(B_\rho(-\ell/2))} = \scp{\varphi_0}{V\varphi_0}_{L^2(\R^4)} + O(L^{-3+\theta}).
            \end{equation}
            Inserting \eqref{sec5:eq:error_VPhi}, \eqref{sec5:eq:Vphi0_auf_Ganzraum_ergänzen} and \eqref{sec5:eq:Vphi0plus_auf_Ganzraum_ergänzen} into \eqref{sec5:eq:splitting_Vphi0_Terme} implies \eqref{sec5:eq:firststep_Vphi0_Terme}.

            By definition, we have
            \begin{equation}\label{sec5:eq:nablaPhi_Ball}
                2 \|\nabla\phi(\ell,\cdot)\|^2_{L^2(B_\rho(\ell/2))} = 2\|\nabla\varphi_0\|^2_{L^2(B_\rho(0))}
            \end{equation}
            Since $\varphi_0 \in \dot{H}^1(\R^4)$ is the zero energy solution, in the sense that
            \begin{equation*}
                \int_{\R^4} V\abs{\varphi_0}^2 dx + \|\nabla\varphi_0\|^2_{L^2(\R^4)} = 0,
            \end{equation*}
            we find
            \begin{equation}\label{sec5:eq:Trick_resonanz}
                \int_{\R^4} V\abs{\varphi_0}^2 dx + \|\nabla\varphi_0\|^2_{L^2(B_\rho(0))} = - \|\nabla\varphi_0\|^2_{L^2(\R^4\setminus B_\rho(0))}.
            \end{equation}
            By \eqref{sec5:eq:firststep_Vphi0_Terme}, \eqref{sec5:eq:nablaPhi_Ball} and \eqref{sec5:eq:Trick_resonanz} it follows
            \begin{equation*}
                \scp{\phi(\ell,\cdot)}{V_\pm\phi(\ell,\cdot)} + 2 \|\nabla\phi(\ell,\cdot)\|^2_{L^2(B_\rho(\ell/2))} = -2\norm{\nabla\varphi_0}^2_{L^2(\R^4\setminus B_\rho(0)} + O(L^{-3+\theta}).
            \end{equation*}
            Using the explicit decay of $\varphi_0$ from Lemma \ref{sec5:lem:AbfallverhaltenResonanz} implies
            \begin{equation} \label{sec5:eq:Ergebnis Norm nabla phi_0 außerhalb}
                \|\nabla\varphi_0\|^2_{L^2(\R^4\setminus B_\rho(0))} = \int_{\R^4\setminus B_\rho(0)}\abs{\frac{2}{\abs{x}^3}}^2 -\frac{4}{\abs{x}^3}\frac{x}{\abs{x}}\cdot\epsilon_2(x) + \abs{\epsilon_2(x)}^2\,dx .
            \end{equation}
            A direct calculation shows
            \begin{equation*}
                \int_{\R^4\setminus B_\rho(0)}\abs{\frac{2}{\abs{x}^3}}^2\, = 2\pi^2\int_\rho^\infty 4r^{-6} r^3 dr = 4\pi^2L^{-2\theta}.
            \end{equation*}
            For the terms involving $\epsilon_2$ in the right--hand side of \eqref{sec5:eq:Ergebnis Norm nabla phi_0 außerhalb} we calculate for the specific choice $\theta=4/(4+\delta)>3/(3+\delta)$
            \begin{equation*}
                \abs{\int_{\R^4\setminus B_\rho(0)}-\frac{4}{\abs{x}^3}\frac{x}{\abs{x}}\cdot\epsilon_2(x) + \abs{\epsilon_2(x)}^2\,dx }\lesssim\int_{\R^4\setminus B_\rho(0)} \abs{x}^{-6-\delta} \, dx \in O(R^{-3+\theta}).
            \end{equation*}
            This concludes the proof of Proposition~\ref{sec5:prop: varphi0 Terme}.
        \end{proof}
\subsubsection{Proof of Proposition~\ref{sec5:prop:nabla Annulus Terme}}\label{sec5:subsec:prop2}
\begin{proof}[Proof of Proposition~\ref{sec5:prop:nabla Annulus Terme}]
    We prove the following two relations
            \begin{equation}
            \|\nabla f(\ell, \cdot)\|^2_{L^2(A_\rho(\ell/2))} \in O(L^{-3+\theta})\label{sec5:eq:Norm nabla f_R:copy},
        \end{equation}
        and 
        \begin{equation}\label{sec5:eq:Cross Term nabla f_R:copy}
            \begin{split}
                &4\Re\langle \nabla f(\ell,\cdot),\nabla\Gamma(\ell,\cdot) \rangle_{L^2(A_\rho(\ell/2))} \\
                &= C_1(\mu_0)L^{-2}-C_2(\mu_0)\log(\mu_0L^{\theta-1})L^{-2}+ O(L^{-3+\theta}).
            \end{split}
        \end{equation}
    with $C_1(\mu_0)$ and $C_2(\mu_0)$ defined in Proposition \ref{sec5:prop:nabla Annulus Terme}.
    
    Note that $|A_\rho(\ell/2)|\lesssim L^{4\theta} $ and consequently to prove \eqref{sec5:eq:Norm nabla f_R:copy} it suffices to show the pointwise inequality
    \begin{equation}\label{sec5:eq:Abschätzung nabla f_R:copy}
         \abs{\nabla_x f(\ell,x)}\lesssim  L^{-2-\theta}\log L, \quad 
        x\in A_\rho(\pm \ell/2).
    \end{equation}
    such that
    \begin{equation}\label{sec5:eq:proof Norm nabla f_R}
                \|\nabla f(\ell,\cdot)\|^2_{L^2(A_\rho(\ell/2))} \lesssim L^{-4+2\theta}\log(L)^2 \in O(L^{-3+\theta}).
    \end{equation}
    The pointwise estimate \eqref{sec5:eq:Abschätzung nabla f_R:copy} follows from the definition of $f$ and the specific choice of $\theta=4/(4+\delta)
    >3/(3+\delta)$. We defer this estimate to Lemma~\ref{app:lem:Anulus terme} in which we prove several bounds on $f$ and its derivatives.

    We continue by proving \eqref{sec5:eq:Cross Term nabla f_R:copy}. As any of the involved functions are real we will omit the real-part expression. Due to the definition of $\Gamma$ in \eqref{sec5:eq:def_Gamma&f} we have
    \begin{equation}\label{sec6:eq:prop_two_terms_fr_gmu}
        \begin{split}
            \langle \nabla f(\ell,\cdot),\nabla\Gamma(\ell,\cdot) \rangle_{L^2(A_\rho(\ell/2))} = &\langle \nabla f(\ell,\cdot),\nabla G_\mu(\cdot+\ell/2) \rangle_{L^2(A_\rho(\ell/2))} \\
            &+ \langle \nabla f(\ell,\cdot),\nabla G_\mu(\cdot-\ell/2) \rangle_{L^2(A_\rho(\ell/2))}.
        \end{split}
    \end{equation}
    We investigate the two terms in the right-hand side of \eqref{sec6:eq:prop_two_terms_fr_gmu} separately.
    We first show
    \begin{equation} \label{sec6:eq:prop_two_terms_fr_gmu first}
        \scp{\nabla f(\ell, \cdot)}{\nabla G_\mu(\cdot+\ell/2)}_{L^2(A_\rho(\ell/2))} \in O(L^{-3+\theta}).
     \end{equation}
    On can take the derivative of $G_\mu$  to find
    \begin{equation}\label{sec5:eq:Ableitung_Gmu}
        \begin{split}
        \nabla G_\mu(x)&=-\mu^2\frac{K_2(\mu\abs{x})}{\abs{x}}\frac{x}{\abs{x}} , \quad x\neq 0, 
        \end{split}
    \end{equation}
            and since $\rho= \rho(\ell)=L^\theta$, $\mu=\mu_0/L$ we have for $x\in A_\rho(\ell/2)$ that $\mu \abs{x+\ell/2} \to \mu_0$ as $L\to \infty$ and consequently
            \begin{equation}\label{sec5:eq:gmu+_is_small}
                \abs{\nabla G_\mu(x+\ell/2) }\lesssim L^{-3}, \quad x\in A_\rho(\ell/2).
            \end{equation}
            Together with the pointwise bound on $\nabla_x f(\ell,x)$ in \eqref{sec5:eq:Abschätzung nabla f_R:copy} we conclude
            \begin{equation} 
                \langle \nabla f(\ell,\cdot),\nabla G_\mu(\cdot +\ell/2)\rangle_{L^2(A_\rho(\Vec{R}/2))}\in O(L^{-3+\theta}).\label{sec5:eq:Cross Term nabla f_R nabla Gmu +R}
            \end{equation}
            Inserting this into \eqref{sec6:eq:prop_two_terms_fr_gmu} yields
            \begin{equation*}
                \langle \nabla f(\ell,\cdot),\nabla\Gamma(\ell,\cdot) \rangle_{L^2(A_\rho(\ell/2))} = \langle \nabla f(\ell,\cdot),\nabla G_\mu(\cdot-\ell/2) \rangle_{L^2(A_\rho(\ell/2))} + O(L^{-3+\theta}).
            \end{equation*}
            Thus equality \eqref{sec5:eq:Cross Term nabla f_R:copy} follows if we can show the following \underline{\textbf{claim}}:
            \begin{equation}\label{sec5:eq:Cross Term nabla f_R mit Gmu-}
            \begin{split}
                &4\langle \nabla f(\ell,\cdot),\nabla G_\mu(\cdot-\ell/2) \rangle_{L^2(A_\rho(\ell/2))} \\
                &= C_1(\mu_0)L^{-2}-C_2(\mu_0)\log(\mu_0L^{\theta-1})L^{-2}+ O(L^{-3+\theta}).
            \end{split}
        \end{equation}
        By direct computation one can find that
        \begin{equation*}%\label{sec5:eq:Zerlegung nabla f_R}
                \nabla_x f(\ell,x) = \sum_{j=1}^5 h_{j}(\ell,x),
            \end{equation*}
            where functions $h_{j}(\ell,\cdot)\colon A_\rho(\ell/2)\to \R^4$ are defined as
            \begin{equation*}%\label{sec6:eq:def h_R,j}
                \begin{split}
                    h_{1}(\ell,x)\coloneqq &-\frac{1}{\rho}\left(\abs{x-\ell/2}^{-2}-G_\mu(x-\ell/2)\right)\frac{x-\ell/2}{\abs{x-\ell/2}},\\
                    h_{2}(\ell,x)\coloneqq & \frac{1}{\rho} G_\mu(x+\ell/2)\frac{x-\ell/2}{\abs{x-\ell/2}} \\
                    h_{3}(\ell,x)\coloneqq &\left(2-\frac{\abs{x-\ell/2}}{\rho}\right)\left(\mu^2\frac{K_2\left(\mu\abs{x-\ell/2}\right)}{\abs{x-\ell/2}}-\frac{2}{\abs{x-\ell/2}^3}\right)\frac{x-\ell/2}{\abs{x-\ell/2}},\\
                    h_{4}(\ell,x) \coloneqq&\left(2-\frac{\abs{x-\ell/2}}{\rho}\right)\left(\mu^2\frac{K_2\left(\mu\abs{x+\ell/2}\right)}{\abs{x+\ell/2}} \right)\frac{x+\ell/2}{\abs{x+\ell/2}},\\
                    h_5(\ell,x)\coloneqq & -\frac{1}{\rho}\epsilon_1(x-\ell/2)\frac{x-\ell/2}{\abs{x-\ell/2}} + u_\rho(x-\ell/2) \epsilon_2(x-\ell/2).
                \end{split}
            \end{equation*}
            \textbf{Note: }By definition of $\epsilon_1$ and $\epsilon_2$ in Lemma~\ref{sec5:lem:AbfallverhaltenResonanz} one has $\nabla_x \epsilon_1(x) = \epsilon_2(x)$.
            
            \medskip
            The claim \eqref{sec5:eq:Cross Term nabla f_R:copy} follows if we can show
            \begin{equation}\label{sec6:eq:goal_Anulus terme mit h_R,j}
                \begin{split}
                    \scp{h_{1}(\ell,\cdot)}{\nabla G_\mu(\cdot-\ell/2)}_{L^2(A_\rho(\ell/2))} &=  D_1(\mu_0) L^{-2}- \frac{1}{4}C_2(\mu_0)\log(\mu_0L^{\theta-1})L^{-2} + O(L^{-3+\theta}),\\
                    \scp{h_{2}(\ell,\cdot)}{\nabla G_\mu(\cdot-\ell/2)}_{L^2(A_\rho(\ell/2))} &= D_2(\mu_0)L^{-2}+O(L^{-3+\theta}),\\
                    \scp{h_{3}(\ell,\cdot)}{\nabla G_\mu(\cdot-\ell/2)}_{L^2(A_\rho(\ell/2))} &= D_3(\mu_0)L^{-2}+O(L^{-3+\theta}),\\
                    \scp{h_{4}(\ell,\cdot)}{\nabla G_\mu(\cdot-\ell/2)}_{L^2(A_\rho(\ell/2))} &\in O(L^{-3+\theta}),\\
                    \scp{h_{5}(\ell,\cdot)}{\nabla G_\mu(\cdot-\ell/2)}_{L^2(A_\rho(\ell/2))} &\in O(L^{-3+\theta}),
                \end{split}
            \end{equation}
            with suitable constants $D_1(\mu_0), D_2(\mu_0), D_3(\mu_0)$ satisfying
            \begin{equation}\label{sec5:eq:aufsummieren D_j}
                C_1(\mu_0) = 4\sum_{j=1}^3 D_j(\mu_0).
            \end{equation}
            This then concludes the proof of Proposition \ref{sec5:prop:nabla Annulus Terme}. We will now prove each of the relations in \eqref{sec6:eq:goal_Anulus terme mit h_R,j} as five independent claims below. For the proof of the claims recall $\theta=4/(4+\delta)>3/(3+\delta)$. \\

\medskip
\noindent\textbf{\underline{Claim 1:}}
Let $L>0$ be sufficiently large. Then
\begin{equation}\label{sec5:eq:Crossterm h_R,1 nabla G mu}
\scp{h_{1}(\ell,\cdot)}{\nabla G_\mu(\cdot-\ell/2)}_{L^2(A_\rho(\ell/2))}
=  D_1(\mu_0) L^{-2}- \frac{1}{4}C_2(\mu_0)\log(\mu_0L^{\theta-1})L^{-2}
+ O(L^{-3+\theta}),
\end{equation}
where
\begin{equation}\label{sec5:prop2:claim_1_consts}
D_1(\mu_0)=2\pi^2\mu_0^2\Big(\tfrac32-\gamma-\log2\Big),
\qquad
C_2(\mu_0)=8\pi^2\mu_0^2.
\end{equation}

\medskip
\noindent\textbf{Proof of Claim 1:}
Write $y=x-\ell/2$ and $t=|y|$. On the annulus 
$A_\rho(\ell/2)=\{x:\rho\le |x-\ell/2|\le 2\rho\}$ we have $t\sim\rho=L^\theta$.
By the definition of $G_\mu$,
\[
h_1(\ell,y+\ell/2) 
=-L^{-\theta}\Big(t^{-2}-\mu\frac{K_1(\mu t)}{t}\Big)\frac{y}{t}=: \widetilde h_1(y),
\qquad
\nabla G_\mu(y)
=-\mu^2\frac{K_2(\mu t)}{t}\frac{y}{t}.
\]
Both vector fields are radial, hence
\[
\widetilde h_1(y)\cdot\nabla G_\mu(y)
=
L^{-\theta}\mu^2
\Big(t^{-2}-\mu\frac{K_1(\mu t)}{t}\Big)\frac{K_2(\mu t)}{t}.
\]
Note that we have
\begin{equation} \label{sec5:es:prop2claim1}
 \scp{h_1(\ell,\cdot)}{\nabla G_\mu(\cdot-\ell/2)}_{L^2(A_\rho(\ell/2))} = \scp{\widetilde h_1}{\nabla G_\mu}_{L^2(A_\rho(0))}.
\end{equation}
Rewriting the integral in spherical coordinates yields
\begin{align}
\scp{\widetilde h_1}{\nabla G_\mu}_{L^2(A_\rho(0))}
&=
2\pi^2 \rho^{-1}\mu^2
\int_{\rho}^{2\rho}
\Big(t^{-2}-\mu\frac{K_1(\mu t)}{t}\Big)
\frac{K_2(\mu t)}{t}\, t^3\,dt \notag\\
&=
2\pi^2\frac{\mu}{\rho}
\int_{\mu\rho}^{2\mu\rho}(1-sK_1(s))K_2(s)\,ds,
\label{sec5:eq:claim1_radial_red}
\end{align}
where we used $s=\mu t$. Set $a=\mu\rho=\mu_0L^{\theta-1}\to0$. From the expansions of Bessel functions in Remark \ref{app:cor:bessel_aux} we find
\[
(1-sK_1(s))K_2(s)
=
-\log(s)-\gamma+\tfrac12+\log(2)+O(s^2),
\qquad s\to0,
\]
and obtain
\[
\int_a^{2a}(1-sK_1(s))K_2(s)\,ds
=
-a\log(a)
+a\Big(\tfrac32-\gamma-\log(2)\Big)
+O(a^3).
\]
Inserting this into \eqref{sec5:eq:claim1_radial_red} gives
\[
\scp{\widetilde h_1}{\nabla G_\mu}_{L^2(A_\rho(0))}
=
2\pi^2\frac{\mu}{\rho}
\Big[
-a\log( a)
+a\Big(\tfrac32-\gamma-\log(2)\Big)
+O(a^3)
\Big].
\]

Since $\mu=\mu_0L^{-1}$ and $\rho=L^\theta$,
\[
\frac{\mu}{\rho}a=\mu^2=\mu_0^2L^{-2},
\qquad
\frac{\mu}{\rho}a^3\in O(L^{-3+\theta}),
\]
and $\log( a)=\log(\mu_0L^{\theta-1})$. Therefore,
\[
\scp{\widetilde h_1}{\nabla G_\mu}_{L^2(A_\rho(0))}
=
2\pi^2\mu_0^2\Big(\tfrac32-\gamma-\log2\Big)L^{-2}
-2\pi^2\mu_0^2L^{-2}\log(\mu_0L^{\theta-1})
+O(L^{-3+\theta}),
\]
which aggress with the constants in \eqref{sec5:prop2:claim_1_consts} and together with \eqref{sec5:es:prop2claim1} proves the claim.

\medskip
            \noindent\textbf{\underline{Claim 2:}} 
           For $L>0$ large enough,
            \begin{equation}\label{sec5:eq:Crossterm h_R,2 nabla G mu}
            \scp{h_{2}(\ell,\cdot)}{\nabla G_\mu(\cdot-\ell/2)}_{L^2(A_\rho(\ell/2))}
            = D_2(\mu_0)L^{-2}+O(L^{-3+\theta}),
            \end{equation}
            where
            \begin{equation} \label{sec5:eq:D2}
            D_2(\mu_0)=-4\pi^2\mu_0K_1(\mu_0).
            \end{equation}

\medskip
\noindent\textbf{Proof of Claim 2:}
Set $y = x - \ell/2$ and $t = |y|$. Then if $x\in A_\rho(\ell/2)$ we have $y\in A_\rho(0)$. By definition of $h_2$ and $G_\mu$ we have
\begin{equation*}
h_2(\ell, y+\ell/2) = \frac{1}{\rho} G_\mu(y+\ell)\frac{y}{t}, \qquad
\nabla G_\mu(y) = -\mu^2 \frac{K_2(\mu t)}{t}\frac{y}{t}.
\end{equation*}
Hence,
\begin{equation}\label{sec5:eq:insert_taylor}
h_2(\ell, y+\ell/2)\cdot \nabla G_\mu(y)
=- \frac{1}{\rho} G_\mu(y+\ell)\mu^2 \frac{K_2(\mu t)}{t}.
\end{equation}
and consequently
\begin{equation}\label{sec5:eq:claim2_shifted}
    \scp{h_{2}(\ell,\cdot)}{\nabla G_\mu(\cdot-\ell/2)}_{L^2(A_\rho(\ell/2))} = \scp{h_{2}(\ell,\cdot+\ell/2)}{\nabla G_\mu}_{L^2(A_\rho(0))}
\end{equation}
In the following, we compute the right--hand side of \eqref{sec5:eq:claim2_shifted}. By Taylor's theorem one finds $G_\mu(y+\ell) = G_\mu(\ell) + O(L^{-3+\theta})$ as we prove below:
Fix $y\in A_\rho(0)$, i.e. $\rho \le |y| \le 2\rho$. Define $\omega(b)=G_\mu(\ell+b y)$ for $b\in[0,1]$. By the mean value theorem there exists $b_y\in(0,1)$ such that
\begin{equation*}
G_\mu(y+\ell)-G_\mu(\ell)=\omega'(b_y)=\nabla G_\mu(\ell+b_y y)\cdot y.
\end{equation*}
Setting $\xi(y)=b_y y$ gives $|\xi(y)|<|y|\le 2\rho$, hence $\xi(y)\in B_{2\rho}(0)$, and thus
\begin{equation}\label{sec5:eq:taylor Gmu}
G_\mu(y+\ell)=G_\mu(\ell)+\nabla G_\mu(\ell+\xi(y))\cdot y.
\end{equation}
For $y\in A_\rho(0)$ define
\begin{equation*}
    q_L(y) =\nabla G_\mu(\ell+\xi(y))\cdot y.
\end{equation*}
For such $y$ we have $|\ell+\xi(y)|=L+O(\rho)=L+O(L^\theta)\sim L$. Using
\begin{equation*}
\nabla G_\mu(z)=-\mu^2\frac{K_2(\mu|z|)}{|z|}\frac{z}{|z|}, \qquad \mu=\mu_0/L,
\end{equation*}
and $\mu|\ell+\xi(y)|\geq \mu_0/2$ for $L>0$ large enough, we obtain
\begin{equation*}
|\nabla G_\mu(\ell+\xi(y))|\lesssim \frac{K_2(\mu_0/2)\mu^2}{|\ell+\xi(y)|}\sim \frac{1/L^2}{L}=L^{-3}.
\end{equation*}
Therefore, by definition
\begin{equation}\label{sec5:eq:Abschätzung_qL}
\abs{q_L(y)}=|\nabla G_\mu(\ell+\xi(y))\cdot y|\lesssim L^{-3}|y|\lesssim L^{-3+\theta}.
\end{equation}
uniformly in $y\in A_\rho(0)$. Inserting the expansion \eqref{sec5:eq:taylor Gmu} into \eqref{sec5:eq:insert_taylor} yields
\begin{equation}\label{sec5:eq:h2_scp_taylor}
    \begin{split}
        &h_2(\ell,y+\ell/2)\cdot \nabla G_\mu(y) =-\frac{1}{\rho}G_\mu(\ell)\mu^2\frac{K_2(\mu t)}{t} - \frac{1}{\rho}\mu^2\frac{K_2(\mu t)}{t} q_L(y).
    \end{split}
\end{equation}
We now integrate over $A_\rho(0)$ to obtain the desired expression. We begin with the remainder term using spherical coordinates in $\R^4$, the estimate \eqref{sec5:eq:Abschätzung_qL} and the Bessel series expansion \eqref{app1:eq:K-expansions} $s^2K_2(s)=2 + O(s^2)$ yields together
\begin{equation}\label{sec5:eq:h2_scp_main_term}
    \begin{split}
        &\frac{1}{\rho}\int_{A_\rho(0)}\mu^2\frac{K_2(\mu t)}{t} \abs{q_L(y)}\,dy \\
        &\leq 2\pi^2\frac{\sup_{y\in A_\rho(0)}|q_L(y)|}{\rho}\int_\rho^{2\rho}\mu^2 t^2 K_2(\mu t)\,dt \in O(L^{-3+\theta})
    \end{split}
\end{equation}
It remains to compute the integral of the leading term in \eqref{sec5:eq:h2_scp_taylor}. Integrating over $A_\rho(0)$ yields with $s=\mu t$ 
\begin{equation*}
    \frac{1}{\rho}G_\mu(\ell)\mu^2\int_{A_\rho(0)} \frac{K_2(\mu\abs{y})}{\abs{y}} dy = \frac{2\pi^2}{\mu \rho}G_\mu(\ell) \int_{\mu \rho}^{2\mu \rho} s^2 K_2(s) ds.
\end{equation*}
Using $s^2K_2(s)=2+O(s^2)$ and $a:=\mu\rho=\mu_0L^{\theta-1}$, we have
\begin{equation*}
\int_{a}^{2a}s^2K_2(s)\,ds
=\int_a^{2a}\big(2+O(s^2)\big)\,ds
=2a+O(a^3).
\end{equation*}
Therefore,
\begin{equation*}
\frac{2\pi^2}{\mu\rho}G_\mu(\ell)\int_{\mu\rho}^{2\mu\rho}s^2K_2(s)\,ds
=
\frac{2\pi^2}{a}G_\mu(\ell)\big(2a+O(a^3)\big)
=
4\pi^2G_\mu(\ell)+O\!\big(G_\mu(\ell)a^2\big).
\end{equation*}
Since 
\begin{equation*}
    G_\mu(\ell)=\frac{\mu K_1(\mu L)}{L}=\frac{\mu_0}{L^2}K_1(\mu_0), \qquad a^2=\mu_0^2L^{2\theta-2},
\end{equation*} 
this becomes
\begin{equation*}
4\pi^2G_\mu(\ell)+O\!\big(L^{-2}\cdot L^{2\theta-2}\big)
= 4\pi^2\mu_0 K_1(\mu_0) L^{-2}+O(L^{-4+2\theta}).
\end{equation*}
To summarize the above: We integrate \eqref{sec5:eq:h2_scp_taylor} over $A_\rho(0)$ and apply \eqref{sec5:eq:h2_scp_main_term} to find
\begin{equation}
     \scp{h_{2}(\ell,\cdot+\ell/2)}{\nabla G_\mu}_{L^2(A_\rho(0))} = \frac{1}{\rho}G_\mu(\ell)\mu^2\int_{A_\rho(0)} \frac{K_2(\mu\abs{y})}{\abs{y}} dy + O(L^{-3+\theta}).
\end{equation}
Using the subsequent estimates on the remaining integral and comparing the result with the constant $D_2(\mu_0)$ in \eqref{sec5:eq:D2} proves the claim.

\medskip
\noindent\textbf{\underline{Claim 3}:}
Let $L>0$ be sufficiently large. Then
\begin{equation}\label{sec5:eq:Crossterm h_R,3 nabla G mu}
\scp{h_{3}(\ell,\cdot)}{\nabla G_\mu(\cdot-\ell/2)}_{L^2(A_\rho(\ell/2))}
= D_3(\mu_0) L^{-2}+O(L^{-3+\theta}),
\end{equation}
where
\begin{equation*}
D_3(\mu_0)=2\pi^2\mu_0^2\big(2\log(2)-1\big).
\end{equation*}
\noindent\textbf{Proof of Claim 3:}
Set $y = x - \ell/2$ and $t = |y|$. Then if $x\in A_\rho(\ell/2)$ we have $y\in A_\rho(0)$.
In these coordinates we have
\begin{equation*}
    \begin{split}
            h_{3}(\ell,y+\ell/2)&=\left(2-\frac{t}{\rho}\right)\left(\mu^2\frac{K_2\left(\mu t\right)}{t}-\frac{2}{t^3}\right)\frac{y}{t}=:\widetilde h_3(y), \\
    \nabla G_\mu(y)
&=-\mu^2\frac{K_2(\mu t)}{t}\frac{y}{t}
    \end{split}
\end{equation*}
and consequently
\begin{equation*}
\widetilde h_3(y) \cdot \nabla G_\mu(y) =-\left(2-\frac{t}{\rho}\right)\left(\mu^2\frac{K_2(\mu t)}{t}-\frac{2}{t^3}\right)\left(\mu^2\frac{K_2(\mu t)}{t}\right).
\end{equation*}
Integrating over $A_\rho(0)$ and substitution $s=\mu t$ yields
\begin{equation*}
    \langle \widetilde{h}_3, \nabla G_\mu \rangle_{L^2(A_\rho(0))} =2\pi^2 \frac{\mu}{\rho} \int_{\mu\rho}^{2\mu\rho}\left(2\mu\rho-s\right)(2-s^2K_2(s))\frac{K_2(s)}{s} ds.
\end{equation*}
Applying the expansion of Bessel functions in \eqref{app1:eq:(2-z^2K2)K2} we find
\begin{equation*}
    (2-s^2 K_2(s))K_2(s)\frac{1}{s} = \frac{1}{s} + q(s), \quad q \in O(s\log (s)).
\end{equation*}
Defining $a:=\mu\rho=\mu_0L^{\theta-1}$, we have
\begin{equation} \label{sec4:eq:prop3_expanded}
   \frac{\mu}{\rho} \int_a^{2a} (2a-s)(2-s^2 K_2(s))K_2(s)\frac{1}{s} ds =\frac{\mu}{\rho}\int_a^{2a} \frac{2a-s}{s} ds + \frac{\mu}{\rho}\int_a^{2a} (2a-s)q(s) ds.
\end{equation}
Since \(0 \le 2a - s \le a\) for \(s \in [a,2a]\), we estimate
\begin{equation*}
\left|\int_a^{2a} (2a-s) q(s)\, ds \right|
\le a \int_a^{2a} |q(s)|\, ds
\lesssim a \int_a^{2a} s |\log s|\, ds
\lesssim a^3 |\log a|
\end{equation*}
and since $\mu=\mu_0L^{-1}$ and $\rho=L^\theta$ we find
\begin{equation} \label{sec5:eq:prop3:rem_int}
    \abs{\frac{\mu}{\rho}a^3\log(a)} \lesssim L^{-4+2\theta}\log(L) \in O(L^{-3+\theta}).
\end{equation}
Solving the first integral in the right-hand side of \eqref{sec4:eq:prop3_expanded} yields with  $a=\mu\rho$ 
\begin{equation} \label{sec5:prop3:main_int}
    \frac{\mu}{\rho}\int_a^{2a} \frac{2a-s}{s} \, ds
    =\mu^2\bigl(2\log(2)-1\bigr).
\end{equation}
Combining \eqref{sec5:prop3:main_int} and \eqref{sec5:eq:prop3:rem_int} shows with $\mu=\mu_0/L$
\begin{equation} \label{sec4:eq:prop3_int_solved}
   2\pi^2\frac{\mu}{\rho} \int_a^{2a} (2a-s)(2-s^2 K_2(s))K_2(s)\frac{1}{s} ds =2\pi^2 \mu_0^2\bigl(2\log(2)-1\bigr)L^{-2} +O(L^{-3+\theta}).
\end{equation}
Finally, the change of variables $y=x-\ell/2$ maps $A_\rho(\ell/2)$ onto $A_\rho(0)$,
which yields \eqref{sec5:eq:Crossterm h_R,3 nabla G mu}.

\medskip
\noindent\textbf{\underline{Claim 4}:} Let $L>0$ be large enough, then
\begin{equation}\label{sec5:eq:Crossterm h_R,4 nabla G mu}
    \scp{h_{4}(\ell,\cdot)}{\nabla G_\mu(\cdot-\ell/2)}_{L^2(A_\rho(\ell/2))}\in O(L^{-3+\theta}).
\end{equation}
\noindent\textbf{Proof of Claim 4:}
Set $y=x-\ell/2$ and $t=|y|$. For $y\in A_\rho(0)$ we have
\begin{equation} \label{sec5:eq:h_4 and G_mu'}
        h_{4}(\ell,y+\ell/2)
        =\left(2-\frac{|y|}{\rho}\right)\left(\mu^2\frac{K_2\!\left(\mu|y+\ell|\right)}{|y+\ell|} \right)\frac{y+\ell}{|y+\ell|},
        \qquad
        \nabla G_\mu(y)=-\mu^2\frac{K_2(\mu t)}{t}\frac{y}{t}.
\end{equation}
Since $t\in[\rho,2\rho]$, we have $\bigl|2-\frac{|y|}{\rho}\bigr|\le 2$.
Moreover, by \eqref{app1:lem:Kj_expansion} (in particular, $K_2(z)\lesssim z^{-2}$ for $z>0$) we obtain for any $p>0$
\begin{equation}\label{sec5:eq:K2_pointwise_bound}
    \mu^2\frac{K_2(\mu p)}{p}\lesssim \mu^2\frac{(\mu p)^{-2}}{p}=p^{-3}.
\end{equation}
We note that due to \eqref{sec5:eq:K2_pointwise_bound} we find
\begin{equation} \label{sec5:eq:G_mu'1norm}
    \norm{\nabla G_\mu(\cdot - \ell/2)}_{L^1(A_\rho(\ell/2))} =\mu^2\frac{K_2(\mu|y|)}{|y|}\,dy \lesssim  L^{-3}\int_{A_\rho(0)} |y|^{-3}\,dy \lesssim L^\theta
\end{equation}
Using that $u_\rho(y) = (2-\abs{y}/\rho)\leq 1$ and $\mu\abs{y+l}\in O(1)$ on $A_\rho(0)$ we conclude from \eqref{sec5:eq:h_4 and G_mu'} that
\begin{equation}\label{sec5:eq:h4_Lminus3}
    |h_4(\ell,y+\ell)| \lesssim \mu^2|y+\ell|^{-1}
    \lesssim L^{-3}, \quad y \in A_\rho(0)
\end{equation}
for $L>0$ sufficiently large.
The last step above follows by the inverse triangle inequality,
\begin{equation}\label{sec5:eq:ell_separation}
    |y+\ell|\ge |\ell|-|y|
    \ge |\ell|-2\rho
    \ge \frac{|\ell|}{2}.
\end{equation}
Using Hölder's inequality we find
\begin{equation}\label{sec5:eq:L1_estimate_nablaGmu_Anulus_second_step}
    \left|\scp{h_{4}(\ell,\cdot)}{\nabla G_\mu(\cdot-\ell/2)}_{L^2(A_\rho(\ell/2))}\right|
    \lesssim \rho\,L^{-3}.
\end{equation}
Finally, by using $\rho=L^\theta$, we conclude the claim.\\
\noindent\textbf{\underline{Claim 5}:} Let $\theta=4/(4+\delta)>3/(3+\delta)$ and $L>0$ sufficiently large. Then
\begin{equation*}
    \scp{h_{5}(\ell,\cdot)}{\nabla G_\mu(\cdot-\ell/2)}_{L^2(A_\rho(\ell/2))} \in O(L^{-3+\theta}).
\end{equation*}
\textbf{Proof of Claim 5:} Set $y=x-\ell/2$ and $t=|y|$. For $y\in A_\rho(0)$ we have
\begin{equation}\label{sec5:eq:widetildeh5}
    \widetilde h_5(y) = h_5(\ell, y + \ell/2) = -L^{-\theta}\epsilon_1(y)\frac{y}{t}+u_\rho(y) \epsilon_2(y),  \qquad
        \nabla G_\mu(y)=-\mu^2\frac{K_2(\mu t)}{t}\frac{y}{t}.
\end{equation}
To prove the claim we analyze 
\begin{equation*}
    \scp{h_5(\ell,\cdot)}{\nabla G_\mu(\cdot-\ell/2)}_{L^2(A_\rho(\ell/2))} = \scp{\widetilde h_5}{\nabla G_\mu}_{L^2(A_\rho(0))}
\end{equation*}
Due to the bounds on $\epsilon_1$ and $\epsilon_2$ in Lemma \ref{sec5:lem:AbfallverhaltenResonanz} we know for $L>0$ large enough and $y\in A_\rho(0)$ together with $u_\rho \leq 1$ we derive from \eqref{sec5:eq:widetildeh5}
\begin{equation*}
    \begin{split}
            \abs{\widetilde h_5(y)} \lesssim L^{-\theta} t^{-2-\delta} + t^{-3-\delta}.
    \end{split}
\end{equation*}
Due to the fact, that $t=\abs{y}>L^{\theta}$ for $y\in A_\rho(0)$ we see
\begin{equation*}
    \abs{\widetilde h_5(y)}\lesssim L^{-\theta(3+\delta)}.
\end{equation*}
Since $\theta=4/(4+\delta)> 3/(3+\delta)$, it follows $\abs{\widetilde h_5(y)}\lesssim L^{-3}$. Using \eqref{sec5:eq:G_mu'1norm} together with Hölder's inequality we find
\begin{equation*}
    \abs{\scp{h_5(\ell,\cdot)}{\nabla G_\mu(\cdot-\ell/2)}_{L^2(A_\rho(\ell/2))}}= \abs{\scp{\widetilde h_5}{\nabla G_\mu}_{L^2(A_\rho(0))}}\lesssim L^{-3} L^\theta \in O(L^{-3+\theta}).
\end{equation*}
This concludes claim 5.\\

\medskip
Note that the constants $D_1(\mu_0), D_2(\mu_0)$ and $D_3(\mu_0)$ do satisfy \eqref{sec5:eq:aufsummieren D_j}. As explained earlier, proving claims 1--5 completes the proof of Proposition~\ref{sec5:prop:nabla Annulus Terme}. 
\end{proof}

\subsubsection{Proof of 
Proposition~\ref{sec5:prop:nabla G mu Terme}}\label{sec5:subsec:prop3}
      \begin{proof}
        The proof proceeds in two steps. In the first step, we show that
            \begin{equation*}
                \norm{\nabla G_\mu(\cdot-\ell/2)}^2_{L^2(\Omega_\rho(\ell/2)} = \norm{\nabla G_\mu}^2_{L^2(\R^4\setminus B_\rho(0))} + O(L^{-3+\theta})
            \end{equation*}
            and evaluate the norm on the right side. In the second step, we show
            \begin{equation}\label{sec5:eq:prop3:remaining_int}
                \begin{split}
                &\scp{\nabla G_\mu(\cdot-\ell/2)}{\nabla G_\mu(\cdot+\ell/2)}_{L^2(\Omega_\rho(\ell/2))}\\
                = &\int_{\R^4} \nabla G_\mu(x-\ell/2) \cdot \nabla G_\mu(x+\ell/2)\, dx + O(L^{-3+\theta})\, .
                \end{split}
            \end{equation}
            The remaining integral in the right-hand side of \eqref{sec5:eq:prop3:remaining_int} will be solved in the Appendix. We give that as Lemma \ref{app2:lem:Faltung Ableitung Gmu mit sich selbst}.
            Recall that for any $\ell \in \R^4$ with $L=|\ell|$ we always have
            \begin{equation*}
               \rho = L^{\theta}, \quad \mu= \mu_0 L^{-1},
            \end{equation*}
            for some $\theta\in (0,1)$ and $\mu_0>0$ to be chosen later.
            
            \medskip
            \noindent \textbf{Step 1: } Applying the expansion of Bessel functions in \ref{app1:lem:Kj_expansion} we find
            \begin{equation*}
                \abs{\nabla G_\mu(y)}=\mu^2\frac{K_2(\mu \abs{y})}{\abs{y}} \lesssim \abs{y}^{-3}
            \end{equation*}
            for any $y\in B_\rho(0)$. Thus for any $x\in B_\rho(-\ell/2)$
            \begin{equation}
                \abs{\nabla G_\mu(x-\ell/2)} \lesssim \abs{x-\ell/2}^{-3} \lesssim L^{-3} \, .\label{sec5:eq:Abschätzung nabla Gmu auf anderem Ball}
            \end{equation}
            Since $\rho=L^\theta$ we have $\abs{B_\rho(-\ell/2)}\lesssim L^{4\theta}$ and therefore
            \begin{equation}
                \|\nabla G_\mu(\cdot-\ell/2)\|^2_{L^2(B_\rho(-\ell/2))}\lesssim L^{-6+4\theta} \in O(L^{-3+\theta}).\label{sec5:eq:Norm Nabla Gmu +R klein auf Ball um R}
            \end{equation}
            Note that $\Omega_\rho(\ell/2) = \R^4 \setminus (B_\rho(\ell/2) \cup B_\rho(-\ell/2))$ and due to \eqref{sec5:eq:Norm Nabla Gmu +R klein auf Ball um R} the function $\nabla G_\mu$ is small on $B_\rho(-\ell/2)$ and consequently
            \begin{equation}
                \|\nabla(G_\mu(\cdot-\ell/2)\|^2_{L^2(\Omega_{\rho}(\ell/2))} = \|\nabla G_\mu\|^2_{L^2(\R^4\setminus B_\rho(0))} + O(L^{-3+\theta}).\label{sec5:eq:Ball dazunehmen liefert kleinen Fehler bei Norm nabla Gmu}
            \end{equation}
            A direct calculation shows
            \begin{equation*}
                \|\nabla G_\mu\|^2_{L^2(\R^4\setminus B_\rho(0))} = 2\pi^2\int_\rho^\infty \abs{\mu^2 \frac{K_2(\mu r)}{r}}^2 r^3 dr = 2\pi^2 \mu^2 \int_{\mu\rho}^\infty sK_2(s)^2 ds.
            \end{equation*}
            We compute the antiderivative involving $K_2$ in Remark~\ref{app:cor:bessel_aux} explicitly and thus arrive at
            \begin{align*}
                \|\nabla G_\mu\|^2_{L^2(\R^4\setminus B_\rho(0))}= 2\pi^2\mu^2\left[\frac{1}{2}s^2\left(K_2(s)^2-K_1(s)K_3(s)\right)\right]_{s=\mu\rho}^\infty. %\\
                %&=2\pi^2 \mu^2\left(-\frac{1}{2}(\mu\rho)^2\left(K_2((\mu\rho))^2-K_1((\mu\rho))K_3((\mu\rho))\right) \right).
            \end{align*}
            Applying the Bessel series expansion in \eqref{app1:eq:K1K3-K2^2} together with $\rho=L^\theta $ and $ \mu=\mu_0 L^{-1} $, yields
            \begin{equation*}
            \begin{split}
                \|\nabla G_\mu\|^2_{L^2(\R^4\setminus B_\rho(0))} & =\frac{1}{2}\left(4\pi^2 L^{-2\theta}+C_2(\mu_0)\log(\mu_0L^{\theta-1})L^{-2}+C_3(\mu_0) L^{-2} \right)\\
                 &+ O(L^{-3+\theta})
            \end{split}
            \end{equation*}
            where $C_2(\mu_0),C_3(\mu_0)$ are as in \eqref{sec5:eq:Konstanten Norm nabla Gmu}. Due to \eqref{sec5:eq:Ball dazunehmen liefert kleinen Fehler bei Norm nabla Gmu} we get 
            \begin{equation*}
                 2\|\nabla G_\mu(\cdot-\ell/2)\|^2_{L^2(\Omega_{\rho}(\ell/2))} =4\pi^2 L^{-2\theta}+C_2(\mu_0)\log(\mu_0L^{\theta-1})L^{-2}+C_3(\mu_0) L^{-2}+ O(L^{-3+\theta}).
            \end{equation*}
            This finishes the proof of \eqref{sec5:eq:Ergebnis Norm nabla Gmu}.\\

\medskip
            \noindent\textbf{Step 2:} First, we show 
            \begin{equation*}
                \int_{B_\rho(-\ell/2)}\nabla G_\mu(x-\ell/2)\nabla G_\mu(x+\ell/2)\, dx \in O(R^{-3+\theta}).
            \end{equation*}
            Due to \eqref{sec5:eq:Abschätzung nabla Gmu auf anderem Ball} we see
            \begin{equation*}
                \abs{\int_{B_\rho(-\ell/2)}\nabla G_\mu(x-\ell/2)\nabla G_\mu(x+\ell/2)\, dx}\lesssim L^{-3}\int_{B_\rho(-\ell/2)}\abs{\nabla G_\mu(x+\ell/2)}\, dx.
            \end{equation*}
            After a translation and a transformation to spherical coordinates it follows
            \begin{equation*}
                 \abs{\int_{B_\rho(\ell/2)} \nabla G_\mu(x+\ell/2) \nabla G_\mu(x-\ell/2) dx} \lesssim  L^{-3} \int_0^\rho \mu^2 r^2 K_2(\mu r)\, dr.
            \end{equation*}
            Since $\mu=\mu_0/L, \rho =L^\theta$ and $s^2K_2(s)\in O(1)$ for $s\ll 1$ this implies
            \begin{equation*}
                \abs{\int_{B_\rho(\ell/2)} \nabla G_\mu(x+\ell/2) \nabla G_\mu(x-\ell/2) dx} \lesssim  L^{-2} \int_0^{\mu\rho} s^2K_2(s)\, ds \in O(L^{-3+\theta}).
            \end{equation*}
            By symmetry the same statement holds for the integral over $B_\rho(-\ell/2)$. Consequentially we obtain
            \begin{equation}\label{sec5:eq:Cross Term nabla Gmu auf R4 erweitern}
            \begin{split}
                 &\scp{\nabla G_\mu(\cdot+\ell/2)}{\nabla G_\mu(\cdot-\ell/2)}_{L^2(\Omega_{\rho}(\ell/2))}\\
                =&\int_{\R^4} \nabla G_\mu(x+\ell/2)\cdot \nabla G_\mu(x-\ell/2) dx+ O(L^{-3+\theta}).
            \end{split}
            \end{equation}
            We compute the integral above in the appendix in Lemma \ref{app2:lem:Faltung Ableitung Gmu mit sich selbst} which together with $\mu=\mu_0/L$ gives
            \begin{equation*}
                \int_{\R^4} \nabla G_\mu(x+\ell/2)\cdot \nabla G_\mu(x-\ell/2) dx=%\scp{\nabla G_\mu(\cdot+\ell/2)}{\nabla G_\mu(\cdot-\ell/2)}_{L^2(\R^4)}=
                \sum_{j=1}^42\pi^2\left(\mu_0K_1(\mu_0) L^{-2}-\mu_0^2 K_2(\mu_0) L^{-2}\frac{\ell_j^2}{L^2}\right).
            \end{equation*}
            It follows
            \begin{equation*}
                2\scp{\nabla G_\mu(\cdot+\ell/2)}{\nabla G_\mu(\cdot-\ell/2)}_{L^2(\Omega_{\rho}(\ell/2))}= C_4(\mu_0) L^{-2} + O(L^{-3+\theta}).
            \end{equation*}
            where $C_4(\mu_0)$ is as in \eqref{sec5:eq:Konstanten Norm nabla Gmu}. This finishes the proof of \eqref{sec5:eq:cross Term Gmu} and completes the proof of Proposition~\ref{sec5:prop:nabla G mu Terme}. 
        \end{proof}
%%%%%%%%%%%%%%%%%%%%%%%%%%%%%%%%%%%%%%%%
%%%%%%%%%%% SUBSECTION KINETIC AUX
%%%%%%%%%%%%%%%%%%%%%%%%%%%%%%%%%%%%%%%%%
\subsection{Kinetic Control}\label{sec5:sub:aux}
In this subsection we establish the estimate on the kinetic terms in \eqref{sec3:eq:for_born_op} from Lemma~\ref{sec2:lem:proto_efimov} by proving Lemma~\ref{sec5:lem:aux}. Lemma~\ref{sec5:lem:aux} is the key input for the adiabatic reduction in the proof of
Theorem~\ref{sec2:thrm:main}. The estimate controls the defect between
the kinetic energy associated with translations of the potential wells in the $(x_3,x_4)$-plane
and the overall kinetic energy in that two-dimensional direction. In this subsection we choose $\ell(r)\in \R^4$ with $\ell(r) = (0,0,r)$ where $r\in \R^2$ such that $L=|\ell(r)|=|r|$ and as before
\begin{equation*}
    \mu= \mu_0/L, \quad \rho = L^\theta
\end{equation*}
for some $\mu_0>0$ and $\theta\in (0,1)$ to be chosen later.

Before turning to the proof, we briefly explain the structure of the
argument. The above defect does not vanish identically, since the
Schrödinger $H[\ell(r)]$ is not invariant under translations in the
$(x_3,x_4)$–plane due to its genuine double–well structure.
However, the test function $\phi(\ell(r),\cdot)$ is constructed such that it
coincides locally with the zero–energy solution of the corresponding
single–well operator in a neighborhood of each individual potential
well, which is indeed invariant under translations in the $(x_3,x_4)$-plane. As a consequence, these local contributions to the kinetic
expression vanish. Meaning for $x=(x_1,x_2,x_3,x_4) \in \R^4$ and $\varphi_0$ the zero-energy solution
\begin{equation}\label{sec5:eq:chainrule}
    P_r^2 \varphi_0( x \pm \ell(r)/2 ) = \frac{1}{4}P_{(x_3,x_4)}^2 \varphi_0(x\pm \ell(r)/2 ) \, .
\end{equation}

It follows that the leading contribution to the kinetic estimate
originates entirely from configurations far away from both wells,
namely from the region $\Omega_\rho$, where $\phi(\ell,\cdot)$ is a superposition of shifted Green functions $G_\mu$ (namely the function $\Gamma$). In
the present section we reduce the kinetic estimate to this far–field
regime and identify the relevant contributions. The explicit analysis
of the far–field terms involves a detailed study of Bessel functions
and is formulated as the independent
Proposition~\ref{sec5:prop:Ableitung nach R von G mu Termen} below. The proof of that
proposition is deferred to
Section~\ref{subsec64: Energy außerhalb terme}.

\begin{proposition}\label{sec5:prop:Ableitung nach R von G mu Termen}
        For $\ell \in \R^4$ with $\ell= (0,0,r)$ where $r\in \R^2$ and $L=|\ell(r)|=\abs{r}$ large enough let $\mu=\mu_0/L$ and $\Gamma(\ell, \cdot )$ be the function which is defined in \eqref{sec5:eq:def_Gamma&f} as
        \begin{equation*}
            \Gamma(\ell(r),x ) = G_\mu(x+\ell(r)/2) +G_\mu(x-\ell(r)/2).
        \end{equation*}
        Then 
        \begin{equation*}
          \left \langle \Gamma(\ell(r),\cdot ) , \left(P_r^2-\frac{1}{4}P_{(x_3,x_4)}^2\right)  \Gamma(\ell(r),\cdot  )  \right\rangle_{L^2(\Omega_\rho)}  =   C_5(\mu_0) L^{-2}+O(L^{-3+\theta})
        \end{equation*}
        where
        \begin{align*}
            C_5(\mu_0) =(4\pi^2/3)\left(1+\mu_0 K_1(\mu_0)+\tfrac{7}{2}\mu_0^2 K_0(\mu_0)\right).
        \end{align*}
    \end{proposition}
        \begin{proof}
            We defer the proof of Proposition \ref{sec5:prop:Ableitung nach R von G mu Termen} to Section \ref{subsec64: Energy außerhalb terme}.
        \end{proof}
We will give now the proof of Lemma~\ref{sec5:lem:aux} under the assumption of Proposition~\ref{sec5:prop:Ableitung nach R von G mu Termen}.
\begin{proof}[Proof of ~\ref{sec5:lem:aux}]
    The proof proceeds in three steps. In the first step we prove that
    \begin{equation} \label{sec5:eq:kin_est_first_step}
        \begin{split}
            &\left \langle \widetilde\Phi(\ell(r),\cdot), \left(P_r^2-\frac{1}{4}P_{(x_3,x_4)}^2\right) \widetilde\Phi(\ell(r),\cdot)  \right\rangle \\
            &\leq \norm{\phi(\ell(r),\cdot)}^{-2}  \int_{\R^4} \left( \abs{P_r \phi(\ell(r),x)}^2 - \frac{1}{4} \abs{P_{(x_3,x_4)} \phi(\ell(r),x)}^2 \right) dx 
        \end{split}
    \end{equation}
    and in the second step we estimate the integral right-hand side of \eqref{sec5:eq:kin_est_first_step} over the regions in the vicinity of the individual wells, meaning the sets $A_\rho(\pm \ell/2)$ and $B_\rho(\pm\ell/2)$. In the final third step we apply Proposition~\ref{sec5:prop:Ableitung nach R von G mu Termen} to control the contribution from regions far outside the potential wells and conclude the statement.

\medskip
\textbf{Step 1:}
    The function $\widetilde \Phi(\ell(r),\cdot) $ is normalized for any fixed $r\in \R^2$ meaning that first a function $\phi(\ell(r),\cdot)$ was defined for fixed $r\in \R^2$ as function in $L^2(\R^4)$ and then normalized such that $\widetilde \Phi(\ell(r),\cdot) = \phi(\ell(r),\cdot)/\norm{\phi(\ell(r),\cdot)}$ where $\norm{\cdot }$ denotes the $L^2(\R^4,dx)$ norm. Thus, $\norm{\phi(\ell(r),\cdot)}$ is a function depending on $r\in \R^2$. We define for $r\in \R^2$
    \begin{equation*}
        N(r) = \norm{\phi(\ell(r),\cdot)},
    \end{equation*}
    and by abuse of notation, we write $\phi(\ell(r),\cdot)=\phi(r,\cdot)$.
    
    Following Remark~\ref{sec5:rem:weak_l_der} the mapping $r\mapsto N(r)$ is in $H^{1}_{\loc}(\R^2)$ and due to the quotient rule we find for any $r\in \R^2\setminus\{0\}$ and $x\in \R^4$
    \begin{equation*}
        P_r \widetilde \Phi(\ell(r),x)  = \left((P_r \phi(r,x))N(r)- \phi(r,x) (P_rN(r))\right)N(r)^{-2}
    \end{equation*}
     Note that both $N(r)$ and $P_r N(r)$ are independent of $x$ and consequently expanding the square and integrating against $x$ yields
    \begin{equation*}
        \begin{split}
            \norm{P_r \widetilde \Phi(\ell(r),\cdot)  }^2 = N(r)^{-4}\Bigg{ (}&N(r)^{2} \int_{\R^4}\abs{P_r \phi(r,x)}^2 dx  + (P_r N(r))^2\int_{\R^4} \abs{ \phi(\ell(r), x)}^2 dx\\
            &-2\left[ N(r)(P_r N(r)) \int_{\R^4} \phi(r,x) P_r \phi(r,x)  dx \right] \Bigg{ )}. 
        \end{split}
    \end{equation*}
     Using the definition of $N(r)$ together with the fact that $N(r)$ and  $P_r N(r)$ are independent of $x$ we find
        \begin{equation} \label{sec5:eq:r_der_square_compl.}
        \begin{split}
            \norm{P_r \widetilde \Phi(\ell(r),\cdot)  }^2 
            =N(r)^{-4}\Bigg{ (}&N(r)^{2} \int_{\R^4}\abs{P_r \phi(r,x)}^2 dx +   (P_r N(r))^2 N(r)^2\\
            &-N(r)(P_r N(r)) 2 \int_{\R^4} \phi(r,x) P_r \phi(r,x)  dx  \Bigg{ )} 
        \end{split}
    \end{equation}
    Since $P_r \abs{\phi(r,x)}^2 = 2 \phi(r,x) P_r \phi(r,x)$ we find 
    \begin{equation} \label{sec5:eq:r_der_out}
        \begin{split}
            2 \int_{\R^4} \phi(r,x) P_r \phi(r,x)  dx  = P_r N(r)^2 = 2 N(r) \left( P_r N(r) \right). 
        \end{split}
    \end{equation}
    Inserting \eqref{sec5:eq:r_der_out} into \eqref{sec5:eq:r_der_square_compl.} we arrive at
        \begin{equation} \label{sec5:eq:r_der_simpl}
        \begin{split}
            \norm{P_r \widetilde \Phi(\ell(r),\cdot)  }^2 
            &=N(r)^{-2} \left( \int_{\R^4}\abs{P_r \phi(r,x)}^2 dx -   (P_r N(r))^2 \right) \\
            &\leq N(r)^{-2} \int_{\R^4}\abs{P_r \phi(r,x)}^2 dx \, .
        \end{split}
    \end{equation}
    Consequently, inserting \eqref{sec5:eq:r_der_simpl} into  the left-hand side of\eqref{sec5:eq:kin_est_first_step} concludes the proof of Step 1.

\medskip
\textbf{Step 2:} In this step we estimate the right-hand side of \eqref{sec5:eq:kin_est_first_step} on the sets $A_\rho(\pm \ell/2)$ and $B_\rho(\pm\ell/2)$ and show
\begin{equation}
    \begin{split}
    \norm{P_r\phi(r,\cdot)}^2 -\frac{1}{4} \norm{P_{(x_3,x_4)}\phi(r,\cdot)}^2  =&\norm{P_r\Gamma(\ell(r),\cdot) }^2_{L^2(\Omega_\rho(\ell/2))} - \frac{1}{4} \norm{P_{(x_3,x_4)}\Gamma(\ell(r),\cdot)}^2_{L^2(\Omega_\rho(\ell/2))} \label{sec5:eq:kin_est_second_step}\\
    &+ O(L^{-3+\theta}\log(L)).
    \end{split}
\end{equation}
We begin with the set $B_\rho(\pm \ell(r)/2)$. By definition of $\phi$ in \eqref{sec5:eq:Testfunktion auf Gebieten} we have
\begin{equation*}
    \phi(r,x) =\varphi_0(x\pm \ell(r)/2) \quad \text{for all } x\in B_\rho(\pm\ell(r)/2)\, .
\end{equation*}
Consequently, by the chain rule for derivatives we have 
\begin{equation*}
    \|P_r\phi(r,\cdot)\|^2_{L^2(B_\rho(\pm\ell/2))} - \frac{1}{4}\|P_{(x_3,x_4)}\phi(r,\cdot)\|^2_{L^2(B_\rho(\pm\ell/2))} = 0.
\end{equation*}
We continue with the integral over the annuli  $A_\rho(\pm \ell(r)/2)$. Let $A_\pm=A_\rho(\pm \ell(r)/2)$ we then show
\begin{equation}\label{sec5:eq:proof_kin_est_step_2_aim}
    \begin{split}
    &\norm{\nabla_r\phi(r,\cdot)}^2_{L^2(A_{\pm})} -\frac{1}{4} \norm{\nabla_{(x_3,x_4)}\phi(r,\cdot)}^2_{L^2(A_{\pm})}\\
            = &\norm{\nabla_r\Gamma(\ell(r),\cdot) }^2_{L^2(A_{\pm})} - \frac{1}{4} \norm{\nabla_{(x_3,x_4)}\Gamma(\ell(r),\cdot)}^2_{L^2(A_{\pm})} + O(L^{-3+\theta}\log(L)).
    \end{split}
\end{equation}
We conduct the proof for $A_+$, the proof for $A_-$ is similar. By the definition of $f$ and $\Gamma$ in \eqref{sec5:eq:def_Gamma&f} we have for any $r\in \R^2$ and $x\in A_\rho(\ell(r))$
\begin{equation}\label{sec5:eq:der_fR}
    P_r\phi(r,x) = P_rf(\ell(r),x) + P_r\Gamma(\ell(r),x).
\end{equation}
By abuse of notation we will write $f(\ell(r),x)=f(r,x)$ and $\Gamma(\ell(r),x)=\Gamma(r,x)$ in the following. 
Inserting \eqref{sec5:eq:der_fR} into the norms in left-hand side of \eqref{sec5:eq:proof_kin_est_step_2_aim} and completing the square yields
\begin{equation*}
    \begin{split}
         \norm{P_r\phi}^2_{L^2(A_{+})} &=\norm{P_r f}^2_{L^2(A_{+})} + \norm{P_r\Gamma}^2_{L^2(A_{+})}+ 2 \langle P_r f, P_r\Gamma\rangle_{L^2(A_{+})} ,\\
          \norm{P_{(x_3,x_4)}\phi}^2_{L^2(A_{+})} &=\norm{P_{(x_3,x_4)} f}^2_{L^2(A_{+})} + \norm{P_{(x_3,x_4)} \Gamma}^2_{L^2(A_{+})} + 2 \langle P_{(x_3,x_4)} f, P_{(x_3,x_4)}\Gamma\rangle_{L^2(A_{+})},
    \end{split}
\end{equation*}
where, by abuse of notation, we will write
\begin{equation*}
    \norm{P_r\phi}^2_{L^2(A_{+})} = \norm{P_r\phi(r,\cdot)}^2_{L^2(A_{+})}
\end{equation*}
and similar for the other norms and inner products. Consequently, to prove \eqref{sec5:eq:proof_kin_est_step_2_aim} it suffices to show
\begin{equation} \label{sec5:eq:suff1}
    \norm{P_r f}^2_{L^2(A_+)}-\frac{1}{4}\norm{P_{(x_3,x_4)}f}^2_{L^2(A_+)} \in O(L^{-3+\theta}\log(L))
\end{equation}
and
\begin{equation}\label{sec5:eq:suff2}
     \langle P_r f, P_r\Gamma \rangle_{L^2(A_{+})} - \frac{1}{4} \langle P_{(x_3,x_4)} f, P_{(x_3,x_4)}\Gamma\rangle_{L^2(A_{+})} \in O(L^{-3+\theta}\log(L)).
\end{equation}
Next we apply Lemma~\ref{app:lem:Anulus terme} which is valid for the choice  $\theta=4/(4+\delta)\in(3/(3+\delta),1)$ and states
\begin{equation} \label{sec5:eq:decay_fr}
    \abs{ P_x f(r,x)} \lesssim L^{-2-\theta} \log(L) \, \text{ and } \abs{ P_r f(r,x)} \lesssim L^{-2-\theta} \log(L) \quad x \in A_+
\end{equation}
The proof of Lemma~\ref{app:lem:Anulus terme} is deferred to Section \ref{subsec:6.1 interpolation}.

The set $A_+$ is an annulus centered around $\ell/2$ with inner radius $\rho = L^\theta$ and outer radius $2\rho$. Consequently, $\abs{A_+} \lesssim L^{4\theta}$ which together with the decay properties of $f$ above  immediately proves \eqref{sec5:eq:suff1}. We continue by proving \eqref{sec5:eq:suff2}. Recall that by definition $\mu=\mu(r)=\mu_0/L$ and
\begin{equation*}
            \Gamma(r,x ) = G_\mu(x+\ell(r)/2) +G_\mu(x-\ell(r)/2).
\end{equation*}
To show \eqref{sec5:eq:suff2} we discuss the following estimate for both functions $G_\mu(\cdot \pm \ell/2)$ independently, meaning we prove
\begin{equation}\label{sec5:eq:suff2_both_trans}
    \begin{split}
        &\langle P_r f,  P_r G_\mu(\cdot \pm \ell(r)/2) \rangle_{L^2(A_+)} - \frac{1}{4}\langle P_{(x_3,x_4)} f, P_{(x_3,x_4)}  G_\mu(\cdot \pm \ell(r)/2)\rangle_{L^2(A_{+})} \\
        &\in O(L^{-3+\theta}\log(L)).
    \end{split}
\end{equation}
The expression above involving $G_\mu(\cdot +\ell/2) $ is easier to treat since a short computation shows for $x\in B_{2\rho}(\ell(r)/2) $
\begin{equation*}
        \abs{P_{(x_3,x_4)} G_\mu(x+\ell(r)/2)} \lesssim L^{-3}\quad\text{ and }\quad \abs{P_r G_\mu(x+\ell(r)/2)} \lesssim L^{-3}.
\end{equation*}
The inequality \eqref{sec5:eq:suff2_both_trans} follows in this case from the estimate above together with \eqref{sec5:eq:decay_fr} and Schwarz inequality immediately. To prove \eqref{sec5:eq:suff2} it remains to show \eqref{sec5:eq:suff2_both_trans} for the term involving $G_\mu(\cdot -\ell(r)/2) $.

A direct calculation using the definition of $G_\mu$ shows
\begin{equation*}
  P_r G_\mu(x-\ell(r)/2) = -i \frac{\mu^2}{L}K_0\left(\mu\abs{ x-\ell(r)/2}\right)\frac{r}{L}-\frac{1}{2}P_{(x_3,x_4)}G_\mu(x-\ell(r)/2).
\end{equation*}
 We define $g(r,\cdot)\colon A_\rho(\ell/2)\to\R^2$ as
\begin{equation*}
    -i g(r,x) = \frac{1}{2}P_{(x_3,x_4)}f(r,x) + P_r f(r,x)
\end{equation*}
Then the left-hand side of \eqref{sec5:eq:suff2_both_trans} with $\mu = \mu_0/L$ reads
\begin{equation} \label{sec5:eq:step_2_gr_terms}
    \begin{split}
    &\scp{P_r f}{P_r G_\mu(\cdot-\ell(r)/2)}_{L^2(A_+)} - \frac{1}{4}\scp{P_{(x_3,x_4)}f}{P_{(x_3,x_4)}G_\mu(\cdot-\ell(r)/2)}_{L^2(A_+)}\\
    = &\frac{1}{2}i\scp{P_{(x_3,x_4)}f}{\mu_0^2L^{-3}K_0\left(\mu_0L^{-1}\abs{\cdot-\ell(r)/2}\right)\frac{r}{L}}_{L^2(A_+)}\\
    &-i\scp{g}{P_rG_\mu(\cdot-\ell(r)/2)}_{L^2(A_+)}.
    \end{split}
\end{equation}
Applying the series expansions of Bessel functions in \eqref{app1:eq:K-expansions} one can find on $A_+$
\begin{equation}\label{sec5:eq:decaypropertiesBesselAbleitungR:01}
    \begin{split}
    \abs{\frac{\mu_0^2}{L^3}K_0\left(\mu\abs{x-\ell(r)/2}\right)}&\lesssim L^{-3}\log(L),  \\
    \abs{P_rG_\mu(x-\ell(r)/2)}&\lesssim L^{-3\theta} .
    \end{split}
\end{equation}
For the choice $\theta=4/(4+\delta)>3/(3+\delta)$ we will prove in Lemma~\ref{app:lem:Anulus terme} the following estimates on the set $A_+$
\begin{equation}\label{sec5:eq:decaypropertiesBesselAbleitungR}
    \begin{split}
        \abs{g(r,x)}&\lesssim L^{-3}\log(L), \\
        \abs{P_{(x_3,x_4)}f(r,x)}&\lesssim L^{-2-\theta}\log(L).
    \end{split}
\end{equation}
The proof of Lemma~\ref{app:lem:Anulus terme} is deferred to Section~\ref{subsec:6.1 interpolation}. Applying the estimates in \eqref{sec5:eq:decaypropertiesBesselAbleitungR:01} and \eqref{sec5:eq:decaypropertiesBesselAbleitungR} together with the fact that $\abs{A_+}\lesssim L^{4\theta}$ to the right--hand side of \eqref{sec5:eq:step_2_gr_terms} yields \eqref{sec5:eq:suff2_both_trans} which concludes step 2. 

\medskip
\noindent\textbf{Step 3:} In this step we combine the result of Proposition \ref{sec5:prop:Ableitung nach R von G mu Termen} and the previous steps to finish the proof of Lemma~\ref{sec5:lem:aux}.
Combining \eqref{sec5:eq:kin_est_first_step} and \eqref{sec5:eq:kin_est_second_step} we arrive at
\begin{equation*}
    \begin{split}
        \left \langle \widetilde\Phi(\ell,\cdot), \left(P_r^2-\frac{1}{4}P_{(x_3,x_4)}^2\right) \widetilde\Phi(\ell,\cdot)  \right\rangle&\leq \norm{\phi(\ell(r),\cdot}^{-2}\bigg[\norm{P_r\Gamma }^2_{L^2(\Omega_\rho(\ell/2))}\\
        &- \frac{1}{4} \norm{P_{(x_3,x_4)}\Gamma}^2_{L^2(\Omega_\rho(\ell/2))} + O(L^{-3+\theta}\log(L))\bigg].%\label{sec5:eq:kin_est_second_step}
    \end{split}
\end{equation*}
Applying Proposition~\ref{sec5:prop:Ableitung nach R von G mu Termen} then yields
\begin{equation}\label{sec5:eq:almost_done_kin}
    \left \langle \widetilde\Phi(\ell,\cdot), \left(P_r^2-\frac{1}{4}P_{(x_3,x_4)}^2\right) \widetilde\Phi(\ell,\cdot)  \right\rangle \leq \norm{\phi(\ell(r),\cdot)}^{-2}\left(C_5(\mu_0) L^{-2} + O(L^{-3+\theta}\log(L))\right).
\end{equation}
According to Lemma~\ref{sec5:lem:Norm Testfunktion} we have
\begin{equation*}
    \norm{\phi(\ell(r),\cdot)}^{2} = 4\pi^2\log(L) + O(1)
\end{equation*}
and combining this with \eqref{sec5:eq:almost_done_kin} we arrive at
\begin{equation}\label{sec5:eq:kin_est_done}
    \left \langle \widetilde\Phi(\ell,\cdot), \left(P_r^2-\frac{1}{4}P_{(x_3,x_4)}^2\right) \widetilde\Phi(\ell,\cdot)  \right\rangle \leq \left((4\pi^2)^{-1}C_5(\mu_0) L^{-2}\log(L)^{-1} + O(L^{-3+\theta})\right).
\end{equation}
Finally we note that in the limit $\mu_0 \to 0^+$
\begin{equation*}
   (4\pi^2)^{-1}C_5(\mu_0) =\frac{1}{3}\left(1+\mu_0 K_1(\mu_0)+\tfrac{7}{2}\mu_0^2 K_0(\mu_0)\right) \to \frac{2}{3}
\end{equation*}
which follows from series expansion of Bessel functions in \eqref{app1:lem:Kj_expansion}. 
Recall that we aim to show and upper bound for the left--hand side of \eqref{sec5:eq:kin_est_done} of the form
$(2/3+\varepsilon)L^{-2}\log(L)^{-1}$ for some given fixed $\varepsilon>0$. Consequently, we choose $L_0(\varepsilon)>0$ large enough such that for any $L>L_0(\varepsilon)$ the contribution in $O(L^{-3+\theta})$ can only add a term of the size $(\varepsilon/2)L^{-2}\log(L)^{-1}$ and choose $\mu_0>0$ small enough such that
\begin{equation*}
    (4\pi^2)^{-1}C_5(\mu_0) =\frac{1}{3}\left(1+\mu_0 K_1(\mu_0)-\tfrac{7}{2}\mu_0^2 K_0(\mu_0)\right) \leq  \frac{2}{3} + \varepsilon/2
\end{equation*}
This concludes the proof of Lemma~\ref{sec5:lem:aux}.
\end{proof}

\section{Auxiliary Estimates for the Test Function}\label{sec:aux_lemmatas}
In this section we give the proof of Lemma~\ref{app:lem:Anulus terme} and Proposition \ref{sec5:prop:Ableitung nach R von G mu Termen}. Throughout the section for a given $\ell \in \R^4\setminus\{0\}$ with $L=|\ell|$ we have
\[
\rho=\rho(L)=L^{\theta}, \qquad \mu=\mu(L)=\mu_0/L,
\]
with parameters $\theta=4/(4+\delta)>3/(3+\delta)$ and $\mu_0>0$. 
\subsection{On the Interpolation Between $\Gamma$ and $\varphi_0$}\label{subsec:6.1 interpolation}
We continue by proving the following Lemma on the decay of the interpolation function $f$ defined in \eqref{sec5:eq:def_Gamma&f}.

\begin{lemma}\label{app:lem:Anulus terme}
For given $\ell\in \R^4$ with $L=|\ell|>0$ let $f(\ell,\cdot)$ be the function defined in \eqref{sec5:eq:def_Gamma&f}. We always assume $L>0$ large enough. Then for all $x\in A_\rho(\pm \ell/2)$,
\begin{equation}\label{sec6:eq:Abschätzung f_R}
    \abs{f(\ell,x)}\lesssim  L^{-2}\log(L)
\end{equation}
and
\begin{equation}\label{sec5:eq:Abschätzung nabla f_R}
    \abs{P_x f(\ell,x)}\lesssim  L^{-2-\theta}\log(L).
\end{equation}
The weak derivative of $f$ with respect to $\ell$ admits the decomposition
\begin{equation}\label{definition_g_remainder}
   P_{\ell} \, f(\ell,x)= -ig(\ell,x)\mp \frac{1}{2}P_{x} f(\ell,x), \quad  x\in A_\rho(\pm\ell/2).
\end{equation}
The function $g\colon A_\rho(+\ell/2)\cup A_\rho(-\ell/2)\to\R$ obeys
\begin{equation}\label{sec5:eq:Abschätzung g_r}
    \abs{g(\ell,x)}\lesssim L^{-3}\log(L).
\end{equation}
\end{lemma}
\begin{proof}
We establish the bounds on the set \(A_\rho(\ell/2)\); the argument for \(x \in A_\rho(-\ell/2)\) is then similar. The proof proceeds in three steps. First, we prove
\eqref{sec6:eq:Abschätzung f_R}. Second, we compute \(P_x f\) and
show the bound \eqref{sec5:eq:Abschätzung nabla f_R}. Third, we
compute \(P_\ell f\), isolate the function \(g\), and derive the
estimate \eqref{sec5:eq:Abschätzung g_r}.

    \medskip
    \noindent\textbf{Step 1:} 
            By definition we have for any $x\in A_\rho(\ell/2)$
            \begin{equation}\label{sec6:eq:fR_long}
                f(\ell,x) = u_\rho(x-\ell/2)\left(\varphi_0(x-\ell/2)-G_\mu(x-\ell/2) - G_\mu(x+\ell/2)\right).
            \end{equation}
           For $x\in A_\rho(\ell/2)$ we have $\abs{u_\rho(x-\ell/2)}\leq 1$. Using the series expansion of modified Bessel functions in \eqref{app1:eq:K-expansions} and the properties of $\varphi_0$ in Lemma \ref{sec5:lem:AbfallverhaltenResonanz} implies uniformly in $x\in A_\rho(\ell/2)$
            \begin{equation} \label{sec6:eq:first estimates Step 1 Annulus Lemma}
                \begin{split}
                    \abs{\varphi_0(x-\ell/2)-G_\mu(x-\ell/2)}&\lesssim L^{-2}\log(L) + \abs{\epsilon_1(x-\ell/2)},\\
                \abs{G_\mu(x+\ell/2)}&\lesssim L^{-2}.
                \end{split}
            \end{equation}
            From Lemma~\ref{sec5:lem:AbfallverhaltenResonanz} we conclude $\abs{\epsilon_1(x-\ell/2)}\lesssim\abs{x-\ell/2}^{-2-\delta}$. With the choice $\theta>3/(3+\delta)$ we see for $x\in A_\rho(\ell/2)$
            \begin{equation*}
                \abs{\epsilon_1(x-\ell/2)}\lesssim L^{-\theta(2+\delta)} \in O(L^{-3+\theta}).
            \end{equation*}
            Inserting this into \eqref{sec6:eq:first estimates Step 1 Annulus Lemma} implies
            \begin{equation}\label{sec6:eq:second estimates Step 1 Annulus Lemma}
                \begin{split}
                    \abs{\varphi_0(x-\ell/2)-G_\mu(x-\ell/2)}&\lesssim L^{-2}\log(L),\\
                \abs{G_\mu(x+\ell/2)}&\lesssim L^{-2}.
                \end{split}
            \end{equation}
            Using \eqref{sec6:eq:second estimates Step 1 Annulus Lemma} in \eqref{sec6:eq:fR_long} directly proves the inequality \eqref{sec6:eq:Abschätzung f_R}.
            
    \medskip
    \noindent\textbf{Step 2:} We will compute $P_x f= -i\nabla_x f$ and show \eqref{sec5:eq:Abschätzung nabla f_R}. 
            Taking the derivative of $G_\mu$ defined in \eqref{sec:eq:Gmu_def} and $\varphi_0$ in \eqref{sec5:eq:exact_tail} gives
            \begin{equation}\label{sec6:eq:Ableitung_Gmu}
                \begin{split}
                    \nabla_x G_\mu(x)&=-\mu^2\frac{K_2(\mu\abs{x})}{\abs{x}}\frac{x}{\abs{x}} , \quad x\neq 0, \\
                    \nabla_x \varphi_0(x) &= - 2\frac{x}{\abs{x}^{4}}+ \epsilon_2(x), \quad \abs{x}>\rho_0\,.
                \end{split}
            \end{equation}
            and by definition of $u_\rho$ in \eqref{sec5:cut_off_u}
            \begin{equation*}
                \nabla_x u_\rho(x-\ell/2) = -\rho^{-1}\left( \frac{x-\ell/2}{\abs{x-\ell/2}}\right), \quad x\in A_\rho(\ell/2).
            \end{equation*}
        Combining the derivatives above yields for $x\in A_\rho(\ell/2)$
            \begin{equation*}%\label{sec5:eq:Zerlegung nabla f_R}
                P_x f(\ell,x) =-i\nabla_x f(\ell,x)= -i\sum_{j=1}^5 h_{j}(\ell,x),
            \end{equation*}
            where functions $h_{j}(\ell,\cdot)\colon A_\rho(\ell/2)\to \R^4$ are defined as
            \begin{equation*}%\label{sec6:eq:def h_R,j}
                \begin{split}
                    h_{1}(\ell,x)\coloneqq &-\frac{1}{\rho}\left(\abs{x-\ell/2}^{-2}-G_\mu(x-\ell/2)\right)\frac{x-\ell/2}{\abs{x-\ell/2}},\\
                    h_{2}(\ell,x)\coloneqq & \frac{1}{\rho} G_\mu(x+\ell/2)\frac{x-\ell/2}{\abs{x-\ell/2}} \\
                    h_{3}(\ell,x)\coloneqq &u_\rho(x-\ell/2)\left(\mu^2\frac{K_2\left(\mu\abs{x-\ell/2}\right)}{\abs{x-\ell/2}}-\frac{2}{\abs{x-\ell/2}^3}\right)\frac{x-\ell/2}{\abs{x-\ell/2}},\\
                    h_{4}(\ell,x) \coloneqq&u_\rho(x-\ell/2)\left(\mu^2\frac{K_2\left(\mu\abs{x+\ell/2}\right)}{\abs{x+\ell/2}} \right)\frac{x+\ell/2}{\abs{x+\ell/2}},\\
                    h_5(\ell,x)\coloneqq & -\frac{1}{\rho} \epsilon_1(x-\ell/2)\frac{x-\ell/2}{\abs{x-\ell/2}} + u_\rho(x-\ell/2)\epsilon_2(x-\ell/2).
                \end{split}
            \end{equation*}
            \textbf{Note: }By definition of $\epsilon_1$ and $\epsilon_2$ in Lemma~\ref{sec5:lem:AbfallverhaltenResonanz} one has $\nabla_x \epsilon_1(x) = \epsilon_2(x)$ and recall that $u_\rho(x-\ell/2)=2-\abs{x-\ell/2}/\rho$ and on $x\in A_\rho(\ell/2)$ we have $0\leq u_\rho(x-\ell/2)\leq 1$.\\ 
            
            For $L>0$ large enough it follows $\abs{x+\ell/2}\geq L/2$ for $x\in A_\rho(\ell/2)$. Using the series expansion of modified Bessel functions of the second kind in \eqref{app1:eq:K-expansions}, the decay of $\epsilon_1$ and $\epsilon_2$ in Lemma \ref{sec5:lem:AbfallverhaltenResonanz} together with $\theta>3/(3+\delta)$ and $\mu=\mu_0/L$ give immediately on $A_\rho(\ell/2)$ the following bounds
            \begin{equation}\label{sec6:eq:estimates h_R,j}
                \begin{split}
                    \abs{h_{1}(\ell,x)}&\lesssim L^{-2-\theta}\log(L),\\
                    \abs{h_{2}(\ell,x)}&\lesssim L^{-2-\theta},\\
                    \abs{h_{3}(\ell,x)}&\lesssim L^{-2-\theta},\\
                    \abs{h_{4}(\ell,x)}&\lesssim L^{-3},\\
                    \abs{h_{5}(\ell,x)}&\lesssim L^{-3}.
                \end{split}
            \end{equation}
            From the decay of $h_{1}$ we conclude inequality \eqref{sec5:eq:Abschätzung nabla f_R}.

            \medskip
            \noindent\textbf{Step 3:} We compute \(P_{\ell} f\), isolate the function \(g\), and derive the estimate \eqref{sec5:eq:Abschätzung g_r}. First note that by chain rule ($P_\rho = -i\partial_\rho$, $P_\mu = -i\partial_\mu$)
            \begin{equation}\label{sec6:eq:Kettenregel u_rho}
                \begin{split}
                    \nabla_{\ell} u_{\rho}(x-\ell/2) &= \left(\partial_\rho u_\rho(x-\ell/2) \right)\nabla_{\ell}\rho(\ell)\\
                    &\phantom{=}+ \left(\nabla_x u_\rho(x-\ell/2)\right)\left(\frac{x-\ell/2}{\abs{x-\ell/2}}\cdot \nabla_{\ell}\left(\abs{x-\ell/2}\right)\right)\\
                    &=\theta L^{\theta -1}\left(\partial_\rho u_\rho(x-\ell/2)\right)(\ell/L) - \frac{1}{2}\nabla_x u_\rho(x-\ell/2)
                \end{split}
            \end{equation}
            and
            \begin{equation}\label{sec6:eq:Kettenregel G_mu}
            \begin{split}
                \nabla_{\ell}G_\mu(x\pm\ell/2) &= \left(\partial_\mu G_\mu(x\pm\ell/2)\right)\nabla_{\ell}\mu + \left(\nabla_x G_\mu(x\pm\ell/2)\right)\left(\frac{x\pm\ell/2}{\abs{x-\ell/2}}\cdot \nabla_{\ell}\abs{x\pm\ell/2}\right)\\
                &=-\mu_0 L^{-2}\left(\partial_\mu G_\mu(x\pm\ell/2)\right)(\ell/L) \pm\frac{1}{2}\nabla_x G_\mu(x\pm\ell/2)
            \end{split}
            \end{equation}
            
            Using the definitions of $u_\rho$ in \eqref{sec5:cut_off_u} and $G_\mu$ in \eqref{sec:eq:Gmu_def} we find the derivatives
            \begin{equation}\label{sec6:eq:Ableitung nach mu und rho}
            \begin{split}
                 \partial_\rho u_\rho(x-\ell/2) &= \frac{\abs{x-\ell/2}}{L^{2\theta}}, \\
                 \partial_\mu G_\mu(x\pm\ell/2) &=  \frac{K_1\left( \mu \abs{x\pm\ell/2}\right)}{\abs{x\pm\ell/2}} + \mu K_1'\left( \mu \abs{x\pm\ell/2}\right).
            \end{split}
            \end{equation}
            Applying the recurrence formula for derivatives of Bessel functions in \cite[p. 79 eq. (3)]{book:W:1944} (see \eqref{app1:eq:reoc}) we find
            \begin{equation}\label{sec6:eq:ableitung Gmu nach mu}
                \partial_\mu G_\mu(x\pm\ell/2) = - \mu K_0\left(\mu\abs{x\pm\ell/2}\right).
            \end{equation}
            In the lemma, the function $g$ is defined on the annulus $A_\rho(\ell/2)$ as
            \begin{equation*}
                -ig(\ell,x)=P_{\ell} f(\ell,x)+\frac{1}{2}P_x f(\ell,x)\, .
            \end{equation*}
            We find by applying \eqref{sec6:eq:Kettenregel u_rho}- \eqref{sec6:eq:ableitung Gmu nach mu} the explicit expression
            \begin{equation}\label{sec6:g_est}
                \begin{split}
                    g(\ell,x)= &-\mu_0 L^{-2} u_\rho(x-\ell/2)\left(\mu K_0\left(\mu\abs{x-\ell/2}\right)+\mu K_0\left(\mu\abs{x+\ell/2}\right)\right)\frac{\ell}{L}\\
                &+\left(\varphi_0(x-\ell/2)-G_\mu(x-\ell/2)-G_\mu(x+\ell/2)\right)\theta\frac{\abs{x-\ell/2}}{L^{1+\theta}}\frac{\ell}{L}\\
                &-\underbrace{u_\rho(x-\ell/2)\nabla_x G_\mu(x+\ell/2)}_{=h_{4}(\ell,x)}.
                \end{split}
            \end{equation}
            Using $\mu=\mu_0/L$ and estimating the last two lines in \eqref{sec6:g_est} by application of  \eqref{sec6:eq:second estimates Step 1 Annulus Lemma} and \eqref{sec6:eq:estimates h_R,j} we arrive at
            \begin{equation}\label{sec6:eq:vorletzte Abschätzung g}
                g(\ell,x)= -\mu_0^2 L^{-3} u_\rho(x-\ell/2)\left(
                K_0\left(\mu\abs{x-\ell/2}\right)+ K_0\left(\mu\abs{x+\ell/2}\right)\right)\frac{\ell}{L} + O(L^{-3}\log(L)).
            \end{equation}
            Due to the series expansion for Bessel functions in \eqref{app1:eq:K-expansions} and since $\mu\abs{x-\ell/2}\sim L^{-1+\theta} $ we see
            \begin{equation}\label{sec6:eq:letzte Abschätzung g}
                \abs{-\mu_0^2 L^{-3} u_\rho(x-\ell/2)\left(
                K_0\left(\mu\abs{x-\ell/2}\right)+ K_0\left(\mu\abs{x+\ell/2}\right)\right)} \lesssim L^{-3}\log(L).
            \end{equation}
            Inserting \eqref{sec6:eq:letzte Abschätzung g} into the right-hand side of \eqref{sec6:eq:vorletzte Abschätzung g} concludes the proof of \eqref{sec5:eq:Abschätzung g_r}.
        \end{proof}

\subsection{Proof of Proposition \ref{sec5:prop:Ableitung nach R von G mu Termen}}\label{subsec64: Energy außerhalb terme}
        \begin{proposition}[=Proposition \ref{sec5:prop:Ableitung nach R von G mu Termen}]\label{sec5:prop:Ableitung nach R von G mu Termen:copy}
   For $\ell \in \R^4$ with $\ell= (0,0,r)$ where $r\in \R^2$ and $L=|\ell(r)|=\abs{r}$ large enough let $\mu=\mu(r)=\mu_0/L$ and $\Gamma(\ell, \cdot )$ be the function which is defined in \eqref{sec5:eq:def_Gamma&f} as
        \begin{equation*}
            \Gamma(\ell,x ) = G_\mu(x+\ell/2) +G_\mu(x-\ell/2).
        \end{equation*}
        Then 
        \begin{equation}\label{sec6:eq:prop Ableitung nach R außerhalb to show}
          \left \langle \Gamma(\ell(r),\cdot ) , \left(P_r^2-\frac{1}{4}P_{(x_3,x_4)}^2\right)  \Gamma(\ell(r),\cdot  )  \right\rangle_{L^2(\Omega_\rho)}  =   C_5(\mu_0) L^{-2}+O(L^{-3+\theta})
        \end{equation}
        where
        \begin{align*}
            C_5(\mu_0) =(4\pi^2/3)\left(1+\mu_0 K_1(\mu_0)+\tfrac{7}{2}\mu_0^2 K_0(\mu_0)\right).
        \end{align*}
        \end{proposition}
        \begin{proof}
             For any admissible, $y_1,y_2:\R^2\times \R^4 \to \R$ define
        \begin{equation}
            \begin{split}\label{sec6:eq:def Form Q}
                \mathcal{Q}(y_1,y_2) &\coloneqq \scp{\nabla_r y_1(r,\cdot)}{\nabla_r y_2(r,\cdot)}_{L^2(\Omega_\rho)}\\
                &\phantom{\coloneqq}-(1/4)\scp{\nabla_{(x_3,x_4)} y_1(r,\cdot)}{\nabla_{(x_3,x_4)} y_2(r,\cdot)}_{L^2(\Omega_\rho)},
            \end{split}
            \end{equation}
            and $\mathcal{Q}(y_1)\coloneqq \mathcal{Q}(y_1,y_1)$. By abuse of notation we denote $\Gamma : \R^2 \times \R^4 \to \R$ defined as $\Gamma(r,x) = \Gamma(\ell(r),x)$ by the same latter and therefore make sense of the expression $\mathcal{Q}(\Gamma)$ and similar for functions $G_\mu(\cdot \pm \ell(r)/2)$.
            
            We will show \eqref{sec6:eq:prop Ableitung nach R außerhalb to show}, which then reads
            \begin{equation*}
                \mathcal{Q}(\Gamma) = C_5(\mu_0) L^{-2} + O(L^{-3+\theta}).
            \end{equation*}
            By definition of $\Gamma$ in \eqref{sec5:eq:def_Gamma&f} it follows
            \begin{equation*}
                \mathcal{Q}(\Gamma) =\mathcal{Q}(G_\mu(\cdot+\ell/2))+\mathcal{Q}(G_\mu(\cdot-\ell/2)) +2 \mathcal{Q}(G_\mu(\cdot+\ell/2),G_\mu(\cdot-\ell/2)).
            \end{equation*}
            In the first step we show
            \begin{equation}\label{sec6:eq:goal step1 ableitung nach R}
                \mathcal{Q}(G_\mu(\cdot+\ell/2))+\mathcal{Q}(G_\mu(\cdot-\ell/2)) = (4\pi^2/3) L^{-2}+ O(L^{-3+\theta}).
            \end{equation}
            In the second step we will prove
            \begin{equation*}
                2 \mathcal{Q}(G_\mu(\cdot+\ell/2),G_\mu(\cdot-\ell/2)) = (4\pi^2/3)(\mu_0K_1(\mu_0)+\tfrac{7}{2}\mu_0^2K_0(\mu_0)) + O(L^{-3+\theta})
            \end{equation*}
            which then finishes the proof of \eqref{sec6:eq:prop Ableitung nach R außerhalb to show}. We use the aberration $\Omega_\rho(\ell/2) = \Omega_\rho $ as $\ell(r) = (0,0,r)$ is fixed for some $L>0$ large enough. 
            
            \medskip
            \noindent\textbf{Step 1:} We will show relation \eqref{sec6:eq:goal step1 ableitung nach R}. For any fixed $r\in \R^2$ we define, $u_\pm = G_{\mu(r)}(\cdot \pm \ell(r)/2)$ then the $r$-gradient of $u_\pm$ can be written as
\begin{equation}\label{eq:grad_r_u_pm}
\nabla_r u_\pm(x)
=
\pm \frac{1}{2}\,\nabla_{(x_3,x_4)} u_\pm(x)
-\frac{\mu}{L}\,\frac{r}{L}\,\partial_\mu u_\pm(x).
\end{equation}
Squaring this expression yields
\begin{equation}\label{eq:grad_r_u_pm_squared}
\abs{\nabla_r u_\pm(x)}^2
=
\frac{1}{4}\,\abs{\nabla_{(x_3,x_4)} u_\pm(x)}^2
+\frac{\mu^2}{L^2}\,\abs{\partial_\mu u_\pm(x)}^2
\mp \frac{\mu}{L}\, \frac{r}{L} \Bigl(\nabla_{(x_3,x_4)} u_\pm(x)\,\partial_\mu u_\pm(x)\Bigr).
\end{equation}
Thus, for the bilinear form $\mathcal{Q}$ we find
\begin{equation}\label{eq:Q_u_pm}
\begin{split}
    \mathcal Q(u_\pm)
=
\frac{\mu^2}{L^2}\,\norm{\partial_\mu u_\pm}_{L^2(\Omega_\rho)}^2
\mp \mu\,\frac{1}{L^2}\,
\scp{\,r\cdot\nabla_{(x_3,x_4)} u_\pm\,}{\partial_\mu u_\pm}_{L^2(\Omega_\rho)} \\
= \frac{\mu^2}{L^2}\,\norm{\partial_\mu u_+}_{L^2(\Omega_\rho)}^2
- \mu\,\frac{1}{L^2}\,
\scp{\,r\cdot\nabla_{(x_3,x_4)} u_+\,}{\partial_\mu u_+}_{L^2(\Omega_\rho)}.
\end{split}
\end{equation}
Here we have used the reflection $x \mapsto -x$, which leaves $\Omega_\rho$ invariant and satisfies
$u_+(x)=u_-(-x)$, we obtain
\[
\nabla_{(x_3,x_4)} u_+(x) = -\,\nabla_{(x_3,x_4)} u_-(-x),
\qquad
\partial_\mu u_+(x)=\partial_\mu u_-(-x).
\]
Hence, summing in the following sense, we arrive at
\begin{equation}\label{sec6:eq:Q u pm aufsummiert}
    \mathcal Q(u_+) + \mathcal Q(u_-) = 2\frac{\mu^2}{L^2}\,\norm{\partial_\mu u_\pm}_{L^2(\Omega_\rho)}^2 - 2\mu\,\frac{1}{L^2}\,
\scp{\,r\cdot\nabla_{(x_3,x_4)} u_+\,}{\partial_\mu u_+}_{L^2(\Omega_\rho)}.
\end{equation}
Next we show that by the underlying symmetry
\begin{equation}\label{sec6:eq:I term vanishes}
    \scp{\,r\cdot\nabla_{(x_3,x_4)} u_+\,}{\partial_\mu u_+}_{L^2(\R^4\setminus B_\rho(-\ell/2))}=0.
\end{equation}
This follows from the transformation $x\mapsto x+\ell/2$ which maps
\begin{equation*}
    \R^4\setminus B_\rho(-\ell/2) \to \R^4\setminus B_\rho(0).
\end{equation*}
By the definition of $u_+$ together with the derivative of $G_\mu$ in \eqref{sec6:eq:Ableitung_Gmu} we arrive at
\begin{equation}\label{sec6:eq:I term after transformation}
    \scp{\,r\cdot\nabla_{(x_3,x_4)} u_+\,}{\partial_\mu u_+}_{L^2(\R^4\setminus B_\rho(-\ell/2))}=\int_{\R^4\setminus B_\rho(0)}  \left(-\mu^2\frac{K_2(\mu\abs x)}{\abs x} \frac{r\cdot x}{\abs x}\right) \partial_\mu\left(\mu\frac{K_1(\mu\abs x)}{\abs x}\right) \, dx
\end{equation}
The right-hand side of \eqref{sec6:eq:I term after transformation} vanishes due to symmetry, which proves \eqref{sec6:eq:I term vanishes}. Using \eqref{sec6:eq:I term vanishes} in \eqref{sec6:eq:Q u pm aufsummiert} implies with the definition of $\Omega_\rho$
\begin{equation}\label{sec6:eq: Q u pm aufsummiert Zwischenergebnis}
    \mathcal Q(u_+) + \mathcal Q(u_-) = 2\frac{\mu^2}{L^2}\,\norm{\partial_\mu u_+}_{L^2(\Omega_\rho)}^2 + 2\mu\,\frac{1}{L^2}\,
\scp{\,r\cdot\nabla_{(x_3,x_4)} u_+\,}{\partial_\mu u_+}_{L^2(B_\rho(+\ell/2))}.
\end{equation}
Next we show that $\abs{\nabla_{(x_3,x_4)} u_+}$ and $\abs{\partial_\mu u_+}$ are sufficiently small on $B_\rho(+\ell/2)$. On the set $B_\rho(+\ell/2)$ for $L>0$ large enough we have
            \begin{equation*}
                \mu \abs{x+\ell/2} = \mu_0 L^{-1}\abs{x+\ell/2} \geq  \mu_0/2
            \end{equation*}
            and using that $K_0$ and $K_2$ are monotonically decreasing yields with the derivatives of $G_\mu$ with respect to $x$ and $\mu$ in \eqref{sec6:eq:Ableitung_Gmu} and \eqref{sec6:eq:Ableitung nach mu und rho}
            \begin{equation}\label{sec6:eq:sup_ests:copy}
                \begin{split}
                    \sup_{x\in B_\rho(+ \ell/2)}\abs{\nabla_{x}u_+(x)} &= \sup_{x\in B_\rho(+ \ell/2)}\mu^2 \frac{K_2\left(\mu\abs{x+\ell/2}\right)}{\abs{x+\ell/2}} \lesssim L^{-3}, \\
                    \sup_{x\in B_\rho(+ \ell/2)} \abs{\partial_\mu u_+(x)}&= \sup_{x\in B_\rho(+ \ell/2)} \abs{\mu K_0(\mu|x+ \ell/2|)} \lesssim L^{-1}.
                \end{split}
            \end{equation}
Since $\abs{B_\rho(0)}\sim L^{4\theta}$ and applying \eqref{sec6:eq:sup_ests:copy} we derive from \eqref{sec6:eq: Q u pm aufsummiert Zwischenergebnis}
\begin{equation}\label{sec6:eq: Q u pm aufsummiert Zwischenergebnis2}
    \mathcal Q(u_+) + \mathcal Q(u_-) = 2\frac{\mu^2}{L^2}\,\norm{\partial_\mu u_+}_{L^2(\Omega_\rho)}^2 + O(R^{-3+\theta}).
\end{equation}
Next, we show
\begin{equation}\label{eq:Omega_decomposition}
\norm{\partial_\mu u_+}_{L^2(\Omega_\rho)}^2
=
\norm{\partial_\mu u_+}_{L^2(\R^4)}^2
-
\sum_{\sigma=\pm}
\norm{\partial_\mu u_+}_{L^2(B_\rho(\sigma\ell/2))}^2 =\norm{\partial_\mu u_+}_{L^2(\R^4)}^2+O(L^{1+\theta}).
\end{equation}
Due to \eqref{sec6:eq:sup_ests:copy} we have
\begin{equation}\label{sec6:eq: dmu u+ auf plus lhalbe}
\norm{\partial_\mu u_+}_{L^2(B_\rho(\ell/2))}^2
\lesssim L^{-2}\abs{B_\rho}
\sim L^{-2}\rho^4
=
L^{-2+4\theta}.
\end{equation}
By definition and change to spherical coordinates we have
\begin{equation*}
    \norm{\partial_\mu u_+}_{L^2(B_\rho(-\ell/2))}^2=\mu^2 \int_{B_\rho(-\ell/2)} K_0(\mu\abs{x+\ell/2})^2 dx=\mu^2\int_0^\rho K_0(\mu t)^2 t^3 \, dt.
\end{equation*}
 Substituting $s=\mu t$ and set $a=\mu\rho=\mu_0L^{\theta-1}\to0$ implies
 \begin{equation*}
     \norm{\partial_\mu u_+}_{L^2(B_\rho(-\ell/2))}^2=\mu^{-2} \int_0^aK_0(s)^2s^3 ds.
 \end{equation*}
 Since $K_0(s)\sim -\log(s)$ for $s\to 0^+$ (see \eqref{app1:eq:K-expansions}) we conclude
 \begin{equation}\label{sec6:eq: dmu u+ auf minus lhalbe}
     \norm{\partial_\mu u_+}_{L^2(B_\rho(-\ell/2))}^2\lesssim \mu^{-2} a^4 \log(a)^2\lesssim L^{-2+4\theta}\log(L)^2\lesssim L^{-1+3\theta}.
 \end{equation}
 Using \eqref{sec6:eq: dmu u+ auf plus lhalbe} and \eqref{sec6:eq: dmu u+ auf minus lhalbe} implies \eqref{eq:Omega_decomposition}. By applying \eqref{eq:Omega_decomposition} to \eqref{sec6:eq: Q u pm aufsummiert Zwischenergebnis2} it follows
 \begin{equation*}
     \mathcal Q(u_+) + \mathcal Q(u_-) = 2\frac{\mu^2}{L^2}\,\norm{\partial_\mu u_+}_{L^2(\R^4)}^2 + O(R^{-3+\theta}).
 \end{equation*}
 It remains to compute the remaining norm. By definition of $u_+$ we have
 \begin{equation*}
     \norm{\partial_\mu u_+}_{L^2(\R^4)}^2 = \mu^2\int_{\R^4}K_0(\mu\abs{x})^2 dx.
 \end{equation*}
 This integral is solved explicitly in the appendix as Corollary \ref{app2:cor:L2 norm K_0}. This concludes the proof of relation \eqref{sec6:eq:goal step1 ableitung nach R}.

\noindent\textbf{Step 2:} In the following we prove
            \begin{equation}\label{sec6:eq:goal step2 ableitung nach R}
                2 \mathcal{Q}(u_+,u_-) = 2\pi^2\left(4/3\mu_0^2K_2(\mu_0)+\mu_0^2K_0(\mu_0)-2\mu_0K_1(\mu_0)\right)L^{-2}+O(L^{-3+\theta}).
            \end{equation}
            Since all functions are real-valued we obtain by definition
            \begin{equation}\label{sec6:eq:Crossterm Q erste Version}
                \begin{split}
                    2 \mathcal{Q}(u_+,u_-)
                    = 2\int_{\Omega_\rho} \nabla_r u_+(x)\cdot\nabla_r u_-(x) \,dx
                    -(1/2) \nabla_{(x_3,x_4)}u_+(x)\cdot\nabla_{(x_3,x_4)}u_-(x) dx
                \end{split}
            \end{equation}
            By inserting \eqref{eq:grad_r_u_pm} into \eqref{sec6:eq:Crossterm Q erste Version} we arrive at
            \begin{equation}\label{sec6:eq:Crossterm Q ausmultipliziert}
            \begin{split}
                 2 \mathcal{Q}(u_+,u_-) = &2\frac{\mu^2}{L^2}\scp{\partial_\mu u_+}{\partial_\mu u_-}_{L^2(\Omega_\rho)}-\frac{\mu}{L^2}\scp{r\cdot\nabla_{(x_3,x_4)}u_+}{\partial_\mu u_-}_{L^2(\Omega_\rho)}\\
                 &+ \frac{\mu}{L^2}\scp{r\cdot \nabla_{(x_3,x_4)}u_-}{\partial_\mu u_+}_{L^2(\Omega_\rho)}- \scp{\nabla_{(x_3,x_4)}u_+}{\nabla_{(x_3,x_4)}u_-}_{L^2(\Omega_\rho)}.
            \end{split}
            \end{equation}
            Due to the symmetry of $\Omega_\rho$ and the substitution $x\mapsto -x$ in the first term of the second line of \eqref{sec6:eq:Crossterm Q ausmultipliziert} we obtain 
            \begin{equation}\label{sec6:eq:Crossterm Q vereinfacht}
            \begin{split}
                 2 \mathcal{Q}(u_+,u_-) = &2\frac{\mu^2}{L^2}\scp{\partial_\mu u_+}{\partial_\mu u_-}_{L^2(\Omega_\rho)}-2\frac{\mu}{L^2}\scp{r\cdot\nabla_{(x_3,x_4)}u_+}{\partial_\mu u_-}_{L^2(\Omega_\rho)}\\
                 &- \scp{\nabla_{(x_3,x_4)}u_+}{\nabla_{(x_3,x_4)}u_-}_{L^2(\Omega_\rho)}.
            \end{split}
            \end{equation}
            We claim, similar to step 1, that extending the domain of integration in the inner products on the right-hand side of \eqref{sec6:eq:Crossterm Q vereinfacht} by the sets $B_\rho(\pm\vec{R}/2)$ yields only corrections in $O(L^{-3+\theta})$, meaning
            \begin{align}
                \frac{\mu^2}{L^2}\scp{\partial_\mu u_+}{\partial_\mu u_-}_{L^2(B_\rho(\pm\ell/2))}&\in O(L^{-3+\theta}),\label{sec6:eq:claim first crossterm}\\
                \frac{\mu}{L^2}\scp{r\cdot\nabla_{(x_3,x_4)}u_+}{\partial_\mu u_-}_{L^2(B_\rho(\pm\ell/2))}&\in O(L^{-3+\theta}),\label{sec6:eq:claim second crossterm}\\
                \scp{\nabla_{(x_3,x_4)}u_+}{\nabla_{(x_3,x_4)}u_-}_{L^2(B_\rho(\pm\ell/2))}&\in O(L^{-3+\theta})\label{sec6:eq:claim third crossterm}.
            \end{align}
            For convenience we present these calculations at the end of this proof. Consequently, using \eqref{sec6:eq:claim first crossterm}-\eqref{sec6:eq:claim third crossterm} in \eqref{sec6:eq:Crossterm Q vereinfacht} we obtain
            \begin{equation}\label{sec6:eq:Ableitung nach R second step Zerlegung mit Fehler}
            \begin{split}
                 2 \mathcal{Q}(u_+,u_-) = &2\frac{\mu^2}{L^2}\scp{\partial_\mu u_+}{\partial_\mu u_-}_{L^2(\R^4)}-2\frac{\mu}{L^2}\scp{r\cdot\nabla_{(x_3,x_4)}u_+}{\partial_\mu u_-}_{L^2(\R^4)}\\
                 &- \scp{\nabla_{(x_3,x_4)}u_+}{\nabla_{(x_3,x_4)}u_-}_{L^2(\R^4)}+O(L^{-3+\theta}).
            \end{split}
            \end{equation}
            The remaining inner products are evaluated in the Appendix \ref{app:integrals}. We present these calculations as the content of Lemma \ref{app2:lem:Faltung K0 mit sich selbst}, Lemma \ref{app2:lem:Faltung K0 mit Gmu} and Lemma \ref{app2:lem:Faltung Ableitung Gmu mit sich selbst}. In particular, we show in Lemma \ref{app2:lem:Faltung K0 mit sich selbst}
            \begin{equation}\label{sec6:eq:value I_2}
                2\frac{\mu^2}{L^2}\scp{\partial_\mu u_+}{\partial_\mu u_-}_{L^2(\R^4)}=(2\pi^2/3)\mu_0^2K_2(\mu_0) L^{-2},
            \end{equation}
            in Lemma \ref{app2:lem:Faltung K0 mit Gmu}
            \begin{equation}\label{sec6:eq:value I_3}
                -2\frac{\mu}{L^2}\scp{r\cdot\nabla_{(x_3,x_4)}u_+}{\partial_\mu u_-}_{L^2(\R^4)} = 2\pi^2\mu_0^2K_0(\mu_0) L^{-2}.
            \end{equation}
            Due to Lemma \ref{app2:lem:Faltung Ableitung Gmu mit sich selbst} we conclude
            \begin{equation}\label{sec6:eq:value I_4}
                - \scp{\nabla_{(x_3,x_4)}u_+}{\nabla_{(x_3,x_4)}u_-}_{L^2(\R^4)} = 2\pi^2(\mu_0^2K_2(\mu_0)-2\mu_0K_1(\mu_0))L^{-2}.
            \end{equation}
            Inserting \eqref{sec6:eq:value I_2}-\eqref{sec6:eq:value I_4} in the right-hand side of \eqref{sec6:eq:Ableitung nach R second step Zerlegung mit Fehler} 
            \begin{equation*}
                 2 \mathcal{Q}(u_+,u_-) =2\pi^2\left(4/3\mu_0^2K_2(\mu_0)+\mu_0^2K_0(\mu_0)-2\mu_0K_1(\mu_0)\right)L^{-2}+O(L^{-3+\theta}).
            \end{equation*}
            Applying the recurrence relations of modified Bessel functions in \eqref{app1:eq:reoc} completes, under the assumption of \eqref{sec6:eq:claim first crossterm}-\eqref{sec6:eq:claim third crossterm}, the proof of step 2. It remains to prove \eqref{sec6:eq:claim first crossterm}-\eqref{sec6:eq:claim third crossterm}.\\
            We start by proving \eqref{sec6:eq:claim first crossterm}. Recall the definition of $u_\pm$, \eqref{sec6:eq: dmu u+ auf plus lhalbe} and \eqref{sec6:eq: dmu u+ auf minus lhalbe} to see
            \begin{equation}\label{sec6:eq:gesammelte Ball Abschätzungen upm}
                \begin{split}
                    \norm{\partial_\mu u_+}_{L^2(B_\rho(+\ell/2))}^2= \norm{\partial_\mu u_-}_{L^2(B_\rho(-\ell/2))}^2 &\lesssim L^{-2+4\theta},\\
                    \norm{\partial_\mu u_+}_{L^2(B_\rho(-\ell/2))}^2= \norm{\partial_\mu u_-}_{L^2(B_\rho(+\ell/2))}^2 &\lesssim L^{-2+4\theta}\log(L).
                \end{split}
            \end{equation}
            Applying the Cauchy-Schwarz inequality yields \eqref{sec6:eq:claim first crossterm}.\\
            To prove \eqref{sec6:eq:claim second crossterm} first note that by Cauchy-Schwarz and the fact that $2ab\leq a^2+b^2$ for $a,b\in\R$ we obtain
            \begin{equation}\label{sec6:Crossterm step2 Ball plus}
                \abs{\frac{\mu}{L^2}\scp{r\cdot\nabla_{(x_3,x_4)}u_+}{\partial_\mu u_-}_{L^2(B_\rho(+\ell/2))}} \lesssim \norm{\nabla_{(x_3,x_4)}u_+}^2_{L^2(B_\rho(+\ell/2))} + \frac{\mu^2}{L^2} \norm{\partial_\mu u_-}^2_{L^2(B_\rho(+\ell/2))}.
            \end{equation}
            For $x\in B_\rho(+\ell/2)$ we know by definition
            \begin{equation}\label{sec6:eq:Abschätzung nabla x u+}
                \abs{\nabla_{(x_3,x_4)}u_+(x)}\lesssim L^{-3}.
            \end{equation}
            Since $\abs{B_\rho(+\ell/2)}\sim R^{4\theta}$, we see by inserting \eqref{sec6:eq:gesammelte Ball Abschätzungen upm} and \eqref{sec6:eq:Abschätzung nabla x u+} into \eqref{sec6:Crossterm step2 Ball plus}
            \begin{equation}\label{sec6:Crossterm step2 Ball plus Ergebnis}
                \frac{\mu}{L^2}\scp{r\cdot\nabla_{(x_3,x_4)}u_+}{\partial_\mu u_-}_{L^2(B_\rho(+\ell/2))} \in O(L^{-3+\theta}).
            \end{equation}
            For $x\in B_\rho(-\ell/2)$ it follows $\abs{\mu/L^2\partial_\mu u_-(x)\cdot r}\lesssim L^{-3}$. By adding the derivatives with respect to $x_1$ and $x_2$, this implies
            \begin{equation}\label{sec6:eq:second claim Anfang schwieriger Ball}
                 \abs{\frac{\mu}{L^2}\scp{r\cdot\nabla_{(x_3,x_4)}u_+}{\partial_\mu u_-}_{L^2(B_\rho(-\ell/2))}}\lesssim L^{-3} \int_{B_\rho(-\ell/2)} \abs{\nabla_{x}u_+(x)}\, dx.
            \end{equation}
            By using the derivative of $u_+$, we see
            \begin{equation*}
                \int_{B_\rho(-\ell/2)} \abs{\nabla_{x}u_+(x)} \, dx = \int_{B_\rho(-\ell/2)} \mu^2\frac{K_2(\mu\abs{x+\ell/2})}{\abs{x+\ell/2}} \, dx.
            \end{equation*}
            Passing to spherical coordinates yields
            \begin{equation}\label{sec6:eq: Zwischenschritt L1 Norm auf Ball von nabla u plus}
                \int_{B_\rho(-\ell/2)} \abs{\nabla_{x}u_+(x)} \, dx = \mu^{-1} \int_0^{\mu\rho} s^2 K_2(s)\, ds
            \end{equation}
            The integrand on the right--hand side of \eqref{sec6:eq: Zwischenschritt L1 Norm auf Ball von nabla u plus} is bounded and therefore
            \begin{equation}\label{sec6:eq:L1 Norm auf Ball von nabla u plus}
                \int_{B_\rho(-\ell/2)} \abs{\nabla_{x}u_+(x)} \, dx \lesssim L^{\theta}.
            \end{equation}
            Inserting \eqref{sec6:eq:L1 Norm auf Ball von nabla u plus} into \eqref{sec6:eq:second claim Anfang schwieriger Ball} implies
            \begin{equation}\label{sec6:Crossterm step2 Ball minus Ergebnis}
                \frac{\mu}{L^2}\scp{r\cdot\nabla_{(x_3,x_4)}u_+}{\partial_\mu u_-}_{L^2(B_\rho(-\ell/2))}\in O(L^{-3+\theta}).
            \end{equation}
            Combining \eqref{sec6:Crossterm step2 Ball plus Ergebnis} and \eqref{sec6:Crossterm step2 Ball minus Ergebnis} implies \eqref{sec6:eq:claim second crossterm}.\\
            In order to prove \eqref{sec6:eq:claim third crossterm}, it is enough to show
            \begin{equation}\label{sec6:eq:claim to show third crossterm}
                \scp{\nabla_{(x_3,x_4)}u_+}{\nabla_{(x_3,x_4)}u_-}_{L^2(B_\rho(-\ell/2))}\in O(L^{-3+\theta})
            \end{equation}
            due to symmetry. Applying \eqref{sec6:eq:L1 Norm auf Ball von nabla u plus} and $\abs{\nabla_{(x_3,x_4)}u_-(x)}\lesssim L^{-3}$ for $x\in B_\rho(-\ell/2)$ implies \eqref{sec6:eq:claim to show third crossterm}.
            This finishes the proof of \eqref{sec6:eq:claim first crossterm}-\eqref{sec6:eq:claim third crossterm} and therefore concludes the proof of Proposition \ref{sec5:prop:Ableitung nach R von G mu Termen:copy}.
        \end{proof}
\newpage
%%%%%%%%%%%%%%%%%%%%%%%%%%%%%%%%%%%%%%%%%%%
%%%%%%%%%%%%%%%%%%%%%%%%%%%%%%%%%%%%%%%%%%%
%%%%%%   APENDIX
%%%%%%%%%%%%%%%%%%%%%%%%%%%%%%%%%%%%%%%%%%%
%%%%%%%%%%%%%%%%%%%%%%%%%%%%%%%%%%%%%%%%%%%
\appendix

\section{On Some Series Expansion of Bessel Functions}\label{app:Bessel}
In the proof of Lemma~\ref{sec2:lem:proto_efimov} we use the expansion of the
modified Bessel functions of the second kind $K_j$ for $j\in\{1,2,3\}$ at
zero. We derive the near–zero behavior of $K_0$ and $K_1$ from their
series representations, and then obtain the expansions of $K_2$ and $K_3$
from the recurrence relations.

We use the following representations of $K_0$, $K_1$ and $I_{\nu}$ from
\cite[p.~80, eqs.~(14)–(15) and p.~77, eq.~(2)]{book:W:1944}
\begin{equation} \label{app1:eq:def_bessels}
    \begin{aligned}
        K_0(z) &= -\log(z/2)\, I_0(z)
        + \sum_{m=0}^{\infty} \frac{(z^{2}/4)^{m}}{(m!)^{2}}\, \Psi(m+1), \\[0.4em]
        K_1(z) &= z^{-1}
        + \log(z/2)\, I_1(z)
        - \frac{z}{4} \sum_{m=0}^{\infty}
          \frac{(z^{2}/4)^{m}}{m!(m+1)!}\, \bigl(\Psi(m+1)+\Psi(m+2)\bigr), \\[0.4em]
        I_{\nu}(z) &= \sum_{m=0}^{\infty}
        \frac{(z/2)^{\nu+2m}}{m!\,\Gamma_\ast(\nu+m+1)}, \qquad \nu \ge 0 ,
    \end{aligned}
\end{equation}
where $\Gamma_\ast$ is the usual gamma function and $\Psi=\Gamma_\ast'/\Gamma_\ast$ is the digamma function. We use the symbol $\Gamma_\ast$ to avoid confusion with the function $\Gamma$ defined earlier.

\begin{remark}
Although one can work with complex arguments $z \in \mathbb{C}$ by choosing
a branch of the logarithm and restricting to sectors avoiding the negative
real axis, we only require the near–zero expansions along the positive real
axis. For this reason, and to simplify notation, we state and use all
asymptotic expansions for real $z>0$, even though the arguments extend to
appropriate complex sectors.
\end{remark}

Following \cite[p.~79, eq.~(1) -- (4) ]{book:W:1944}, the recurrence relations for
the modified Bessel functions of the second kind and their derivatives are
\begin{equation}\label{app1:eq:reoc}
    \begin{split}
        K_{\nu+1}(z) &= K_{\nu-1}(z) + 2 \nu\, z^{-1} K_{\nu}(z),
    \qquad \nu \ge 1, z>0. \\
        -K_\nu'(z) &= K_{\nu-1}(z)+ \nu z^{-1} K_{\nu}(z),
    \qquad \nu \ge 1, z>0. \\
   -2K_\nu'(z) &=  K_{\nu-1}(z) + K_{\nu+1}(z) \qquad \nu \ge 1, z>0. \\
   -K_\nu'(z) &= K_{\nu+1}(z) -\nu z^{-1} K_{\nu}(z) \qquad \nu \ge 1, z>0.
    \end{split}
\end{equation}
(where the latter two relations follow from the first two.) Equipped with these textbook relations, we prove the following standard expansions for which we did not find an obvious citable source.
\pagebreak[4]
\begin{lemma}[Near–Zero Expansion of Modified Bessel Functions of the Second Kind]
\label{app1:lem:Kj_expansion}
Let $K_j$, $j\in\{0,1,2,3\}$, be the modified Bessel functions defined in
\eqref{app1:eq:def_bessels} and satisfying \eqref{app1:eq:reoc}. Then, as
$z \to 0^{+}$,
\begin{equation}\label{app1:eq:K-expansions}
\begin{aligned}
K_0(z) &= -\log(z/2) - \gamma + O\!\bigl(z^{2}\log z\bigr), \\[0.4em]
K_1(z) &= \frac{1}{z}
+ \frac{z}{2}\!\left(\log(z/2) + \gamma - \frac{1}{2}\right)
+ O\!\bigl(z^{3}\log z\bigr), \\[0.4em]
K_2(z) &= \frac{2}{z^{2}} - \frac{1}{2}
+ O\!\bigl(z^{2}\log z\bigr), \\[0.4em]
K_3(z) &= \frac{8}{z^{3}} - \frac{1}{z} 
+ O\!\bigl(z \log z\bigr).
\end{aligned}
\end{equation}
\end{lemma}

\begin{remark}
By abuse of notation we also use Landau Notation for small arguments in the following sense. For a function $z\mapsto R(z)$ defined for $z\in(0,\infty)$, the notation $R\in O(f)$ as
$z \to 0^{+}$ means that there exist constants $C>0$ and $\delta>0$ such
that
\[
|R(z)| \le C\, |f(z)| \qquad \text{for all } 0 < z < \delta .
\]
\end{remark}
\begin{proof}[Proof of Lemma \ref{app1:lem:Kj_expansion}]
    First we prove the behavior of $K_0$ stated in \eqref{app1:eq:K-expansions}. 
    Defining 
    \begin{equation*}
        s_m(z) =  \frac{(z^{2}/4)^{m}}{(m!)^{2}}\, \Psi(m+1), \quad m \in \N_0
    \end{equation*}
    and noting $s_0(z) = \Psi(1) = -\gamma$ is the Euler--Mascheroni constant we conclude from \eqref{app1:eq:def_bessels}
    \begin{equation}\label{app1:eq:K_0_relevant}
        K_0(z) + \log(z/2) + \gamma = - \log(z/2)\left( I_0(z) -1 \right) + \sum_{m=1}^\infty s_m(z)\, .
    \end{equation}
    Directly from the series expansion in\eqref{app1:eq:def_bessels} we find
    \begin{equation}\label{app1:eq:I_0_small}
        I_0(z) - 1 =  z^2/4 + O(z^4) \, .
    \end{equation}
    Combining \eqref{app1:eq:K_0_relevant} and \eqref{app1:eq:I_0_small} and noting the $s_m \in O(z^2)$ for $m\geq 1$ proves the expansion for $K_0$ in \eqref{app1:eq:K-expansions} immediately.

    Next we prove the behavior of $K_1$ stated in \eqref{app1:eq:K-expansions}. We define
    \begin{equation*}
        a_m(z) = \frac{(z^{2}/4)^{m}}{m!(m+1)!}\, \bigl(\Psi(m+1)+\Psi(m+2)\bigr)
    \end{equation*}
    then $a_0 = \Psi(1) + \Psi(2) = 1-2\gamma$. Using the series expansion of $K_1$ in \eqref{app1:eq:K-expansions} we find
    \begin{equation}\label{app1:eq:K_1_relevant}
        \begin{split}
            K_1(z) - z^{-1} &- \log(z/2)(z/2) - (1-2\gamma)(z/4) \\
            &= \log(z/2)\left( I_1(z)-(z/2) \right) + (z/4)\sum_{m=1}^\infty a_m(z)
        \end{split}  
    \end{equation}
    Using the series expansion of $I_1$ in \eqref{app1:eq:def_bessels} we conclude
    \begin{equation}\label{app1:eq:I_1_small}
        I_1(z)-(z/2) = z^3/16 + O(z^5)\, .
    \end{equation}
    Combining \eqref{app1:eq:I_1_small} and \eqref{app1:eq:K_1_relevant} together with $z\cdot a_m \in O(z^3)$ yields the statement for $K_1$ in \eqref{app1:eq:K-expansions}.

    We use now the recurrence relation \eqref{app1:eq:reoc} to find the near zero extension of $K_2$ and $K_3$. Using $\nu=1$ in \eqref{app1:eq:reoc} gives
    \begin{equation*}
        K_2(z) = K_0(z) + 2z^{-1}K_1(z) \, .
    \end{equation*}
    Note that $z^{-1} K_1(z) = 2/z^2 + O(z^2 \log(z))$ and consequently the relation
    \begin{equation*}
        K_2(z) =  \frac{2}{z^{2}} - \frac{1}{2}+ O\!\bigl(z^{2}\log z\bigr)
    \end{equation*}
    follows immediately. A similar calculation with $\eqref{app1:eq:reoc}$ for $\nu=2$ yields the remaining extension for $K_3$.
\end{proof}

\begin{remark}\label{app:cor:bessel_aux}
    We note the following relations that follow by application of Lemma \ref{app1:lem:Kj_expansion}.
    \begin{align}
                \frac{1}{2}z^2\left(K_0(z)K_2(z)-K_1(z)^2)\right) &= -\log(z)+ O(1). \label{app1:eq:K0K2-K1^2}\\
                \frac{1}{2}z^2\left(K_3(z)K_1(z)-K_2(z)^2 \right) &= \frac{2}{z^2} + \frac{1}{2}\left(4\log(z)+4\gamma-1-2\log(4)\right)+O(z^2\log(z))\label{app1:eq:K1K3-K2^2}\\
                (1-zK_1(z))K_2(z) &=\frac{1}{2}\left(-2\log(z)-2\gamma+1+2\log(2)\right)+O(z^2\log(z)).\label{app1:eq:(1-zK1)K2}\\
                (2-z^2 K_2(z))K_2(z) &= 1+ O(z^2\log(z))\label{app1:eq:(2-z^2K2)K2}
        \end{align}
        Applying the product rule and using the recurrence relations directly shows for any $\nu\geq 1$ and $z>0$
        \begin{equation} \label{app1:lem:Stammfunktionen Bessel}
            sK_\nu(z)^2=\frac{d}{dz}\left(\frac{1}{2}z^2\left(K_\nu(z)^2-K_{\nu-1}(z)K_{\nu+1}(z)\right)\right).
        \end{equation}
\end{remark}

\section{Evaluation of Integrals}\label{app:integrals}
    In this section we provide the evaluation of the integrals involving convolutions of Bessel functions appearing in section \ref{sec:aux_lemmatas}. We will use the following convention of the Fourier transform
    \begin{equation}
        \mathcal{F}(f)(k)=(2\pi)^{-d/2} \int_{\R^d} f(x)e^{-ikx}\, dx.
    \end{equation}
    \begin{lemma}\label{app2:lem:Faltung K0 mit sich selbst}
        Let $\mu>0$, $a,b\in \R^4,a\neq b$. Then
        \begin{equation}\label{app2:eq:Faltung K0 mit sich selbst}
            \int_{\R^4}K_0(\mu\abs{x-a}) K_0(\mu\abs{x-b}) \, dx = 2\pi^2 \frac{\abs{a-b}^2}{6\mu^2}K_2(\mu\abs{a-b}).
        \end{equation}
        \begin{proof}
            For $a\neq b$ and $\mu>0$, the integral on the left--hand side of \eqref{app2:eq:Faltung K0 mit sich selbst} is finite. 
Indeed, $K_0(t)$ decays exponentially as $t\to\infty$ and satisfies $K_0(t)\sim -\log t$ as $t\downarrow0$ (see \eqref{app1:eq:K-expansions}). 
Thus, the integrand has at most logarithmic singularities near $x=a$ and $x=b$, which are integrable in dimension $4$ since $(\log|x|)^2\in L^1_{\mathrm{loc}}(\mathbb R^4)$, while it decays exponentially as $|x|\to\infty$.
\medskip
        
            Define $f_\mu\colon\R^4\to\R, \, f_\mu(x)=K_0(\mu\abs{x})$.
            Let $f\colon\R^d\to\R$ such that $f(x)=f_0(\abs{x})$. Then for $k\in \R^4$
            \begin{equation*}
                \mathcal{F}(f)(k)= \abs{k}^{-d/2+1}\int_0^\infty r^{d/2}  f_0(r) J_{d/2-1}(\abs{k}r)\, dr.
            \end{equation*}
            Since $f_\mu$ is radially symmetric, we see
            \begin{equation*}
                \mathcal{F}(f_\mu)(k) = \abs{k}^{-1}\int_0^\infty r^2 K_0(\mu r) J_1(\abs{k}r)\, dr.
            \end{equation*}
            The integral on the right--hand side is explicitly solved in \cite{gradshteyn2014table}[p.665, 6.521.12] such that
            \begin{equation}
                \mathcal{F}(f_\mu)(k)=\frac{2}{\left(\mu^2+\abs{k}^2\right)^2}.\label{app2:eq:Fouriertransformierte f mu}
            \end{equation}
            Denoting by $\ast$ the usual convolution of functions, yields
            \begin{equation*}
                \int_{\R^4}K_0(\mu\abs{x-a}) K_0(\mu\abs{x-b}) \, dx = \left(f_\mu\ast f_\mu\right)(a-b).
            \end{equation*}
            Therefore, by taking the Fourier Transform, we obtain
            \begin{equation}
                \int_{\R^4}K_0(\mu\abs{x-a}) K_0(\mu\abs{x-b}) \, dx = 4\pi^2 \mathcal{F}^{-1}\left(\mathcal{F}(f_\mu)^2\right)(a-b).\label{app2:eq:Faltung K_0 fouriertransformiert Rechnung}
            \end{equation}
            A direct calculation shows for $k\in\R^4$
            \begin{equation*}
                \mathcal{F}(f_\mu)^2(k) = \frac{4}{\left(\mu^2+\abs{k}^2\right)^4} = \frac{1}{12\mu}\frac{d}{d\mu}\left(\frac{1}{\mu}\frac{d}{d\mu} \mathcal{F}(f_\mu)(k)\right).
            \end{equation*}
            We next justify interchanging the inverse Fourier transform with derivatives in~$\mu$.

Fix $\mu_\ast>0$ and define $I_\ast=[\mu_\ast/2,2\mu_\ast]\subset(0,\infty)$. Set
\begin{equation*}
m(\mu,k):=\mathcal F(f_\mu)(k)=\frac{2}{(\mu^2+|k|^2)^2},
\qquad \mu>0,\ k\in\mathbb R^4.
\end{equation*} 
We have,
\begin{equation*}
\partial_\mu m(\mu,k)=-\frac{8\mu}{(\mu^2+|k|^2)^3},
\end{equation*}
and a second differentiation produces only terms of order $(\mu^2+|k|^2)^{-3}$ and $(\mu^2+|k|^2)^{-4}$. Hence, for $\mu\in I_\ast$, there exists $C_\ast >0$ such that
\begin{equation*}
|\partial_\mu^j m(\mu,k)| \le \frac{C_\ast}{(1+|k|^2)^3}, \qquad j=1,2,
\end{equation*}
and $(1+|k|^2)^{-3}\in L^1(\mathbb R^4)$. Thus $m(\mu,\cdot)$, $\partial_\mu m(\mu,\cdot)$, and $\partial_\mu^2 m(\mu,\cdot)$ are dominated on $I_\ast$ by a common $L^1$-function independent of $\mu$.

Applying dominated convergence
we may interchange $\partial_\mu$ (and $\partial_\mu^2$) with $\mathcal F^{-1}$ for all $\mu\in I_\ast$. Since $\mu_\ast>0$ was arbitrary, this holds for all $\mu>0$. Multiplication by functions bounded on $I_\ast$ (e.g.\ $1/\mu$) does not affect the argument.

It follows
            \begin{equation*}
                4\pi^2 \mathcal{F}^{-1}\left(\mathcal{F}(f_\mu)^2\right)(a-b)= 2\pi^2 \frac{1}{6\mu}\frac{d}{d\mu}\left(\frac{1}{\mu}\frac{d}{d\mu}K_0(\mu\abs{a-b})\right).
            \end{equation*}
            By differentiating and \eqref{app2:eq:Faltung K_0 fouriertransformiert Rechnung} we obtain
            \begin{equation*}
                \int_{\R^4}K_0(\mu\abs{x-a}) K_0(\mu\abs{x-b}) \, dx=2\pi^2\frac{\abs{a-b}^2}{6\mu^2} K_{2}(\mu\abs{a-b}).
            \end{equation*}
        \end{proof}
    \end{lemma}

    \begin{corollary}\label{app2:cor:L2 norm K_0}
        Let $\mu>0$. Then
        \begin{equation*}
            \int_{\R^4}K_0(\mu\abs{x})^2 dx = \frac{2\pi^2}{3\mu^4}.
        \end{equation*}
        \begin{proof}
           Choose $b=0$ and $a=re_1$, we take the limit  $r\downarrow0$. By dominated convergence (using the logarithmic behavior of $K_0$ at $0$ and its exponential decay at $\infty$) and the expansion $K_2(z)\sim 2z^{-2}$ as $z\downarrow0$ from \eqref{app1:eq:K-expansions}, the stated identity follows from Lemma \ref{app2:lem:Faltung K0 mit sich selbst}.

        \end{proof}
    \end{corollary}
\color{black}
    \begin{lemma}\label{app2:lem:Faltung K0 mit Gmu}
        Let $\mu>0$ and $a,b\in\R^4, a\neq b$ and $G_\mu$ as in Lemma \ref{sec5:lem:EigenschaftenGmu}. Then
        \begin{equation*}
            \begin{split}
                \int_{\R^4} K_0(\mu\abs{x-a})\partial_{x_i}G_\mu(x-b)\, dx &=  2\pi^2 \frac{1}{2}\abs{a-b} K_0(\mu\abs{a-b})\frac{b_i-a_i}{\abs{b-a}}\\
                &=\pi^2 K_0(\mu\abs{a-b})(b_i -a_i)
            \end{split}
        \end{equation*}
        \begin{proof}
            Let $f_\mu$ be as in the proof of Lemma \ref{app2:lem:Faltung K0 mit sich selbst}. Then
            \begin{equation*}
                \int_{\R^4} K_0(\mu\abs{x-a})\partial_{x_i}G_\mu(x-b)\, dx = \int_{\R^4} K_0(\mu\abs{x-a})\left(-\partial_{b_i}G_\mu(x-b)\right)\, dx.
            \end{equation*}
            Therefore 
            \begin{equation}\label{app2:eq:Ableitung x zu Ableitung b}
                \int_{\R^4} K_0(\mu\abs{x-a})\partial_{x_i}G_\mu(x-b)\, dx = -\partial_{b_i} \left((f_\mu\ast G_\mu)(a-b)\right).
            \end{equation}
            Due to the fact that $G_\mu$ is the Greens function (see \eqref{sec5:eq:GleichungGmu}) we know for $k\in\R^4$
            \begin{equation}
                \mathcal{F}(G_\mu)(k) = \frac{1}{\mu^2+\abs{k}^2}.\label{app2:eq:Fouriertransformierte G mu}
            \end{equation}
            As in the previous Lemma we have
            \begin{equation}\label{app2:eq:Fourier Faltung fmu Gmu}
                (f_\mu\ast G_\mu)(a-b) = 4\pi^2\mathcal{F}^{-1}\left(\mathcal{F}(f_\mu)\cdot\mathcal{F}(G_\mu)\right)(a-b).
            \end{equation}
            A direct calculation shows
            \begin{equation}\label{app2:eq:Produkt Fourier fmu Gmu}
                \mathcal{F}(f_\mu)\cdot\mathcal{F}(G_\mu) (k) = -\frac{1}{4\mu}\frac{d}{d\mu}\mathcal{F}(f_\mu)(k).
            \end{equation}
            Since $\mathcal F(f_\mu)\mathcal F(G_\mu)(k)=2(\mu^2+|k|^2)^{-3}\in L^1(\mathbb R^4)$ and its $\mu$–derivative is bounded by
            \begin{equation*}
                C(1+|k|^2)^{-3}\in L^1(\mathbb R^4)
            \end{equation*} locally in $\mu$, dominated convergence justifies interchanging $\partial_\mu$ with $\mathcal F^{-1}$.

            Inserting \eqref{app2:eq:Fourier Faltung fmu Gmu} and \eqref{app2:eq:Produkt Fourier fmu Gmu} into \eqref{app2:eq:Ableitung x zu Ableitung b} implies 
            \begin{equation*}
                \int_{\R^4} K_0(\mu\abs{x-a})\partial_{x_i}G_\mu(x-b)\, dx = \partial_{b_i}\left( 2\pi^2 \frac{1}{2\mu} \frac{d}{d\mu} K_0(\mu\abs{a-b})\right).
            \end{equation*}
            Therefore,
            \begin{equation*}
                \int_{\R^4} K_0(\mu\abs{x-a})\partial_{x_i}G_\mu(x-b)\, dx = 2\pi^2\partial_{b_i}\left(-\frac{\abs{a-b}K_1(\mu\abs{a-b})}{2\mu}\right).
            \end{equation*}
            Using the identity $\frac{d}{dz}\big(zK_1(z)\big)=-zK_0(z)$ with $z=\mu|a-b|$ and $\partial_{b_i}|a-b|=\frac{b_i-a_i}{|a-b|}$, we obtain
            \begin{equation*}
                \int_{\R^4} K_0(\mu\abs{x-a})\partial_{x_i}G_\mu(x-b)\, dx = 2\pi^2 \frac{1}{2}\abs{a-b} K_0(\mu\abs{a-b})\frac{b_i-a_i}{\abs{b-a}}.
            \end{equation*}
        \end{proof}
    \end{lemma}

    \begin{corollary}\label{app2:cor:Faltung Gmu mit sich selbst}
        Let $\mu>0$, $a,b\in\R^4, a\neq b$ and $G_\mu$ as in Definition \ref{sec5:lem:EigenschaftenGmu}. Then 
        \begin{equation*}
            \int_{\R^4}  G_\mu(x-a)G_\mu(x-b)\, dx =2\pi^2 K_0(\mu\abs{a-b}).
        \end{equation*}
        \begin{proof}
            As before, we get
            \begin{equation*}
                \int_{\R^4}  G_\mu(x-a)G_\mu(x-b)\, dx = 4\pi^2\mathcal{F}^{-1}\left(\mathcal{F}(G_\mu)^2\right)(a-b).
            \end{equation*}
            Recall \eqref{app2:eq:Fouriertransformierte G mu} and \eqref{app2:eq:Fouriertransformierte f mu} to see for $k\in\R^4$
            \begin{equation*}
                \mathcal{F}(G_\mu)(k)^2 = \frac{1}{2} \mathcal{F}(f_\mu)(k).
            \end{equation*}
            Therefore,
            \begin{equation*}
                \int_{\R^4}  G_\mu(x-a)G_\mu(x-b)\, dx=2\pi^2 f_\mu(a-b) =2\pi^2 K_0(\mu\abs{a-b}).
            \end{equation*}
        \end{proof}
    \end{corollary}

    \begin{lemma}\label{app2:lem:Faltung Ableitung Gmu mit sich selbst}
        Let $\mu>0$, $a,b\in\R^4, a\neq b$ and $G_\mu$ as in Definition \ref{sec5:lem:EigenschaftenGmu}. Then 
        \begin{equation}
            \int_{\R^4} \partial_{x_i} G_\mu(x-a) \partial_{x_i}G_\mu(x-b)\, dx = 2\pi^2 \left(\mu \frac{K_1(\mu\abs{a-b})}{\abs{a-b}}-\mu^2K_2(\mu\abs{a-b})\frac{(a_i-b_i)^2}{\abs{a-b}^2} \right).
        \end{equation}
        \begin{proof}
            By taking the inner derivative into account we have
            \begin{equation*}
                \int_{\R^4} \partial_{x_i} G_\mu(x-a) \partial_{x_i}G_\mu(x-b)\, dx = \int_{\R^4}\partial_{a_i} G_\mu(x-a)\partial_{b_i}G_\mu(x-b)\, dx.
            \end{equation*}
            The interchange of $\partial_{a_i}\partial_{b_i}$ with the integral in \eqref{app2:eq:Faltung Ableitung Gmu mit sich selbst Zwischenschritt} is justified by differentiation under the integral sign (dominated convergence), using that $a\neq b$ separates the singularities of $G_\mu$ and that $G_\mu$ (and its derivatives) decay exponentially at infinity.

            Thus, we arrive at
            \begin{equation}
                \int_{\R^4} \partial_{x_i} G_\mu(x-a) \partial_{x_i}G_\mu(x-b)\, dx = \partial_{a_i}\partial_{b_i} \int_{\R^4}  G_\mu(x-a)G_\mu(x-b)\, dx.\label{app2:eq:Faltung Ableitung Gmu mit sich selbst Zwischenschritt}
            \end{equation}
            Due to Corollary \ref{app2:cor:Faltung Gmu mit sich selbst} we find
            \begin{equation*}
                 \int_{\R^4} \partial_{x_i} G_\mu(x-a) \partial_{x_i}G_\mu(x-b)\, dx = 2\pi^2\partial_{a_i}\partial_{b_i} K_0(\mu\abs{a-b}).
            \end{equation*}
            It follows by taking the derivative
            \begin{equation*}
                \int_{\R^4} \partial_{x_i} G_\mu(x-a) \partial_{x_i}G_\mu(x-b)\, dx = 2\pi^2 \partial_{a_i}\left(\mu\frac{K_1(\mu\abs{a-b})}{\abs{a-b}}(a_i-b_i)\right)
            \end{equation*}
            and we therefore arrive at
            \begin{equation*}
                \int_{\R^4} \partial_{x_i} G_\mu(x-a) \partial_{x_i}G_\mu(x-b)\, dx = 2\pi^2 \left(-\mu^2K_2(\mu\abs{a-b})\frac{(a_i-b_i)^2}{\abs{a-b}^2}+\mu\frac{K_1(\mu\abs{a-b})}{\abs{a-b}} \right).
            \end{equation*}
        \end{proof}
    \end{lemma}
\section{On the Zero--Energy Resonance in $\R^4$}\label{app:D}
\begin{lemma}
\label{app:lem:Resonanz}
Let $V \in L^2{(\R^4)}$ be real--valued, radially symmetric, and compactly supported in
$B_{\rho_0}(0)=\{x\in \R^4: \abs{x} < \rho_0\}$. Assume that the operator
\[
h=P_x^2+V
\]
has a virtual level at zero in the sense of Definition~\ref{sec1:def:virtual_lvl}.
Then there exists a (up to normalization) unique function
$\varphi_0\in\dot H^1(\R^4)$ solving $h\varphi_0=0$ in the sense of distributions, which can be chosen real-valued and positive.
Moreover, there exists a constant $c_0\in\R$ such that
\begin{equation}
\label{sec5:eq:exact_tail:copy}
\varphi_0(x)=c_0 |x|^{-2}
\qquad\text{for all } |x|>\rho_0 .
\end{equation}
\end{lemma}
\begin{proof}
By assumption, for every $\varepsilon>0$ there exists $\varphi_\varepsilon\in \dot H^1(\mathbb R^4)$ with
\begin{equation*}
\langle \varphi_\varepsilon,(P_x^2+(1+\varepsilon)V)\varphi_\varepsilon\rangle < 0 .
\end{equation*}
We normalize
\begin{equation*}
\|P_x\varphi_\varepsilon\|_{2}=1.
\end{equation*}
Then $(\varphi_\varepsilon)$ is bounded in $\dot H^1(\mathbb R^4)$, and hence there exists a subsequence
denoted by the same letter and a function $\varphi_0\in \dot H^1(\mathbb R^4)$ such that
\begin{equation*}
P_x\varphi_\varepsilon \rightharpoonup P_x\varphi_0
\quad \text{weakly in } L^2(\mathbb R^4).
\end{equation*}

Since $h(0)\ge 0$, we have
\begin{equation*}
0\le \langle \varphi_\varepsilon,(P_x^2+V)\varphi_\varepsilon\rangle
= \|P_x\varphi_\varepsilon\|_2^2+\langle \varphi_\varepsilon,V\varphi_\varepsilon\rangle
=1+\langle \varphi_\varepsilon,V\varphi_\varepsilon\rangle .
\end{equation*}
On the other hand,
\begin{equation*}
\langle \varphi_\varepsilon,(P_x^2+(1+\varepsilon)V)\varphi_\varepsilon\rangle
=1+(1+\varepsilon)\langle \varphi_\varepsilon,V\varphi_\varepsilon\rangle \le 0 .
\end{equation*}
Consequently,
\begin{equation*}
-1 \le \langle \varphi_\varepsilon,V\varphi_\varepsilon\rangle
\le -\frac{1}{1+\varepsilon},
\end{equation*}
and hence
\begin{equation*}
\langle \varphi_\varepsilon,V\varphi_\varepsilon\rangle \to -1
\quad \text{as } \varepsilon\downarrow 0 .
\end{equation*}

Let $K\subset\mathbb R^4$ be compact with smooth boundary and $\supp(V)\subset K$.
By the Sobolev inequality in dimension four, $(\varphi_\varepsilon)$ is bounded in
$L^4(\mathbb R^4)$ and thus in $L^4(K)$.
Together with $\|P_x\varphi_\varepsilon\|_2=1$, this implies boundedness in $H^1(K)$.
By the Rellich--Kondrachov theorem, after passing to a subsequence,
\begin{equation*}
\varphi_\varepsilon \to \varphi_0
\quad \text{strongly in } L^2(K).
\end{equation*}
Since $V\in L^2_{\mathrm{loc}}(\mathbb R^4)$ and $(\varphi_\varepsilon)$ is bounded in $L^4(K)$,
it follows that
\begin{equation*}
\langle \varphi_\varepsilon,V\varphi_\varepsilon\rangle
\to
\langle \varphi_0,V\varphi_0\rangle .
\end{equation*}
Hence
\begin{equation*}
\langle \varphi_0,V\varphi_0\rangle=-1 ,
\end{equation*}
and in particular $\varphi_0\not\equiv 0$.

By weak lower semicontinuity of the norm,
\begin{equation*}
\|P_x\varphi_0\|_2
\le \liminf_{\varepsilon\to 0}\|P_x\varphi_\varepsilon\|_2
=1.
\end{equation*}
Since $h(0)\ge 0$ extends from $H^1$ to $\dot H^1$ by density, we obtain
\begin{equation*}
0 \le \|P_x\varphi_0\|_2^2+\langle \varphi_0,V\varphi_0\rangle
\le 1-1=0.
\end{equation*}
Therefore,
\begin{equation*}
\|P_x\varphi_0\|_2^2+\langle \varphi_0,V\varphi_0\rangle=0 .
\end{equation*}
Outside the support of $V$, the function $\varphi_0$ satisfies
\begin{equation*}
P^2_x\varphi_0 = 0.
\end{equation*}
By radial symmetry, the general solution in $|x|>\rho_0$ is
\begin{equation*}
\varphi_0(x)=c_1 + c_0|x|^{-2}.
\end{equation*}
Since $\varphi_0\in \dot H^1(\mathbb R^4)$, the constant term must vanish,
that is $c_1=0$, and thus
\begin{equation*}
\varphi_0(x)=c_0|x|^{-2}, \qquad |x|>\rho_0.
\end{equation*}

Since $h$ has real coefficients, $\overline{\varphi_0}$ is also a solution of
$h\varphi_0=0$, and we may assume that $\varphi_0$ is real--valued.
Moreover, the inequality $|\nabla|\varphi_0||\le|\nabla\varphi_0|$ a.e.\ implies
$|\varphi_0|\in\dot H^1(\R^4)$ and
\[
\|\nabla|\varphi_0|\|_2^2+\langle|\varphi_0|,V|\varphi_0|\rangle
\le
\|\nabla\varphi_0\|_2^2+\langle\varphi_0,V\varphi_0\rangle=0 .
\]
Since $h\ge0$ as a quadratic form, equality holds and $|\varphi_0|$ is again a
zero--energy solution. Replacing $\varphi_0$ by $|\varphi_0|$, we may therefore
assume $\varphi_0\ge0$.

Let $\psi\in \dot H^1(\R^4)$ be any real--valued distributional solution of $h\psi=0$. Choosing $\lambda:=\psi(x_0)/\varphi_0(x_0)$ for any $|x_0|>\rho_0$,
the difference $w:=\psi-\lambda\varphi_0$ satisfies $hw=0$ and vanishes at $x_0$.
By unique continuation for Schr\"odinger operators with
$V\in L^{2}_{\mathrm{loc}}(\R^4)$ (see \cite[Theorem 6.3]{JK:1985} and Remark after \cite[Theorem D.3.2]{thesis:S:2025}), it follows that $w\equiv0$ on $\R^4$.
Hence $\psi=\lambda\varphi_0$, proving uniqueness up to normalization.
\end{proof}
\begin{remark} \label{app:remarkc2}
    In Lemma~\ref{app:lem:Resonanz} we have used that the potential is radially symmetric and compactly supported. In Lemma~\ref{sec5:lem:AbfallverhaltenResonanz} we have stated a similar result for short-range potential at the cost of additional contribution to the resonance function that decay faster. The proof of Lemma~\ref{sec5:lem:AbfallverhaltenResonanz} is an extension of the proof provided above and can be found in \cite[Chapter 3.2 and 3.3]{thesis:S:2025} the subsequent remarks together with \cite[Appendix D]{thesis:S:2025}. One can also follow the original analysis on virtual levels by Yafaev in \cite{Y:1975} where Hölder continuity of the resolvent operator is proven (in dimension $d=3$) to find such decay properties also in dimension $d=4$.
\end{remark}
For convenience and since we could not easily find it in the literature, we give now a straight forward example of a short--range potential that exhibits a virtual level. Such examples are well known and a standard consequence, mainly by application of the Birman--Schwinger principle.

\medskip
\textbf{Example of a critical Potential: } Let $\lambda>0$ and define for $n\geq 3$ the operator
\begin{equation*}
    \hat{h}_\lambda = -\Delta_x - \lambda 1_{\{{\abs{x}<1\}} }\quad \text{ on } L^2(\R^n)
\end{equation*}
\begin{corollary} \label{cor:lambda_crit}
  Then there exists $\lambda_c>0$ such that for any $\lambda < \lambda_c$ it is  $\sigma_{\disc}(\hat{h}_\lambda) = \emptyset $ and for $\lambda>\lambda_c$ one has $\sigma_{\disc}(\hat{h}_\lambda) \neq \emptyset $.   
\end{corollary}
\begin{remark}
    Corollary \ref{cor:lambda_crit} does proof that the operator $\hat{h}_{\lambda_c}$ has a virtual level in the sense of Definition \ref{sec1:def:virtual_lvl}. 
\end{remark}
\begin{proof} Due to Weyl's theorem the essential spectrum of $\hat{h}_\lambda$ is $[0,\infty)$ for any $\lambda>0$.
    Due to the Cwikel--Lieb--Rozenblum bound (see e.g. \cite[Theorem XIII.12]{book:RS:1978}) the number of negative eigenvalues $N(\lambda)$ fulfills for $n\geq 2$ 
    \begin{equation*}
        N(\lambda) \leq c_n \lambda^{n/2} \int 1_{\{{\abs{x}<1\}}}(x) dx \lesssim \abs{\lambda}^{n/2}
    \end{equation*}
    for a suitable $c_n>0$. Consequently, there exists $\lambda_0>0$ such that for any $\lambda<\lambda_0$, there are no negative eigenvalues. Moreover, there exists $\lambda_1>0$ such that $\hat{h}_\lambda$ has at least one negative eigenvalue for any $\lambda>\lambda_1$. To show that we construct a test function $\varphi \in L^2$ such that
    \begin{equation*}
       \langle \varphi, \hat{h}_\lambda \varphi \rangle < 0.
    \end{equation*}
    Take, for example any $\varphi \in C_c^\infty$ with $\varphi \equiv 1$ on $\{x\in \R^n: \abs{x}<1\}$ then
    \begin{equation*}
        \lambda \int_{\R^n}1_{\{{\abs{x}<1\}}}(x)  \abs{\varphi(x)}^2 dx = \lambda \abs{B_1(0)}.
    \end{equation*}
    For $\lambda>0$ large enough
    \begin{equation*}
        \langle \varphi, \hat{h}_\lambda \varphi \rangle= \norm{ \nabla \varphi}_2^2  - \lambda \abs{B_1(0)} <0.
    \end{equation*}
    We define the critical coupling $0<\lambda_c<\lambda_1$ as
    \begin{equation*}
        \lambda_c = \inf \{ \lambda>0: N(\lambda) \geq 1 \}.
    \end{equation*}
    Define the ground state energy $E(\lambda) = \inf \sigma(\hat{h}_\lambda)$. Following \cite[Lemma before Theorem XIII.10]{book:RS:1978} we conclude that the spectral measure of $\hat{h}_\lambda$ is continuous in $\lambda \in \R$. Consequently, $E$ has to vanish for some point in the interval $[\lambda_0,\lambda_1)$. Furthermore, $E$ is monotone decreasing for $\lambda\in [0,\infty)$ and strictly monotone once $E$ is negative. It follows that $\lambda_c>0$ exists and is unique.
\end{proof}
\begin{remark}
    The example above uses the Cwikel--Lieb-Rozenblum bound which does not hold in dimension $n=2$. In dimension $n\leq 2$ any shallow potential well does produce negative eigenvalues (see, e.g., \cite{S:1976}). Consequently, the vanishing potential $V\equiv 0$ as a virtual level at zero in dimension $n\leq 2$. 

    Another class of non--trivial critical potentials that are not short-range (meaning they decay as $\abs{x}^{-2}$ at infinity) that have repulsive and attractive regimes are described in \cite{HJL:2023} for any dimension $d\geq 1$.  
\end{remark}

\smallskip\noindent 
%%%%%%%% BIB TEX %%%%%%%%%%%%%%%%%%%%%
\bibliographystyle{plain}
\bibliography{references}
\end{document}

%% file: figure_001.tex
  \begin{tikzpicture}[scale=1, >=Stealth]
    
    % Positions of particles
    \coordinate (r1) at (0,0);
    \coordinate (r2) at (5,1);
    \coordinate (r3) at (1,-3);
    \coordinate (r4) at (6,-3);

    % Particles Beschriftung
    \fill (r1) circle (4pt) node[above left] {$\mathbf 1$};
    \fill (r2) circle (4pt) node[above right] {$\mathbf 2$};
    \fill (r3) circle (4pt) node[below left] {$\mathbf 3$};
    \fill (r4) circle (4pt) node[below right] {$\mathbf 4$};

    % Massenschwerpunkte
    \coordinate (R12) at ($1/2 *(r1)+1/2*(r2)$);
    \coordinate (R34) at ($1/2 *(r3)+1/2*(r4)$);

    \coordinate (R123) at ($1/3*(r1)+1/3*(r2)+1/3*(r3)$);
    \coordinate (R124) at ($1/3*(r1)+1/3*(r2)+1/3*(r4)$);
    \coordinate (R134) at ($1/3*(r1)+1/3*(r3)+1/3*(r4)$);
    \coordinate (R234) at ($1/3*(r2)+1/3*(r3)+1/3*(r4)$);

    \coordinate (R) at ($1/2*(R12)+1/2 *(R34)$);

    % Massenschwerpunkte Beschriftung
    \fill (R123) circle (2pt) node[below] {$\mathbf x_{123}$};
    \fill (R124) circle (2pt) node[below] {$\mathbf x_{124}$};
    \fill (R134) circle (2pt) node[below] {$\mathbf x_{134}$};
    \fill (R234) circle (2pt) node[below] {$\mathbf x_{234}$};
    \fill (R) circle (2pt) node[below left] {$\mathbf \vartheta$};

    % Center of mass Koordinaten
    \draw[->, thick, blue]
        (r1) -- (r2)
        node[midway, above] {$\boldsymbol{q_{12}} $};
    \draw[->, thick, cyan]
        (r3) -- (r4)
        node[midway, below left] {$\boldsymbol{q_{34}} $};
    \draw[->, thick, red]
        (R12) -- (R34)
        node[below, above right] {$\boldsymbol{\xi_{12,34}} $};

    % Ergänzend etas
    \draw[->, thin, orange]
        (R12) -- (r3)
        node[pos=0.66, left] {$\boldsymbol{\eta_{12}^-} $};
    \draw[->, thin, orange]
        (R12) -- (r4)
        node[pos=0.66, right] {$\boldsymbol{\eta_{12}^+} $};
    \draw[->, thin, orange]
        (R34) -- (r1)
        node[pos=0.66, left] {$\boldsymbol{\eta_{34}^-} $};
    \draw[->, thin, orange]
        (R34) -- (r2)
        node[pos=0.66, right] {$\boldsymbol{\eta_{34}^+} $};

\end{tikzpicture}

%% file: figure_002.tex
\begin{tikzpicture}[scale=1, >=Stealth]
        % Positions of particles
    \coordinate (r1) at (0,0);
    \coordinate (r2) at (5,1);
    \coordinate (r3) at (1,-3);
%    \coordinate (r4) at (6,-3);

    % Particles Beschriftung
    \fill (r1) circle (4pt) node[above left] {$\mathbf 1$};
    \fill (r2) circle (4pt) node[above right] {$\mathbf 2$};
    \fill (r3) circle (4pt) node[below left] {$\mathbf 3$};
%    \fill (r4) circle (4pt) node[below right] {$\mathbf 4$};

    % Massenschwerpunkte
    \coordinate (R12) at ($1/2 *(r1)+1/2*(r2)$);
%    \coordinate (R34) at ($1/2 *(r3)+1/2*(r4)$);

    \coordinate (R123) at ($1/3*(r1)+1/3*(r2)+1/3*(r3)$);
%    \coordinate (R124) at ($1/3*(r1)+1/3*(r2)+1/3*(r4)$);
%    \coordinate (R134) at ($1/3*(r1)+1/3*(r3)+1/3*(r4)$);
%    \coordinate (R234) at ($1/3*(r2)+1/3*(r3)+1/3*(r4)$);

    \coordinate (R) at ($1/2*(R12)+1/2 *(R34)$);

    % Massenschwerpunkte Beschriftung
    \fill (R123) circle (2pt) node[below] {$\mathbf x_{123}$};
%    \fill (R124) circle (2pt) node[below] {$\mathbf x_{124}$};
%    \fill (R134) circle (2pt) node[below] {$\mathbf x_{134}$};
%    \fill (R234) circle (2pt) node[below] {$\mathbf x_{234}$};
%    \fill (R) circle (2pt) node[below left] {$\mathbf R$};

    % Center of mass Koordinaten
    \draw[->, thick, blue]
        (r1) -- (r2)
        node[midway, above] {$\boldsymbol{q_{12}} $};
    % \draw[->, thick, purple]
    %     (r3) -- (r4)
    %     node[midway, below left] {$\boldsymbol{q_{34}} $};
    % \draw[->, thick, blue]
    %     (R12) -- (R34)
    %     node[below, above right] {$\boldsymbol{\xi_{12,34}} $};

    % Ergänzend etas
    \draw[->, thin, orange]
        (R12) -- (r3)
        node[pos=0.66, left] {$\boldsymbol{\eta_{12}^-} $};
    % \draw[->, thin, orange]
    %     (R12) -- (r4)
    %     node[pos=0.66, right] {$\boldsymbol{\eta_{12}^+} $};
    % \draw[->, thin, orange]
    %     (R34) -- (r1)
    %     node[pos=0.66, left] {$\boldsymbol{\eta_{34}^-} $};
    % \draw[->, thin, orange]
    %     (R34) -- (r2)
    %     node[pos=0.66, right] {$\boldsymbol{\eta_{34}^+} $};
\end{tikzpicture}

%% file: figure_003.tex
   \begin{tikzpicture}[scale=0.6]

% Parameter: Länge des Vektors L 
\def\L{12} 

% Radien
\def\rInner{2}
\def\rOuter{4}

% Zentren bei ± L/2
\coordinate (C1) at ({\L/2},0);
\coordinate (C2) at ({-\L/2},0);

% Automatische Rahmengröße
\pgfmathsetmacro{\xmin}{-\L/2 - \rOuter - 2}
\pgfmathsetmacro{\xmax}{ \L/2 + \rOuter + 2}
\pgfmathsetmacro{\ymin}{-\rOuter - 1}
\pgfmathsetmacro{\ymax}{ \rOuter + 2}

--- Annuli (zuerst zeichnen) ---
\fill[blue,
    fill opacity=0.12]
  (C1) circle (\rOuter)
  (C1) circle (\rInner);

\fill[blue,
    fill opacity=0.12]
  (C2) circle (\rOuter)
  (C2) circle (\rInner);

\fill[white] (C1) circle (\rInner);
\fill[white] (C2) circle (\rInner);

 % --- Innere Bälle ---
 \fill[magenta, fill opacity=0.25] (C1) circle (\rInner);
 \fill[magenta, fill opacity=0.25] (C2) circle (\rInner);

% --- Randlinien ---
\draw[thick] (C1) circle (\rInner);
\draw[thick] (C1) circle (\rOuter);

\draw[thick] (C2) circle (\rInner);
\draw[thick] (C2) circle (\rOuter);

% --- x-Achse ---
%\draw[->] (\xmin+0.5,0) -- (\xmax-1,0) node[right] {$\ell$};

% --- Nullpunkt ---
\draw plot[only marks, mark=x, mark size=5pt,
  line width=0.8pt] coordinates {(0,0)};
\node[below] at (0,0) {$x=0$};

% --- Markierung der Zentren ---
\fill (C1) circle (2pt);
\fill (C2) circle (2pt);

% Beschriftung der Zentren
\node[below] at (C1) {$+\ell/2$};
\node[below] at (C2) {$-\ell/2$};

% --- Beschriftungen der inneren Bälle (nach unten verschoben) ---
\node at ($(C1)+(0,1)$) {$B_\rho(+\ell/2)$};
\node at ($(C2)+(0,1)$) {$B_\rho(-\ell/2)$};

% --- Beschriftungen der Annuli ---
\node at ($(C1)+(0,2.8)$) {$A_\rho(+\ell/2)$};
\node at ($(C2)+(0,2.8)$) {$A_\rho(-\ell/2)$};

% --- Beschriftung des äußeren Gebietes ---
\node[font=\large] at (0,2) {$\Omega_{2\rho}(\ell/2)$};

\end{tikzpicture}

%% file: figure_004.tex
\begin{tikzpicture}[
    >=Latex,
    scale=1.0,
    axis/.style={thick},
    axisarrow/.style={->, thick},
    dashedline/.style={dashed, thick},
]

% -------------------------------------------------
% Parameter
% -------------------------------------------------
\def\xmin{-0.5}
\def\xmax{13.5}

\def\xbreakL{2.2}
\def\xbreakR{2.8}

\def\s{1.2}
\def\L{8}
\def\suppV{1.6}
\pgfmathsetmacro{\PI}{3.141592653589793}

% -------------------------------------------------
% Platzhalter-Funktionen
% -------------------------------------------------
\pgfmathdeclarefunction{Gammafun}{1}{%
    \pgfmathparse{9*(1/2*(4/sqrt(3)-1/2)+1/2)/(#1-8)^2}
}

\pgfmathdeclarefunction{fleft}{1}{%
    \pgfmathparse{3/2*(4/sqrt(3)-1/2)/(\L-#1)+1/2}
}

\pgfmathdeclarefunction{resunknown}{1}{%
    \pgfmathparse{4 - (#1-\L-0.5)*(#1-\L+0.5)*cos(4*\PI*(#1-(\L-0.5)) r)}
}

\pgfmathdeclarefunction{resknownl}{1}{%
    \pgfmathparse{4/sqrt(2)/sqrt(8-#1)}
}

\pgfmathdeclarefunction{resknownr}{1}{%
    \pgfmathparse{4/sqrt(2)/sqrt(#1-8)}
}

% \pgfmathdeclarefunction{fright}{1}{%
%     \pgfmathparse{3/2*(4/sqrt(3)-1/2)/(#1-\L)+1/2}
% }

\pgfmathdeclarefunction{fright}{1}{%
    \pgfmathparse{2*sqrt(3)/(#1-8)}
}

\pgfmathdeclarefunction{Greensright}{1}{%
    \pgfmathparse{2/sqrt(3)*exp(3)*exp(8-#1)}
}

% -------------------------------------------------
% x-Achse mit Unterbrechung
% -------------------------------------------------

% linker Teil (OHNE Pfeil)
\draw[axis] (\xmin,0) -- (\xbreakL,0);

% Unterbrechung //
\draw[thick] (\xbreakL+0.05,-0.12) -- (\xbreakL+0.25,0.12);
\draw[thick] (\xbreakL+0.25,-0.12) -- (\xbreakL+0.45,0.12);

% rechter Teil (MIT Pfeil)
\draw[axisarrow] (\xbreakR,0) -- (\xmax,0) node[right] {$x\parallel\ell$};

% -------------------------------------------------
% Achsenmarken (Ticks)
% -------------------------------------------------

% Tick bei x = 0
\draw[thick] (0,-0.12) -- (0,0.12);
\node[below] at (0,-0.12) {$x=0$};

% Tick bei x = \ell
\draw[thick] (\L,-0.12) -- (\L,0.12);
\node[below] at (\L,-0.12) {$x=\ell/2$};

% -------------------------------------------------
% y-Achse
% -------------------------------------------------
\draw[axisarrow] (0,-0.2) -- (0,5) node[above] {$\phi(\ell,\cdot)$};

% -------------------------------------------------
% Vertikale gestrichelte Linien
% -------------------------------------------------
\foreach \x in {\L+0.5,\L-0.5,\L-1.5,\L+1.5,\L-3,\L+3} {
    \draw[dashedline] (\x,0) -- (\x,4.5);
}

% -------------------------------------------------
% Abstandsmarkierungen
% -------------------------------------------------

\draw[
    <->,
    thick
]
(\L-1.5,0.4) -- (\L,0.4);
\node[black] at (\L-1,0.6) {$\rho$};

\draw[
    <->,
    thick
]
(\L-3,0.2) -- (\L,0.2);
\node[black] at (\L-2,0.4) {$2\rho$};

\draw[
    <->,
    thick
]
(\L-0.5,0.6) -- (\L,0.6);
\node[black] at (\L-0.25,0.8) {$\rho_0$};

% -------------------------------------------------
% Gamma-Kurve (unterbrochen)
% -------------------------------------------------

% linker Teil
\draw[thick,orange!80!black,domain=0:\xbreakL,smooth]
    plot (\x,{Gammafun(\x)});

% rechter Teil
\draw[thick,orange!80!black,domain=\xbreakR:\L-3,smooth]
    plot (\x,{Gammafun(\x)});

\node[orange!80!black] at (\xbreakL,0.9)
    {$\Gamma$};

% -------------------------------------------------
% Übergang links
% -------------------------------------------------

\fill[
    blue,
    fill opacity=0.12
]
(\L-3,0) rectangle (\L-1.5,4.5);

\draw[thick,blue,domain=\L-3:\L-1.5,smooth]
    plot (\x,{fleft(\x)});

\node[blue] at (\L-2.3,2.5) {$\Gamma+f$};

% -------------------------------------------------
% Mittlerer Bereich
% -------------------------------------------------
% \draw[thick,blue,domain=\L-0.5:\L+0.5,smooth]
%     plot (\x,{fmid(\x)});

\fill[
    magenta,
    fill opacity=0.25
]
(\L-1.5,0) rectangle (\L+1.5,4.5);

\draw[
    thick,
    magenta,
    domain=\L-0.5:\L+0.5,
    smooth,
    samples=200
]
plot (\x,{resunknown(\x)});

\draw[
    thick,
    magenta,
    domain=\L-1.5:\L-0.5,
    smooth,
    samples=200
]
plot (\x,{resknownl(\x)});

\draw[
    thick,
    magenta,
    domain=\L+0.5:\L+1.5,
    smooth,
    samples=200
]
plot (\x,{resknownr(\x)});

\node[magenta] at (\L,5) {$\varphi_0(\cdot-\ell/2)$};

% -------------------------------------------------
% Übergang rechts
% -------------------------------------------------
\draw[thick,blue,domain=\L+1.5:\L+3,smooth]
    plot (\x,{fright(\x)});

\fill[
    blue,
    fill opacity=0.12
]
(\L+1.5,0) rectangle (\L+3,4.5);

% -------------------------------------------------
% Fortsetzung rechts
% -------------------------------------------------
\draw[thick,orange!80!black,domain=\L+3:\L+5,smooth]
    plot (\x,{Greensright(\x)});

\end{tikzpicture}